\documentclass[mnsc, nonblindrev]{informs3_hide}

\usepackage[english]{babel}
\usepackage[autostyle, english = american]{csquotes}
\MakeOuterQuote{"}

\usepackage{bbm}
\usepackage{hyperref}
\usepackage{cleveref}
\crefname{subsection}{subsection}{subsections}
\usepackage[normalem]{ulem}
\usepackage[ruled]{algorithm2e} 

\usepackage{algorithmic}
\usepackage{tikz}
\usepackage{svg}
\usepackage{subcaption}
\usepackage{thm-restate}
\usepackage{tcolorbox}

\usepackage{mathtools}
\usepackage{enumerate}

\tikzstyle{vertex}=[circle, draw, fill, inner sep=0pt, minimum size=5pt]
\newcommand{\vertex}{\node[vertex]}

\SetAlFnt{\small}
\SetAlCapFnt{\small}
\SetAlCapNameFnt{\small}
\SetAlCapHSkip{0pt}
\IncMargin{-\parindent}

\mathcode`l="8000
\begingroup
\makeatletter
\lccode`\~=`\l
\DeclareMathSymbol{\lsb@l}{\mathalpha}{letters}{`l}
\lowercase{\gdef~{\ifnum\the\mathgroup=\m@ne \ell \else \lsb@l \fi}}%
\endgroup

\newcommand{\bE}{\mathbb{E}}

\newcommand{\hG}{\widehat{G}}
\newcommand{\cD}{\mathcal{D}}
\newcommand{\hcD}{\widehat{\mathcal{D}}}
\newcommand{\hcG}{\widehat{\mathcal{G}}}
\newcommand{\ALG}{\mathsf{ALG}}
\newcommand{\OPT}{\mathsf{OPT}}
\newcommand{\ADVICE}{\mathsf{ADVICE}}
\newcommand{\vx}{\mathbf{x}}
\newcommand{\yv}{\mathbf{y}}

\newcommand{\lambdav}{\boldsymbol\lambda}
\newcommand{\thetav}{\boldsymbol\theta}
\newcommand{\gammav}{\boldsymbol\gamma}
\newcommand{\alphav}{\boldsymbol\alpha}
\newcommand{\betav}{\boldsymbol\beta}
\newcommand{\abs}[1]{\left\lvert{#1}\right\rvert}

\newcommand{\fl}{f_{\mathsf{L}}}
\newcommand{\fu}{f_{\mathsf{U}}}
\newcommand{\LB}{\mathsf{LB}}
\newcommand{\UB}{\mathsf{UB}}

\usepackage{color}              

\TheoremsNumberedThrough     
\ECRepeatTheorems

\EquationsNumberedThrough    

\usepackage{natbib,bbm,multirow,multicol}
 \bibpunct[, ]{(}{)}{,}{a}{}{,}%
 \def\bibfont{\small}%
\begin{document}




\TITLE{Online Bipartite Matching with Advice: Tight Robustness-Consistency Tradeoffs for the Two-Stage Model}

\ARTICLEAUTHORS{
\AUTHOR{Billy Jin}
\AFF{Booth School of Business, University of Chicago, \EMAIL{billy.jin@chicagobooth.edu}}
\AUTHOR{Will Ma}
\AFF{Graduate School of Business and Data Science Institute, Columbia University, \EMAIL{wm2428@gsb.columbia.edu}}
}

\ABSTRACT{

Two-stage bipartite matching is a fundamental problem of optimization under uncertainty introduced by \citet{feng2021OLD}, who study it under the stochastic and adversarial paradigms of uncertainty.
We propose a method to interpolate between these paradigms, using the Algorithms with Predictions (ALPS) framework.
To elaborate, given some form of information (e.g.\ a distributional prediction) about the uncertainty, 
we consider the optimal decision assuming that information is correct to be some "advice", whose accuracy is unknown.
In the ALPS framework, we define \emph{consistency} to be an algorithm's performance relative to the advice, and \emph{robustness} to be an algorithm's performance relative to the hindsight-optimal decision.
We characterize the tight tradeoff between consistency and robustness for four settings of two-stage matching: unweighted, vertex-weighted, edge-weighted, and fractional budgeted allocation.
Additionally, we show our algorithm achieves state-of-the-art performance in both synthetic and real-data simulations.

}



\maketitle

\section{Introduction}

Life is full of uncertainty. Decisions often need to be made before all uncertainties are resolved, which is why various approaches have been proposed for modeling and handling uncertainty.
Often, the uncertainty is assumed to be stochastic following some known probability distribution, but in practice the estimate of the distribution may or may not be accurate. At other times, the uncertainty is assumed to be adversarial, but this can be too pessimistic for real life problems.

Indeed, reality often lies between the stochastic and adversarial extremes. One might have access to some information about the uncertainty, but this information can be inaccurate and its accuracy may be unknown.
There is extensive literature on intermediate frameworks that interpolate between these extremes, and
in this paper we consider the Algorithms with Predictions (ALPS) framework.  ALPS originated from the theoretical computer science literature, in which it has a burgeoning body of work.\footnote{See \href{https://algorithms-with-predictions.github.io/}{https://algorithms-with-predictions.github.io/} for a list of papers in this area.}
In this framework, the algorithm is given 
some information about the problem (such as a point prediction, distributional prediction, or advice on what to do), whose accuracy is unknown. In ALPS, an algorithm is typically evaluated under the following two measures:
\begin{enumerate}
\item \textbf{Consistency}: its performance when the given information is perfectly correct;
\item \textbf{Robustness}: its worst-case performance when the information can be arbitrarily wrong.
\end{enumerate}

In early work in this area \citep[e.g.][]{lykouris2021competitive,mitzenmacher2020algorithms}, authors have focused on problems and predictions with which an algorithm can have optimal consistency and robustness.  That is, it gets the "best of both worlds"---performance that stochastic optimization would achieve if the
prediction was correct, and performance that robust optimization would achieve if the algorithm ignored all predictions and focused fully on hedging against an adversarial realization of the uncertainty.

We instead use ALPS as an optimization framework to trade off between consistency and robustness.
Indeed, for most problems it is impossible\footnote{Even in the aforementioned work \citep{lykouris2021competitive,mitzenmacher2020algorithms}, the simultaneous optimality of consistency and robustness is only approximate and ignores lower-order terms.} to get the "best of both worlds" --- 
a high level of consistency requires trusting the given information more, which will hurt the algorithm's worst-case performance when the information is wrong. 
For problems like these, controlling the tradeoff between consistency and robustness leads to a \textit{parametrized family of algorithms}, that reveals how to regularize decisions in a problem-specific way.

We demonstrate this parametrized, problem-specific regularization for two-stage bipartite matching, a fundamental problem of optimization under uncertainty (see \Cref{subsec:introModel}) that was introduced by \citet{feng2021two}.
We show how to make use of given information while maintaining any desired level of regularization, in a way that is specific to the bipartite matching problem and is not a black-box combination of the stochastic and robust optimization solutions (see \Cref{subsec:introAlg}).
Importantly, our algorithm achieves the \textit{tight} robustness-consistency tradeoff, i.e.\ for any desired guarantee on robustness, it achieves the maximum possible consistency.

To our knowledge, Pareto-efficiency results of this form are rare in the ALPS literature (see \Cref{sec:relatedWork}); yet, they are necessary for establishing an ALPS algorithm to be best-possible.
In this sense, we see our contribution as showing how to regularize decisions for two-stage bipartite matching in a way that is principled and rooted in theory, at least according to the ALPS framework.
Moreover, our algorithm has stellar numerical performance
in a data-driven optimization setting where the information consists of a small number of samples from an unknown second-stage distribution (see \Cref{subsec:introNum}).
Hence we show our algorithm and the ALPS framework to have applications beyond the specific objectives of maximizing consistency and robustness.

\subsection{Problem Definition} \label{subsec:introModel}

The online bipartite matching model of \citet{feng2021two} differs from the traditional one of \citet{karp1990optimal} in that there are only two stages of arrivals, but each stage is a "batch" of multiple arrivals at the same time.
This is motivated by ride-hailing platforms where it has been reported that delaying matching decisions, until a large batch of requests has been accumulated, can substantially improve matching efficiency \citep[see][]{feng2021two}.
The two-stage model parsimoniously distills all uncertainty into one stage of future arrivals.

We now describe the two-stage (edge-weighted) bipartite matching model.  There is an underlying bipartite graph $G$ consisting of supply nodes $S$ on one side and demand nodes $D$ on the other, with weighted edges between $S$ and $D$.  The demand nodes are revealed in two stages.  In the first stage, a subset of demand nodes $D_1\subseteq D$ is revealed, along with the weighted edges between $D_1$ and $S$.  At this point, the algorithm must commit to a matching $M_1$ between $D_1$ and $S$, which means selecting a subset of these edges such that any vertex (in $D_1$ or $S$) is incident to at most one selected edge.  Afterward, the remaining demand nodes $D_2=D\setminus D_1$ and their incident edges are revealed, at which point the algorithm selects a matching between $D_2$ and $S$.  Importantly, nodes in $S$ incident to an edge selected in the first stage become unavailable for the second stage, and the algorithm does not know $D_2$ when committing to matches for $D_1$.  The objective is to maximize the total weight of edges selected across both stages.  Note that in the second stage, one should always select the max-weight matching between $D_2$ and the nodes in $S$ still available, and hence only in the first stage is there a decision of matching under uncertainty. We use $G_1$ to denote the first-stage graph, which consists of the supply nodes $S$, the first-stage demand nodes $D_1$, and the edges between them. Similarly, we use $G_2$ to denote the second-stage graph. We let $w(M_1,G_2)$ denote the total weight matched if $M_1$ was selected in the first stage and the best matching (subject to $M_1$ and $G_2$) was selected in the second stage.

The algorithm has some information about the second stage when making the first-stage decision, which we abstractly represent as a piece of "advice" $A$.  The advice $A$ is a suggestion for the matching $M_1$ to select in the first stage. In practice, the advice might be generated from: 
\begin{itemize}
\item A point prediction $\hG$ for the second stage $G_2$, in which case $A\in\argmax_{M_1} w(M_1,\hG)$;
\item A distributional prediction $\hcG$ for $G_2$, in which case $A\in\argmax_{M_1} \bE_{G_2\sim\hcG}[w(M_1,G_2)]$;
\item An expert or black-box algorithm, in which case $A$ is an arbitrary first-stage matching $M_1$.
\end{itemize}
Our algorithm will not require knowledge about how exactly the advice was generated. Our goal is to prescribe a first-stage decision $M^A_1$ for any first-stage graph $G_1$ and given advice $A$, whose performance simultaneously satisfies
\begin{align}
w(M^A_1,G_2) &\ge C\cdot w(A,G_2) \label{eqn:introConsistent}
\\ w(M^A_1,G_2) &\ge R\cdot\max_{M_1}w(M_1,G_2) \label{eqn:introRobust}
\end{align}
for all possible second-stage graphs $G_2$.  We call an algorithm satisfying~\eqref{eqn:introConsistent}--\eqref{eqn:introRobust} for any initial setup, given advice, and second stage \textit{$C$-consistent} and \textit{$R$-robust}, respectively.  Intuitively,  consistency measures the performance relative to the advice. Additionally, the consistency guarantee~\eqref{eqn:introConsistent} implies
\begin{align*}
w(M^A_1,\hG) &\ge C\cdot \max_{M_1} w(M_1,\hG)
\\ \bE_{G_2\sim\hcG}[w(M^A_1,G_2)] &\ge C\cdot \max_{M_1} \bE_{G_2\sim\hcG}[w(M_1,G_2)].
\end{align*}
In other words, when the advice is generated from a point or distributional prediction for $G_2$,~\eqref{eqn:introConsistent} guarantees the prescription $M^A_1$ to be within a factor $C$ of the optimal decision, when the ground truth is indeed drawn according to that prediction.

On the other hand, robustness measures the performance of the algorithm relative to the best matching in hindsight, ensuring that the prescribed decision $M^A_1$ never performs too poorly on any  second stage graph $G_2$.
Indeed, when deploying algorithms with predictions in practice, it is important to demonstrate that they never drastically fail on any test case \citep[see][]{wiermanPres,christianson2023optimal,yeh2024sustaingym}.

\subsection{Algorithm and Theoretical Results} \label{subsec:introAlg}

Our goal is to derive algorithms with advice that are simultaneously $C$-consistent and $R$-robust, i.e.\ satisfy~\eqref{eqn:introConsistent} and~\eqref{eqn:introRobust}, for values of $C$ and $R$ as large as possible.  A naive guarantee on $C$ and $R$ can be derived by flipping a biased coin and, based on the outcome, running one of two algorithms:
\begin{enumerate}[I.]
\item Set $M^A_1=A$, so that {$w(M^A_1,G_2)=w(A,G_2)$} and~\eqref{eqn:introConsistent} holds with $C=1$;
\item Ignore $A$ completely, and focus on maximizing the robustness $R$ in~\eqref{eqn:introRobust}.
\end{enumerate}
By varying the bias of the coin, one can make~\eqref{eqn:introConsistent} and~\eqref{eqn:introRobust} hold {in expectation} for values of $(R,C)$ ranging along the line segment from $(R_{\min},1)$ to $(R_{\max},R_{\max})$, where $R_{\min}$ is the worst-case ratio of $w(A,D_2)/\max_{M_1}w(M_1,D_2)$ for Algorithm~I, and $R_{\max}$ is the value of $R$ for Algorithm~II.

Now, the values of $R_{\min}$ and $R_{\max}$ depend on the setting of online bipartite matching being studied.
We describe these settings below, in order from least to most restrictive.
\begin{itemize}
\item \textbf{Edge-weighted}: edges in graph $G$ can have arbitrary non-negative weights.  Here, it is easy to see that $R_{\min}=0$ and $R_{\max}=1/2$.  The latter is achieved by an algorithm that, with equal probability, either selects the maximum weighted matching in the first stage, or does nothing in the first stage and selects the maximum weighted matching in the second stage.
\item \textbf{Vertex-weighted}: all edges incident to the same supply node must have the same weight, but weights can otherwise be arbitrary.  Here, we have $R_{\min}=0$ and $R_{\max}=3/4$.  The latter can be achieved by the algorithm of \citet{feng2021two} or variants, as we explain later.
\item \textbf{Unweighted}: all edges in $G$ must have the same weight of 1.  Here, it is easy to see that $R_{\min}=1/2$ as long as advice $A$ suggests a maximal matching, which can be assumed without loss of generality.  Meanwhile, we have $R_{\max}=3/4$ as a corollary from the vertex-weighted setting.
\end{itemize}

The guarantees for $(R,C)$ implied by the coin-flip algorithm are plotted in \Cref{fig:coin_flip}.
The corresponding tradeoff between robustness and consistency turns out to be best-possible for the edge-weighted and unweighted settings, as we show in \Cref{sec:unweighted_edge_weighted}.  However, in the vertex-weighted setting, better algorithms with advice are possible, which we now describe.
We allow for fractional matchings, explaining later that this translates to a randomized matching algorithm.  Fractional matching means that the algorithm can \textit{partially} select edges with any value in [0,1], as long as the "fill" of any vertex, defined as the sum of edge values incident to that vertex, does not exceed 1.  The first-stage tradeoff can then be described as "how far to fill" each supply node in $S$, in the face of uncertainty about the second stage and which supply nodes may be needed then.

\begin{figure}
    \centering
    \includesvg[scale=0.5]{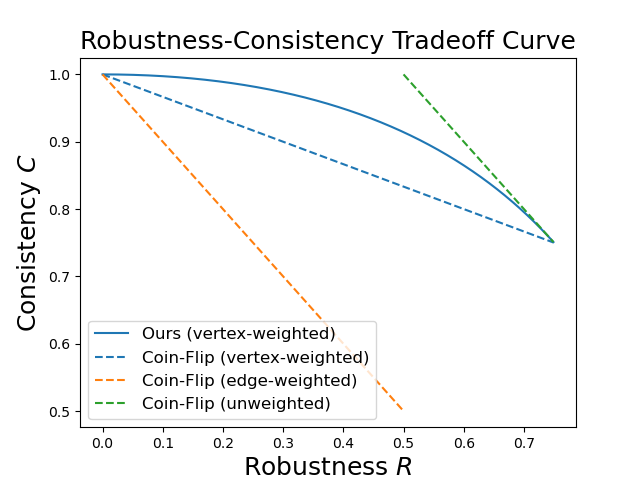}
    \caption{Robustness-Consistency tradeoffs of various algorithms}
    \label{fig:coin_flip}
\end{figure}


In the absence of advice, a common approach to maximizing robustness in online matching is to charge a penalty to each supply node that is a convex function of its fill.
This encourages algorithms to balance how far they fill the supply nodes, thereby ensuring that all supply nodes are partially available for later stages.
In the two-stage model, \citet{feng2021two} show that using a simple linear penalty function and solving a convex optimization problem to handle the batched arrivals leads to a 3/4-robust algorithm.
It is easy to see that their algorithm is only 3/4-consistent, because it makes no use of advice.

\noindent\textbf{{Key algorithmic idea for incorporating advice into online matching.}}
We adjust the penalty function of each supply node based on whether it is filled in the advice.
To elaborate, for any desired robustness level $R$, we characterize a family of penalty functions that guarantees $R$-robustness.
Our algorithm then uses the lower envelope of this function family to penalize supply nodes that are incident to an edge in the suggested matching $A$; that is, it charges smaller penalties to encourage the filling of supply nodes that were also filled in the advice.
Meanwhile, it uses the upper envelope of this family for supply nodes not filled by the advice, resulting in higher penalties that discourage their filling.
The value of $R$ controls the distance between the upper and lower envelope functions, and thereby controls how strongly our algorithm adjusts to the advice, as opposed to trying to balance the supply nodes for robustness.
We prove that our algorithm is $C$-consistent, where $C$ is the unique value (decreasing in $R$) that satisfies the equation $\sqrt{1-R} + \sqrt{1-C} = 1$.
This curve, plotted in \Cref{fig:coin_flip}, is a significant improvement over the coin-flip algorithm.  We also show this elementary, symmetric curve to be the best-possible tradeoff between robustness and consistency for two-stage vertex-weighted bipartite matching.

\noindent\textbf{Extension to Adwords and fractional advice.} We show that our results also hold for the two-stage Adwords (budgeted allocation) problem with advice, specified in \Cref{sec:notation}. Here, we allow both the algorithm and the advice to make fractional allocations.  We extend our penalty function definition, since the fill of a supply node in the suggested (fractional) allocation $A$ can now be any real number $a\in[0,1]$.  The defined function decreases pointwise in $a$, and recovers the upper envelope if $a=0$, and recovers the lower envelope if $a=1$.  We derive the same result, that the best-possible tradeoff between robustness $R$ and consistency $C$ satisfies $\sqrt{1-R} + \sqrt{1-C} = 1$. 

\noindent\textbf{{Proof techniques}.} Our proof builds upon the framework of \citet{feng2021two}, but introduces several new ingredients.  First, because we need to adjust for advice and take the upper/lower envelopes, our penalty functions are no longer \textit{strictly} increasing from 0 to 1.  This requires {an adjusted} structural decomposition for the optimal solution to the convex optimization problem, in which some of the properties from \citet{feng2021two} do not hold.  Second, our result requires characterizing a larger family of penalty functions that guarantees $R$-robustness.  Even when $R=3/4$, i.e.\ when advice is ignored, our proof reveals that a family beyond the linear penalty function of \citet{feng2021two} can be 3/4-robust. In general, because \citet{feng2021two} do not consider online matching with advice, our consistency proofs will be different from their work.

\noindent\textbf{Interpreting our guarantee in terms of advice accuracy.}
Suppose the true second stage $D_2$ follows some underlying
distribution $\cD$, for which the optimal objective value is $w^*(\cD):=\max_{M_1}\bE_{D_2\sim\cD}[w(M_1,D_2)]$.
Our result says that for any $R\in[0,3/4]$, one can simultaneously satisfy desiderata~\eqref{eqn:introConsistent}--\eqref{eqn:introRobust}, where $C$ is defined by $\sqrt{1-R} + \sqrt{1-C} = 1$.
Taking linearity of expectation and solving for $C$, our algorithm's expected performance $\bE_{D_2\sim\cD}[w(M^A_1,D_2)]$ is at least
\begin{align} \label{eqn:123980}
\max\left\{(2\sqrt{1-R}-(1-R))\frac{\bE_{D_2\sim\cD}[w(A,D_2)]}{w^*(\cD)},R\right\}w^*(\cD).
\end{align}
The ratio $\bE_{D_2\sim\cD}[w(A,D_2)]/w^*(\cD)$ can be interpreted as a measure of "advice accuracy", where the hope is that if $A$ optimized a predicted distribution $\hcD$ "close to" $\cD$, then this ratio will be close to 1.
The first argument in~\eqref{eqn:123980} will provide a better guarantee than $R$ under sufficiently high advice accuracy, although we do not try to provide explicit bounds on how prediction accuracy translates to advice accuracy (see \citet{lavastida2021learnable} for some results of this form).

\noindent\textbf{Multiple stages.}
Classical models of online matching and budgeted allocation \citep{karp1990optimal,mehta2007adwords} allow more than two stages of online arrivals, and algorithms that \textit{hybridize} between a fixed "advice" algorithm and traditional algorithms that "balance" against the adversary were introduced in \citet{mahdian2012online}, which is to our knowledge the first occurrence of algorithms with advice in the literature. However, \citet{mahdian2012online} assume the advice to be exogenously fixed instead of adapting to the decisions of the algorithm, making it difficult to define the notion of a "tight" robustness-consistency tradeoff.
Therefore, we instead focus on the elegant two-stage model of \citet{feng2021two}, avoiding complexities in how the advice should adapt to the algorithm's past decisions.

{The guarantees in \citet{mahdian2012online} do apply to the two-stage model of \citet{feng2021two} by treating every demand vertex in each batch as arriving sequentially.}  However, their  guarantees are significantly worse than ours (see \Cref{fig:rc_curve}, which plots in green the tradeoff curve described in Theorem 4.1 of their paper).
In fact, as observed in the upper-left part of \Cref{fig:rc_curve}, the guarantees of \citet{mahdian2012online} can be worse than the simplistic coin-flip algorithm that either runs the "advice" algorithm, or runs the "balance" algorithm from the multi-stage setting\footnote{In the multi-stage setting, the analogue of "balance" is only $(1-\frac{1}{e})$-robust, hence the guarantees for the coin-flip algorithm now range along the line between $(R,C)=(0,1)$ and $(R,C)=(1-\frac{1}{e},1-\frac{1}{e})$ (see the red dashed line in \Cref{fig:rc_curve}).}.  This demonstrates the need for a simpler advice-augmented online matching framework as we propose, in which the Pareto-efficient tradeoff between robustness and consistency can be exactly understood. 

\begin{figure}  
    \centering
        \includesvg[scale=0.5]{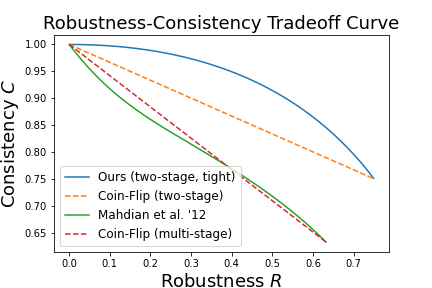}
    \caption{Comparison with the guarantees of Mahdian et al.}
    \label{fig:rc_curve}
\end{figure}

We do believe our key algorithmic idea for incorporating advice into online matching to be potentially applicable even if there are more than two stages.  We also note that our idea based on $R$-parametrized robustness envelopes for penalty functions is quite distinct from the fixed penalty functions used in \citet{mahdian2012online}.

\subsection{Applications and Computational Results} \label{subsec:introNum}

This paper focuses on \textit{two-stage} matching with \textit{partially-reliable} information, for which we now discuss some applications.
\citet{feng2021two} originally motivate two-stage matching using the batched matching of ride-hailing platforms, studying both the stochastic and adversarial extremes.
Our study of the intermediate regime can be motivated by the ride-hailing platform having information in the form of past demands from the same time window, which is only \textit{partially reliable} due to the significance of day-of demand shocks \citep[see][]{yu2023iterative}.
Meanwhile, two-stage matching can also be used as a "surrogate" for inventory \textit{placement} in e-commerce fulfillment, where the second stage demand is assumed to all arrive in one batch \citep[see][]{devalve2023understanding,bai2022joint,bai2022coordinated,jasin2024inventory}.
Placement is then exactly the first-stage decision for a two-stage matching problem, where the information about the second stage might only be partially reliable \citep[see][]{epstein2024optimizing}.
Finally, two-stage matching models have been employed for assigning mobile push ads to users subject to budget constraints \citep[see][]{wang2016dynamic}.

We conduct two sets of computational experiments motivated by these applications.
We first consider a setting with synthetic data, where the second-stage graph is drawn from an underlying distribution from which we are given samples.
The first-stage demand nodes are adjacent to every supply node, allowing the decision-maker full flexibility in which supply nodes to fill subject to a constraint on the total fill.  This is similar to the inventory placement problem given a total amount of inventory, where \textit{not filling} a supply node saves it for the second stage and is analogous to placing inventory there.

We compare our algorithm to the stochastic and adversarial benchmarks, as well as the greedy heuristic.
The stochastic optimization solution, based on the empirical distribution of the second-stage samples, performs best when the sample size is sufficiently large.
Meanwhile, the robust algorithm of \citet{feng2021two}, which ignores the samples, performs best when the sample size is small.
In between these extremes, our algorithm performs best, because it is able to leverage the partially-reliable empirical distribution without overfitting to it.
We also test a setting with distribution shift, in which our algorithm performs best even on larger sample sizes.

We also conduct computational experiments on real data using a public ride-sharing dataset from~\cite{chicago}. This dataset records anonymized ride-sharing trips from 2022, including pickup and dropoff locations, with timestamps rounded to the nearest 15 minutes. We construct two-stage bipartite matching instances from this data, following a similar approach to the experiments in \cite{feng2021two}. We test two methods of generating advice: first, by perturbing the optimal matching in hindsight and testing the algorithms across a range of perturbation levels; second, by taking the advice to be the best matching based on a prediction of the second-stage graph. Overall, our experiments suggest that incorporating advice is valuable in practice, particularly when there is flexibility in controlling the degree to which the advice is followed.



In all our experiments, we focus on the unweighted and vertex-weighted settings.
This can be justified in ride-hailing if the primary objective is to match drivers (supply nodes) who have high idle time, as considered in the experiments of \citet{feng2021two}.
It is also assumed in the mobile push ad application of \citet{wang2016dynamic}, which is a budgeted allocation problem.
This is less justified in the e-commerce fulfillment setting as pioneering work \citep[see][]{jasin2015lp} allows arbitrary edge weights between supply and demand nodes; regardless, we focus on testing vertex weights, to be consistent with the other settings as well as our main algorithmic idea.

\noindent\textbf{Organization of paper.} 
We gather some notation in \Cref{sec:notation}. \Cref{sec:unweighted_edge_weighted} contains the analysis for unweighted and edge-weighted matching. \Cref{sec:prelims} discusses the structure of worst-case instances for vertex-weighted matching and Adwords, which are used in the subsequent analysis.  Our algorithm for vertex-weighted matching is formally defined in \Cref{sec:alg_vertex_weighted}, with its analysis in \Cref{sec:analysis_vertex_weighted}. We generalize the algorithm  to Adwords and fractional advice in \Cref{sec:alg_adwords}, with its analysis in \Cref{sec:adwords_proofs}. The hardness instance certifying the tightness of our tradeoff curve $\sqrt{1-R}+\sqrt{1-C}=1$ is presented in \Cref{sec:tight}. The experimental results are presented in \Cref{sec:experiments}.

\subsection{Further Related Work} \label{sec:relatedWork}

Despite the recent surge of interest in both online matching and online algorithms with advice (see the surveys \citet{huang2024online}  and \citet{mitzenmacher2020algorithms} respectively), literature on online matching with advice has been relatively scant. We mention some papers in this intersection below, as well as other related work. 
More broadly, online bipartite matching focuses on the uncertainty aspect of the general matching problem for two-sided marketplaces.  Other aspects include stability, dynamic learning, competing incentives and objectives, etc.\ which have extensive literatures as well, that are summarized nicely in the book \citet{echenique2023online}.

\noindent\textbf{Online matching with advice.}
Since \citet{mahdian2012online}, more recently \citet{antoniadis2020secretary} have introduced an online random-order edge-weighted bipartite matching problem with advice that predicts the edge weights adjacent to each offline vertex in some optimal offline matching.
\citet{lavastida2021learnable} introduced a framework that formalizes when predictions can be learned from past data, and their algorithmic performance degrades gracefully as a function of prediction error.
Broadly speaking, our work contrasts these three papers in that we are able to understand \textit{tight} tradeoffs in whether to trust some advice, although in an arguably simpler setting.

\noindent\textbf{Online matching models in-between adversarial and stochastic.}
Instead of abstracting all information about the future into a single piece of prediction/advice that could be wrong, another approach is to postulate an explicit model for how the future can deviate from past observations.  Examples of this include the semi-online bipartite matching model of \citet{kumar2019semi}, the partially predictable model of \citet{hwang2021online}, and the multi-channel traffic model of \citet{manshadi2022online}.
We note that random-order arrivals can also be viewed as a form of partial predictability which allows learning \citep{devanur2009adwords}, and moreover it is possible to derive simultaneous guarantees under adversarial and random-order arrivals \citep{mirrokni2012simultaneous} which have the same flavor as robustness-consistency guarantees.
Finally, we mention that the single sample model in \citet{kaplan2022online} can be viewed as a form of online matching with advice that is highly erroneous.

\noindent\textbf{Optimality results in prediction-augmented online algorithms.}
Tight robustness-consistency tradeoffs have become recently understood in prediction-augmented ski rental \citep{purohit2018improving,bamas2020primal,wei2020optimal} and single-commodity accept/reject problems \citep{sun2021pareto,balseiro2022single}.
Our online matching problem contrasts these by having a \textit{multi-dimensional} state space, for which to our knowledge tightness results are rare.
We should mention that in multi-dimensional problems such as prediction-augmented caching \citep{lykouris2021competitive} and online welfare maximization \citep{BGGJ22}, "optimality" results in which consistency can be achieved with no loss of robustness (i.e.\ there is no "tradeoff") have been derived.
{Tight robustness-consistency tradeoffs have also been recently derived in mechanism design with predictions \citep{agrawal2022learning,berger2023optimal,balkanski2023online}.}

\noindent\textbf{Two-stage models.}
Two-stage models capture the essence of optimization under uncertainty, where the first-stage decision must anticipate the uncertainty,
and the second-stage decision is usually a trivial recourse after the uncertainty (in our case the second-stage graph) has been realized.
We refer to \citet{birge2011introduction} and \citet{bertsimas2011theory} for broad overviews of two-stage stochastic and robust optimization.
In this paper we focus on the two-stage online matching model of \citet{feng2021two}, introducing advice to this model and fully characterizing the tradeoff between obeying vs.\ disobeying the advice.
We note that a multi-stage online matching model with batching has also been recently considered in \citet{feng2020batching}, and two-stage matching has been formulated as a robust optimization problem in \citet{housni2020matching}.

\section{Notation}
\label{sec:notation}
The problem input consists of a bipartite graph $G$ with supply nodes $S$ on one side and demand nodes $D$ on the other.  The nodes in $S$ are known in advance, whereas the nodes in $D$  arrive in two batches $D_1, D_2$ with $D = D_1 \cup D_2$. We refer to the nodes in $D$ as online/demand nodes, and the nodes in $S$ as the offline/supply nodes. We index the vertices in $D$ (resp. $S$) with $i$ (resp. $j$). When the vertices in batch $D_k$ ($k=1,2$) arrive, their incident edges $E_k$ are revealed, and the algorithm irrevocably chooses a subset of edges $M_k \subseteq E_k$. The algorithm returns $M := M_1 \cup M_2$.

We consider four versions of online matching: Unweighted matching, edge-weighted matching, vertex-weighted matching, and Adwords. Each comes with a constraint on $M$ and an objective to maximize. In the first three settings, the constraint on $M$ is that it must be a matching in $G$; that is, no two edges in $M$ are incident to the same vertex. Their respective objectives are as follows: 
\begin{enumerate}
    \item \emph{Unweighted.} There are no weights, and the goal is to maximize the number of edges selected.
    \item \emph{Edge-weighted.} Every edge has a non-negative weight, and the goal is to maximize the sum of weights of the edges selected.
    \item \emph{Vertex-weighted.} Every supply node $j$ has a non-negative weight, and the goal is to maximize the sum of weights of the matched offline vertices. 
\end{enumerate}
The Adwords problem was introduced in~\cite{mehta2007adwords}, and the problem formulation is motivated by digital advertising.
\begin{enumerate}
\setcounter{enumi}{3}
\item \emph{Adwords.} Every edge $ij$ has a \emph{bid} $b_{ij}$, and each supply node $j \in S$ has a \emph{budget} $B_j$.\footnote{In the digital advertising interpretation, $j$ represents an advertiser, $i$ represents a consumer, and selecting edge $ij$ represents showing consumer $i$ an ad from advertiser $j$.} Selecting an edge $ij$ consumes an amount of budget from the offline vertex equal to the bid of the edge. The goal is to maximize the total sum of the bids of the selected edges, subject to the constraints that the budgets on the supply nodes are not exceeded.
\end{enumerate}

Note that edge-weighted matching and Adwords generalize the unweighted and vertex-weighted settings. In particular, vertex-weighted matching is Adwords with $b_{ij} = B_j$ for all edges $ij$. However, Adwords is not a special case of edge-weighted matching, and  edge-weighted matching is not a special case of Adwords. We sometimes abuse terminology and also refer to a feasible solution for Adwords as a "matching", even though in Adwords multiple demand nodes can be matched to a given supply node, as long as the budget of the supply node is not exceeded.

We assume the algorithm is given a \emph{suggested matching} $A$ when the first-stage graph is revealed. We evaluate an algorithm by its robustness and consistency, which are defined as follows.

\begin{definition}[Robustness and Consistency] We say that an algorithm is $R$-robust  if
$$\mathbb{E}[\ALG(G,A)] \geq R\cdot\OPT(G) \quad \text{for all $G, A$},$$
and we say an algorithm is $C$-consistent if
$$\mathbb{E}[\ALG(G,A)] \geq C \cdot \ADVICE(G,A) \quad \text{for all $G, A$}.$$
Here, $\mathbb{E}[\ALG(G,A)]$ is the expected value\footnote{The expectation is over any possible randomness in the algorithm, as we allow our algorithms to be randomized. There are no other sources of randomness in this definition.} earned by the algorithm when the input is $G$ and the advice is $A$, while $\OPT(G)$ is the value of the optimal hindsight matching in $G$, and $\ADVICE(G, A)$ is the value that the algorithm would have gotten, had it followed the advice exactly in the first stage (and solved for the optimal matching in the second stage given its first stage decision).
\end{definition}

The main question we ask in this paper is the following.

\begin{tcolorbox}
\centering
\textbf{Question.} What is the optimal tradeoff between robustness and consistency for each of the four matching problems above?
\end{tcolorbox}

We first characterize the tradeoff for unweighted and edge-weighted matching in \Cref{sec:unweighted_edge_weighted}. For these two settings, the optimal tradeoff is easy to characterize, and is achieved by a simple straight-line interpolation between the "fully robust" algorithm and a "fully consistent" algorithm. In the remainder of the paper, we then study the optimal tradeoff for vertex-weighted matching and Adwords, for which the straight-line interpolation is no longer optimal and more involved algorithms and analyses are required. 

\section{Warmup: The Unweighted and Edge-Weighted Settings}
\label{sec:unweighted_edge_weighted}
In this section, we show that for two-stage unweighted and edge-weighted bipartite matching, the optimal robustness-consistency tradeoff is given by a naive straight-line interpolation between the "fully robust" algorithm and the "fully consistent" algorithm.

\subsection{The Unweighted Case} 
\label{app:unweighted}

On the one hand, we have the fully-robust algorithm of \citet{feng2021two}, which is $\frac34$-robust, and hence $\frac34$-consistent. On the other hand, we have 
 the fully-consistent algorithm that always follows the advice, which is 1-consistent and $\frac12$-robust.\footnote{The robustness of $\frac12$ is because any maximal matching is at least  $\frac12$-robust when the graph is unweighted. Note that we may assume without loss of generality that the advice is a maximal matching, as it is never optimal to select a non-maximal matching in the first stage.} Therefore, the naive coin-flip algorithm (that runs the algorithm of \citet{feng2021two} with probability $p$ and follows the advice exactly with probability $1-p$) achieves the robustness-consistency curve defined by the line segment between $(R, C) = (\frac12, 1)$ and $(\frac34, \frac34)$. Somewhat surprisingly, this naive tradeoff is tight in the unweighted setting, which we now show.

\begin{figure}
    \centering
    \begin{tikzpicture}[scale=0.73]
    \vertex (s1) at (0, 3) [label=above left:1] {};
    \vertex (s2) at (0, 0) [label=below left:2]{};
    \vertex (d1) at (3, 3) [label=right:1] {};
    \vertex (d2) at (3, 0) [label=right:2] {};
    \node (label3) at  (2.5, 1.8) {\text{\footnotesize{$1-x$}}};
    \draw (s1) --node[above] {\text{\footnotesize{$x$}}} (d1);
    \draw[green] (s1) -- (d1);
    \draw (s2) -- (d1);
    \draw[dotted] (d2) -- (s1);
    \draw[dotted] (d2) -- (s2);
    \end{tikzpicture}
    \caption{Illustration of the hardness instance for unweighted matching. $S$ is on the left and $D$ is on the right. The first arrival neighbors both vertices of $S$. The second arrival neighbors exactly one vertex of $S$, but it could be either vertex. The advice is to match the green edge $(1,1)$.}
    \label{fig:hardness_unweighted}
\end{figure}

\textbf{The hardness instance.} The hardness instance is illustrated in \Cref{fig:hardness_unweighted}. It is a "Z" graph with two supply nodes $S = \{1,2\}$ and two demand nodes $D = \{1,2\}$, with one demand node arriving in each stage. The first-stage graph consists of both edges $(1,1)$ and $(1,2)$, and the advice suggests matching $(1,1)$. The second-stage graph consists of either $(2,1)$ or $(2,2)$. Since the algorithm does not know which second-stage edge will arrive, so it must hedge against both possibilities when making its first-stage decision.

\begin{proposition}
\label{prop:unweighted}
In the unweighted setting, any algorithm that is $R$-robust can be at most $C$-consistent where $R + C = \frac32$. 
\end{proposition}
\begin{proof}{\emph{Proof}.} 
Let $x := x_{11}$ and $1-x := x_{12}$, so that the algorithm's first-stage decision is entirely characterized by the value of $x$.  There are two cases.
\begin{enumerate}
    \item Edge $(2, 1)$ arrives in the second stage. Then $\ALG(G,A) = 2-x$ and $\ADVICE(G,A) = 1$. 
    \item Edge $(2,2)$ arrives in the second stage. Then $\ALG(G,A) = x + 1$ and $\ADVICE(G,A) = 2$. 
\end{enumerate}
Regardless of which edge arrives in the second stage, $\OPT(G)$ is always equal to $2$. Thus, for the algorithm to be $R$-robust in Case 1, we must have
$$
\frac{2-x}{2} \geq R \implies x \leq 2 - 2R.
$$
On the other hand, for the algorithm to be $C$-consistent in Case 2, we must have
$$
\frac{x+1}{2} \geq C \implies x \geq 2C - 1.
$$
Since the algorithm does not know which of the two cases will happen in the second stage, it must choose an $x$ that satisfies both of the inequalities above. For a desired robustness $R$ and consistency $C$, this is only possible if
$2C-1 \leq 2 - 2R,$
which when rearranged becomes 
$
R + C \leq \frac{3}{2}.
$
\qed
\end{proof}

\subsection{The Edge-Weighted Case}
\label{app:edge_weighted}
A simple straight line tradeoff turns out to also be optimal in the edge-weighted setting. 

\begin{proposition}
\label{prop:edge_weighted}
For two-stage edge-weighted bipartite matching with advice, the optimal robustness-consistency tradeoff is the line segment between $(R, C) = (0,1)$ and $(R, C) = (\frac12,\frac12)$.
\end{proposition}
\begin{proof}{\emph{Proof}.} 
In two-stage edge-weighted bipartite matching, note that the maximum robustness is $\frac12$, which is attained by the algorithm that 1) with probability $\frac12$, finds the maximum matching $M_1$ in the first stage and does nothing in the second stage, and 2) with probability $\frac12$, does nothing in the first stage and finds the maximum matching $M_2$ in the second stage. This algorithm is $\frac12$-robust because the value of the optimal
matching is at most the sum of the values of $M_1$ and $M_2$. Moreover it is an easy exercise to find an example which shows that no algorithm can be more than $\frac12$-robust for two-stage edge-weighted bipartite matching. On the other hand, the algorithm which always follows the advice is 1-consistent, but 0-robust.
Thus, the coin-flip algorithm which naively interpolates between the
 $\frac12$-robust algorithm and the 1-consistent algorithm attains the straight-line tradeoff between $(R,C) = (\frac12, \frac12)$ and $(0, 1)$. 
 
 To show the tradeoff is tight, consider an instance with only one supply node $v$. The first-stage graph consists of a single edge to $v$ with weight 1, and suppose the advice suggests matching the edge. Let $x$ be the probability the algorithm matches the edge. Any $R$-robust algorithm must have $x \leq 1-R$; if $x > 1-R$ then the algorithm is not $R$-robust when the second stage consists of a single edge to $v$ with very high weight. Since $x \leq 1-R$, this means the maximum consistency on this instance is $1-R$ (which is the case if the second-stage graph is empty).
\qed
\end{proof}

\section{Preliminaries for Vertex-Weighted Matching and Adwords}
\label{sec:prelims}
Having determined the optimal robustness-consistency tradeoffs for the unweighted and edge-weighted settings, in the remainder of the paper we focus on vertex-weighted matching and Adwords. For these two settings the optimal tradeoff curve is given by the equation $\sqrt{1-R} + \sqrt{1-C} = 1$ for $R \in [0,\frac34]$, which strictly dominates the straight-line interpolation. In this section we collect some preliminary facts that will be useful for the later analysis.

\subsection{Vertex-Weighted Matching}
\label{subsec:worst_case_vertex_weighted}
As mentioned in the Introduction, we will allow our algorithms to select fractional matchings.  It turns out this is without loss for two-stage vertex-weighted matching, since any fractional algorithm can be converted to a randomized integral algorithm with the same robustness and consistency guarantees. To see this, it is useful to make the following observation about worst-case instances.

\paragraph{\textbf{Worst-case instances.}} When bounding robustness and consistency, one can assume without loss of generality that the second-stage graph consists of a matching (i.e. every vertex has degree at most 1 in the second-stage graph). To see this, suppose we are bounding robustness, and consider the edges \emph{not} matched by $\OPT(G)$ in the second stage. Deleting these edges does not change the value of $\OPT(G)$ and can only decrease the value of the matching found by the algorithm. Therefore we may assume that the second stage graph consists \emph{exactly} of the matching selected by $\OPT(G)$ in the second stage. The same argument shows that when bounding consistency, we may assume that the second-stage graph consists exactly of the matching selected by $\ADVICE(G, A)$ in the second stage. Therefore we may assume the second-stage graph is a matching.

The preceding observation about worst-case instances allows us to focus on the fractional version of the problem.
\begin{proposition}
    \label{prop:vertex_weighted_fractional}
    Given any fractional algorithm for two-stage vertex-weighted bipartite matching, there is a corresponding (randomized) integral algorithm with the same robustness and consistency guarantees.
\end{proposition}

\begin{proof}{\emph{Proof}.} 
In each stage, the algorithm chooses a {fractional matching} instead of an integral one. Let $\vx = (x_{ij}: i \in D_1, j \in S, (i,j) \in E)$ be the fractional matching chosen in the first stage, so that $\vx$ satisfies the constraints $x_i := \sum_{j: (i,j) \in E} x_{ij} \leq 1$ and $x_j := \sum_{i: (i,j) \in E} x_{ij} \leq 1$. Similarly, let $\yv = (y_{ij}: i \in D_2, j \in S, (i,j) \in E)$ be the fractional matching in the second stage, which satisfies $y_i := \sum_{j: (i,j) \in E} y_{ij} \leq 1$ and $y_j := \sum_{i: (i,j) \in E} y_{ij} \leq 1 -  x_j$.
The objective in the fractional problem is to maximize $\sum_{j\in S}w_j(x_j+y_j)$, where $w_j$ is the weight of offline vertex $j$. As mentioned earlier, the first stage decisions $x_j$ are the critical ones; therefore we will use the terminology that each offline vertex $j\in S$ is \textit{filled to water level} $x_j$ at the end of the first stage.
Although the fractional problem seems at first to be easier than the integral problem, it turns out to be straightforward to convert any fractional algorithm for two-stage bipartite matching into a (randomized) integral one with the same robustness/consistency guarantees. 
Consider any fractional algorithm and let $\vx$ be its first-stage output. So, $\vx$ satisfies the constraints $x_i := \sum_{j: (i,j) \in E} x_{ij} \leq 1$ and $x_j := \sum_{i: (i,j) \in E} x_{ij} \leq 1$.
To convert this to a randomized integral algorithm, we sample an integral matching $M_1$ with marginals equal to $\vx$ (i.e. $\mathbb{P}((i,j) \in M) = x_{ij}$ for all edges $(i,j)$ in the first-stage graph).\footnote{Since the bipartite matching polytope is integral, $\vx$ can be written as a convex combination of integral matchings. Such a convex combination can be found in polynomial time using algorithmic versions of Carath\'{e}odory's theorem.} We then take $M_2$ to be the maximum-weight matching in the second-stage graph, subject to $M_1$ already being chosen. 

To analyze the (expected) weight of the integral matching $M_1 \cup M_2$, note that in the first stage the expected weight of $M_1$ is equal to that of $\vx$ by construction. Now consider the second stage. As we have argued above, we may assume that the second-stage graph consists of a matching; let $(i, j)$ be one of these edges. The contribution of $(i,j)$ to the value of the fractional algorithm is $w_j(1 - x_j)$, because the remaining amount that offline vertex $j$ can be filled is $(1-x_j)$. On the other hand, the integral algorithm will match $(i, j)$ in $M_2$ if and only if $j$ is unmatched in $M_1$, which happens with probability $(1- x_j)$. So the contribution of $(i,j)$ to the expected increase of the integral algorithm is also $w_j (1 - x_j)$. \qed
\end{proof}

\subsection{Adwords}
\label{subsec:notation_adwords}
We show that our results extend to two-stage Adwords, when the algorithm and advice can both be fractional. Formally, in two-stage Adwords, let $\vx = (x_{ij}: i \in D_1, j \in S, (i,j) \in E_1)$ be the fractional allocation chosen in the first stage, so that $\vx$ satisfies the constraints 
\begin{equation}
\label{eq:adwords_constraints}
    x_i := \sum_{j: (i,j) \in E_1} x_{ij} \leq 1 ~\text{and}~ x_j := \frac{1}{B_j}\sum_{i: (i,j) \in E_1} b_{ij}x_{ij} \leq 1.
\end{equation}
We interpret $x_j$ as the fraction of $j$'s budget that is used under $\vx$. Similarly, let $\yv = (y_{ij}: i \in D_2, j \in S, (i,j) \in E_2)$ be the fractional matching in the second stage, which satisfies 
$$y_i := \sum_{j: (i,j) \in E_2} y_{ij} \leq 1 ~\text{and}~ y_j := \frac{1}{B_j}\sum_{i: (i,j) \in E_2} b_{ij}y_{ij} \leq 1 -  x_j.$$ The objective is to maximize $\sum_{(i,j) \in E}b_{ij}(x_{ij}+y_{ij})$. We may assume that $b_{ij} > 0$ for all edges $(i,j)$, since if $b_{ij} = 0$ we can delete $(i,j)$ without changing the problem. Also, note that the first stage decisions $x_{ij}$ are the critical ones. We will use the terminology that $j\in S$ is \textit{filled to} level $x_j$ at the end of the first stage.

We show that our guarantees hold for Adwords even if the advice is allowed to be fractional. We let $a_{ij}$ denote the fractional amount the advice fills edge $ij$ in the first stage; so $\mathbf{a}$ is a feasible fractional allocation that satisfies \eqref{eq:adwords_constraints}. Similar to the vertex-weighted case, we may assume the second-stage graph is a matching. 

\begin{proposition}   \label{prop:adwords_worst_case}
    For fractional Adwords, the worst case for robustness and consistency is when the second-stage graph consists of a matching. Moreover, the worst case for consistency is when each edge $(i,j)$ in the second stage matching has $b_{ij} = B_j(1-a_j)$.
\end{proposition}

\begin{proof}{\emph{Proof}.}
    First, we show why we may assume the second-stage graph consists of a matching. To see this, suppose we are bounding robustness. Let $y^*_j$ be the fraction of $j$'s budget used in the second-stage under the optimal solution. Now, imagine creating an alternative second-stage graph, where all edges incident to $j$ are replaced by a single edge with bid $B_jy^*_j$. This does not change the value of $\OPT(G)$, and can only decrease the value of the algorithm. The same argument shows that the worst case for consistency is also when the second-stage graph is a matching.

    Next we show that when bounding consistency, we may further assume that the bid on each edge $(i,j)$ in the second-stage matching is $b_{ij} = B_j(1-a_j)$. This is because, as we just showed, the worst-case for consistency is when the second-stage graph is some matching, say $M$. The ratio of the algorithm's value to the value of following the advice is therefore
    $$\frac{\ALG(G,A)}{\ADVICE(G,A)} = \frac{\mathsf{alg} + \sum_{ij \in M} \min\{B_j(1-x_j), b_{ij}\}}{\mathsf{adv} + \sum_{ij \in M}\min\{B_j(1-a_j), b_{ij}\}},$$
    where $\mathsf{alg}$ denotes the value the algorithm earns in the first stage, and $\mathsf{adv}$ denotes the value that the advice earns in the first stage.
    Considering this ratio as a function of $(b_{ij}: ij \in M)$, it is not hard to see that it is minimized when each $b_{ij}$ is either 0 or $B_j(1-a_j)$. Since having an edge with $b_{ij} = 0$ is the same as not including the edge, we conclude that the worst case for consistency is when the second-stage graph is a matching $M$ with $b_{ij} = B_j(1-a_j)$ for all $(i,j) \in M$. 
    \qed
\end{proof}



We leave it as an open question to characterize the optimal robustness-consistency tradeoff for two-stage \emph{integral} Adwords -- that is, when the algorithm must select an integral allocation. We remark that this problem seems challenging, since integral Adwords is significantly harder than fractional Adwords --- see \cite{huang2020adwords}. In fact, as far as we are aware, even the optimal competitive ratio for two-stage integral Adwords is not known \citep{feng2020batching}.

\section{Algorithm for Vertex-Weighted Matching with Integral Advice}
\label{sec:alg_vertex_weighted}

\begin{algorithm}[t]
\caption{Algorithm for Two-Stage Fractional Vertex-Weighted Bipartite Matching with Advice}
\label{alg:vertex_weighted}
\DontPrintSemicolon
\KwIn{Suggested matching $A$ in the first-stage graph and desired robustness level $R$.}
\begin{algorithmic}[1]
\STATE (Define penalty functions) Let $\fl(x) = \max\{0, 1 - \frac{1-R}{x}\}$ and $\fu(x) = \min\{1, \frac{1-R}{1-x}\}$. 
\STATE (First stage) When the first-stage vertices $D_1$ arrive, let $S_1 \subseteq S$ be the set of offline vertices that are in the suggested matching $A$. Set $f_j = \fl$ for all $j \in S_1$ and $f_j = \fu$ for all $j \in S \setminus S_1$.

Solve the following optimization problem for the first-stage fractional matching $\bar{\vx}$:
\begin{align*}
    (P_1) \quad \max &\sum_{j\in S}w_j\left(x_j-\int_0^{x_j} f_j(t) dt\right) \\
    \mathrm{s.t.\ } x_i:=&\sum_{j: (i,j) \in E} x_{ij} \leq 1 &  \forall \; i \in D_1 \\
    x_j:=&\sum_{i\in D_1: (i,j) \in E} x_{ij} \leq 1 &  \forall \; j \in S \\
    &x_{ij} \geq 0 & \forall \; i\in D_1,(i,j) \in E
\end{align*}



\STATE (Second stage) When the second-stage vertices $D_2$ arrive, solve for the optimal fractional matching $\bar{\yv}$ subject to the capacities already taken by $\bar{\vx}$:
\begin{align*}
    (P_2) \quad \max &\sum_{j \in S} w_jy_j \\
    \mathrm{s.t.\ } y_i:=&\sum_{j: (i,j) \in E} y_{ij} \leq 1 &  \forall \; i \in D_2 \\
    y_j:=&\sum_{i \in D_2: (i,j) \in E} y_{ij} \leq 1 - \bar{x}_j&  \forall \; j \in S \\
    &y_{ij} \geq 0 & \forall \; i \in D_2, \, (i,j) \in E
\end{align*}

\STATE {Return} $\bar{\vx} + \bar{\yv}$. 
\end{algorithmic}
\end{algorithm}

\begin{figure}[t!]
    \centering
    \begin{subfigure}[t]{0.3\textwidth}
        \centering
        \includesvg[scale=0.3]{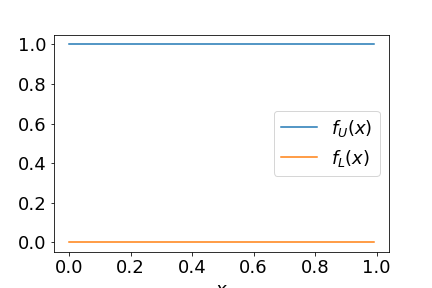}
        \end{subfigure}
        \hspace{1mm}
 \begin{subfigure}[t]{0.3\textwidth}
        \centering
       \includesvg[scale=0.3]{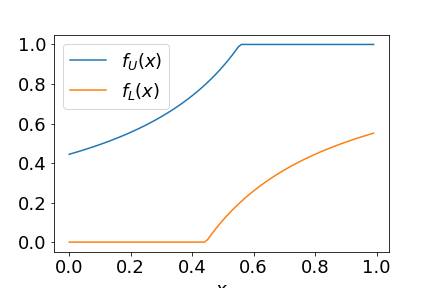}
        \end{subfigure}
 \hspace{1mm}
 \begin{subfigure}[t]{0.3\textwidth}
        \centering
       \includesvg[scale=0.3]{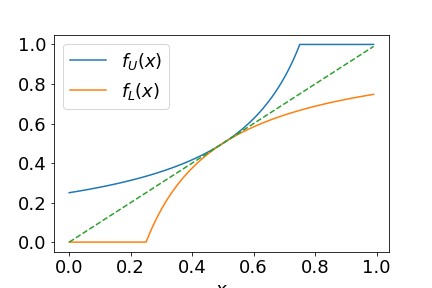}
        \end{subfigure}
    \caption{Plots of $\fl$ and $\fu$ for three values of $R$. Left: $R = 0$. Middle: $R = \frac{5}{9}$. Right: $R = \frac{3}{4}$.  The consistencies achieved are $C=1, C=\frac{8}{9}, C=\frac{3}{4}$, respectively. In the right plot, the green dashed line indicates the penalty function used by \citet{feng2021two}.}
    \label{fig:fufl}
\end{figure}

Here we present our algorithm for two-stage vertex-weighted bipartite matching when the advice is assumed to be integral (we extend to fractional advice in \Cref{sec:alg_adwords}). Recall that from \Cref{prop:vertex_weighted_fractional}, we can round any fractional algorithm to a randomized integral algorithm with no loss. Therefore we will focus on the version of the problem where the algorithm is allowed to make fractional decisions. Our algorithm for this setting is described in Algorithm~\ref{alg:main}. As described in the introduction, the main challenge is to commit to a matching $\vx$ when the first-stage graph is revealed, that incorporates some advice but maintains a worst-case competitive ratio.

Intuitively, the amount $x_j$ that each vertex $j \in S$ is filled to after the first stage is governed by three factors. First, it should depend on the weight $w_j$; the larger $w_j$ is, the larger $x_j$ should be. Second, it should depend on the advice; $x_j$ should be larger if $j$ is recommended by the advice, and smaller if not. Third, it should depend on how much the \emph{other} offline vertices are filled; the larger $x_j$ is compared to the other vertices competing to be filled, the more hesitant we should be of filling it even further, because this can be exploited by a second-stage graph whose edges are incident to the vertices which have already been filled the most.
Typically, the third point is achieved by defining a \emph{penalty} $f(x_j)\in[0,1]$ on the marginal benefit of filling $j$, with $f(x_j)$ increasing in $x_j$, which encourages the balancing of the water levels across $j\in S$. This is combined with the first point by associating a \textit{potential} $w_j(1 - f(x_j))$ with each $j \in S$.  The algorithm then decreases these potentials while filling water levels, with the aim of equalizing them as much as possible.

Our main idea to address the second point, now that there is some advice, is to use a \textit{lower} penalty rate $\fl(x_j)$ to incentivize the filling of offline vertices $j\in S_1$ that are matched by the advice, and a \textit{higher} penalty rate $\fu(x_j)$ for $j\notin S_1$.
More specifically, for any desired "robustness" $R$, our analysis characterizes the envelope of penalty functions $f:[0,1]\to[0,1]$ (see \Cref{thm:vertex_weighted}) that guarantees a worst-case competitive ratio of $R$, and defines $\fl,\fu$ according to the lower, upper boundaries of this envelope respectively.

\Cref{fig:fufl} plots the penalty functions $\fl$ and $\fu$ for three values of $R$. The parameter $R\in[0,\frac{3}{4}]$ affects the separation between $\fl$ and $\fu$.
If $R=0$, then $\fl,\fu$ are the constant functions $0,1$ respectively, from which it can be seen that Algorithm~\ref{alg:main} will fully match all the vertices suggested by the advice.
Intuitively, setting $R=0$ is desirable when we trust that the advice will perform well.
On the other extreme, setting $R=\frac{3}{4}$ guarantees the maximum possible robustness, which is not reliant on the advice.
Surprisingly, even in this case there is a bit of separation between $\fl$ and $\fu$ (see \Cref{fig:fufl}), and hence the algorithm still discriminates between offline vertices $j\in S_1$ vs.\ $j\notin S_1$, despite its distrust in the advice.
We plot one intermediate value where $R=\frac{5}{9}$.


Finally, having decided penalty functions for the offline vertices, there is still the question of how to balance their potential levels when multiple online vertices $D_1$ can arrive in the first stage.
Fortunately, this question was answered by \citet{feng2021two}, who show that problem~$P_1$ is the right one to solve, and establish a neat structural decomposition on its solution.
Our \Cref{lem:decomp_vertex_weighted} is a local version of the structural decomposition from their paper which allows for \textit{heterogeneous} penalty functions across offline vertices $j\in S$ that do not necessarily satisfy their boundary conditions of $f_j(0)=0$ and $f_j(1)=1$. We note that the {original results from} \cite{feng2021two} can also be derived with {this new} structural lemma in place of their structural decomposition. 


\begin{lemma}[Structural Lemma]
\label{lem:decomp_vertex_weighted}
    Let $x$ be the first-stage solution returned by Algorithm \ref{alg:vertex_weighted}. There exist non-negative real numbers $c_i$ for each $i \in D_1$ that satisfy the following properties:
    \begin{enumerate}
        \item For all $(i,j) \in E_1$ with $x_{ij} > 0$, we have $w_j (1-f_j(x_j)) \geq c_i$. 
        \item For all $i, k \in D_1$ such that there exists some $j \in S$ with $x_{k j} > 0$, we have $c_i \geq c_k$. 
        \item For all $i \in D_1$ with  $x_i < 1$, we have $c_i = 0$. 
        \item For any $(i,j) \in E_1$ with $x_j < 1$, we have $c_i \geq w_j(1-f_j(x_j))$.
    \end{enumerate}
\end{lemma}



Note that the value $w_j(1-f_j(x_j))$ is exactly the rate of change of the objective of $(P_1)$ relative to $x_{ij}$. In the Lemma, the value $c_i$ will represent the rate of increase of the objective of $(P_1)$ if the capacity of $i$ were to be increased infinitesimally -- it can be thought of as the "marginal value" or "shadow price" of $i$ at the optimal solution.
Property~1 holds because if $c_i > w_j(1-f_j(x_j))$ instead, then there would be a way to shift some of the flow on edge $(i,j)$ to elsewhere and strictly improve the objective of $(P_1)$.
Meanwhile, Property~2 holds because if ${x}_{kj}>0$, i.e.\ $k$ is being used to fill $j$, then one way to use an $\epsilon$ increase in the capacity of $i$ is to add $\epsilon$ to the current value of $x_{ij}$ and subtract $\epsilon$ from the current value of $x_{kj}$ -- this keeps the value of $x_j$ the same and frees up $\epsilon$ capacity from $k$, which can now be used to improve the objective at a marginal rate of $c_k$. Thus we have $c_i \geq c_k$, since the best way to use the additional capacity of $i$ should improve the objective by at least this rate.
Property~3 holds because if $c_i > 0$ instead, then there would be a way to strictly improve the objective of $(P_1)$ by sending more flow out of vertex $i$. Finally, Property~4 holds because given an infinitesimal increase in the capacity of $i$, one  feasible way to allocate it would be to send flow from $i$ to $j$ -- this increases the objective at a rate of $w_{j}(1-f_j(x_j))$, so the optimal rate of increase $c_i$ should be at least this quantity.

\begin{proof}{\emph{Proof of \Cref{lem:decomp_vertex_weighted}}.} 
    The proof will proceed as follows. We first write down the KKT  conditions of the convex program. Then, we show that setting $c_i$ equal to the optimal dual variable corresponding to the capacity constraint for node $i$ satisfies the desired properties.

    Introducing non-negative dual variables $\lambda_i$, $\theta_j$, and $\gamma_{ij}$ for the constraints of the convex optimization problem, the Lagrangian of the convex program is:
    $$L(\vx; \lambdav, \thetav, \gammav) = \sum_{j \in S} w_j(x_j-F_j(x_j)) + \sum_i \lambda_i(1-x_i) + \sum_j \theta_j(1-x_j) + \sum_{ij} \gamma_{ij} x_{ij}$$
    The KKT conditions state that at optimality, the primal/dual solutions satisfy
    \begin{enumerate}
        \item (Stationarity)
        $w_{j}(1-f_j(x_j)) - \lambda_i - \theta_j + \gamma_{ij} = 0$, for all $(i,j) \in E_1$.
        \item (Complementary Slackness)
        \begin{itemize}
            \item $\lambda_i(1-x_i) = 0$ for all $i \in D_1$,
            \item $\theta_j(1-x_j) = 0$ for all $j \in S$,
            \item $\gamma_{ij}x_{ij} = 0$ for all $(i,j) \in E_1$.
        \end{itemize}
    \end{enumerate}
    We now show that setting $c_i = \lambda_i$ satisfies the properties in the Theorem. We check the properties one by one.
    \begin{enumerate}
        \item Let $(i,j) \in E_1$ be an edge with $x_{ij} > 0$. By complementary slackness, we know $\gamma_{ij} = 0$. Then, stationarity implies
        $$w_{j}(1-f_j(x_j)) = \lambda_i + \theta_j \geq \lambda_i.$$
        \item Consider $i,k \in D_1$ and $j \in S$ with $x_{kj} > 0$. By complementary slackness, $\gamma_{kj} = 0$. By the stationary condition for edge $(k,j)$, we have
        $$w_j(1-f_j(x_j)) = \lambda_k + \theta_j.$$
        On the other hand, by the stationarity condition for edge $(i,j)$, we have
        $$w_j(1-f_j(x_j)) = \lambda_i + \theta_j - \gamma_{ij} \leq \lambda_i + \theta_j.$$
        Comparing the two equations, we see that $\lambda_i \geq \lambda_k$.
        \item This follows directly from complementary slackness, since if $x_i < 1$ then $\lambda_i = 0$. 
        \item Consider $(i,j) \in E_1$ with $x_j < 1$. By complementary slackness, $\mu_j = 0$. Then, stationarity gives
        $$w_j(1-f_j(x_j)) = \lambda_i - \gamma_{ij} \leq \lambda_i.$$
    \end{enumerate}
    \qed
\end{proof}


\subsection{Illustration of how our Algorithm uses Advice} \label{sec:illus}

\begin{figure}

    \centering
    \begin{subfigure}[t]{0.3\textwidth}
        \centering
    \begin{tikzpicture}
    \vertex (s1) at (0, 3) [label=left:$1$] {};
    \vertex (s2) at (0, 2) [label=left:$1$] {};
    \vertex (s3) at (0, 1) [label=left:$2$] {};
    \vertex (s4) at (0, 0) [label=left:$4$] {};
    
    \vertex (d1) at (3, 2) {};
    \vertex (d2) at (3, 1) {};
    \node (label11) at (1.5, 2.8) {\text{\footnotesize{$\frac{3}{11}$}}};
    \node (label12) at (0.7, 2.3) {\text{\footnotesize{$\frac{3}{11}$}}};
    \node (label13) at (1, 1.6) {\text{\footnotesize{$\frac{5}{11}$}}};
    \node (label23) at (2.1, 1.25) {\text{\footnotesize{$\frac{2}{11}$}}};
    \node (label24) at (1.4, 0.2) {\text{\footnotesize{$\frac{9}{11}$}}};
    \draw (s2) -- (d1);
    \draw (d1) -- (s1);
    \draw (d1) -- (s3);
    \draw (d2) -- (s3);
    \draw (d2) -- (s4);
    \end{tikzpicture}
    \caption{The fully robust solution.}
        \label{subfig:illus_a}
    \end{subfigure}
        \hspace{1mm}
 \begin{subfigure}[t]{0.3\textwidth}
        \centering
       \begin{tikzpicture}
    \vertex (s1) at (0, 3) [label=left:$1$] {};
    \vertex (s2) at (0, 2) [label=left:$1$] {};
    \vertex (s3) at (0, 1) [label=left:$2$] {};
    \vertex (s4) at (0, 0) [label=left:$4$] {};
    
    \vertex (d1) at (3, 2) {};
    \vertex (d2) at (3, 1) {};
   \node (label11) at (1.5, 2.8) {\text{\footnotesize{$0$}}};
    \node (label12) at (0.7, 2.3) {\text{\footnotesize{$0$}}};
    \node (label13) at (1, 1.6) {\text{\footnotesize{$1$}}};
    \node (label23) at (2.1, 1.25) {\text{\footnotesize{$0$}}};
    \node (label24) at (1.4, 0.2) {\text{\footnotesize{$1$}}};
    \draw (s2) -- (d1);
    \draw (d1) -- (s1);
    \draw[green] (d1) -- (s3);
    \draw (d2) -- (s3);
    \draw[green] (d2) -- (s4);
    \end{tikzpicture}
    \caption{The algorithm's first-stage matching if the advice suggests $(1,3)$, $(2,4)$.}
        \label{subfig:illus_b}
        \end{subfigure}
 \hspace{1mm}
 \begin{subfigure}[t]{0.3\textwidth}
        \centering
      \begin{tikzpicture}
    \vertex (s1) at (0, 3) [label=left:$1$] {};
    \vertex (s2) at (0, 2) [label=left:$1$] {};
    \vertex (s3) at (0, 1) [label=left:$2$] {};
    \vertex (s4) at (0, 0) [label=left:$4$] {};
    
    \vertex (d1) at (3, 2) {};
    \vertex (d2) at (3, 1) {};
    \node (label11) at (1.5, 2.8) {\text{\footnotesize{$0$}}};
    \node (label12) at (0.7, 2.3) {\text{\footnotesize{$\frac{5}{9}$}}};
    \node (label13) at (1, 1.6) {\text{\footnotesize{$\frac{4}{9}$}}};
    \node (label23) at (2.1, 1.25) {\text{\footnotesize{$\frac{5}{9}$}}};
    \node (label24) at (1.4, 0.2) {\text{\footnotesize{$\frac{4}{9}$}}};
    \draw[green] (s2) -- (d1);
    \draw (d1) -- (s1);
    \draw (d1) -- (s3);
    \draw[green] (d2) -- (s3);
    \draw (d2) -- (s4);
    \end{tikzpicture}
        \caption{The algorithm's first-stage matching if the advice suggests $(1,2)$, $(2,3)$.}
        \label{subfig:illus_c}
        \end{subfigure}
    \caption{Example to illustrate the features of the algorithm. $S$ is on the left and $D_1$ is on the right. The vertices in $S$ are $i=1,2,3,4$ and the vertices in $D_1$ are $j=1,2$ (labelling goes from top to bottom). The number next to $j \in S$ is its weight $w_j$. The green edges are suggested by the advice, and our algorithm's decisions are illustrated for the case where $R=5/9$.}
    \label{fig:illus}
\end{figure}

As an example, suppose $R=\frac{5}{9}$ and $C=\frac{8}{9}$, where we note that $\sqrt{1-R}+\sqrt{1-C}=1$.
Let $D_1$ consist of vertices $i=1,2$ and $S$ consist of vertices $j=1,2,3,4$ with weights $w_j=1,1,2,4$ respectively.
The edges in the graph are $(1,1)$, $(1,2)$, $(1,3)$ and $(2,3)$, $(2,4)$.  

We first demonstrate the case where the advice suggests edges (1,3) and (2,4), which would match weight $w_3+w_4=6$ in the first stage.
Intuitively, this advice is not so "extreme" in that it is greedily matching the highest-weight vertices in $S$ while it can, without assuming it can later match them in the second stage.
In this case our algorithm would follow the advice exactly (see \Cref{subfig:illus_b}), and have a consistency of 1.
Its robustness would be at least $\frac{6}{7}$, with the worst case being when the second stage consists of a single edge (3,4), which cannot be matched by our algorithm but increases the optimal offline matching from 6 to 7.
Nonetheless, this suffices because $\frac{6}{7}$ is well above the targeted robustness of $R=\frac{5}{9}$.

We now demonstrate the case where the advice suggests edges (1,2) and (2,3) instead, perhaps predicting that an edge (3,4) will allow us to match offline vertex 4 later.
This advice is quite "extreme" in that it is skipping the highest-weight vertex in $S$ in the first stage, based on a second-stage prediction which may not come to fruition.
Following it exactly would give a weight of $w_2+w_3=3$ in the first stage, which cannot be more than $\frac{3}{7}$-robust\footnote{The worst case is if the second-stage graph consists of a single edge $(3,2)$, in which case $\ALG(G,A) = 3$ and $\OPT(G) = 7$.}, significantly lower than the target of $\frac{5}{9}$.
Meanwhile, the maximally robust fractional matching based on linear penalty functions (\citet{feng2021two}; shown in \Cref{subfig:illus_a}), which judiciously balances between \textit{all} the offline vertices, is not $\frac{8}{9}$-consistent\footnote{If the second-stage graph consists of edges $(3,1)$ and $(4,4)$, then the fully robust solution gets a value of $1+\frac{3}{11}+2\cdot\frac{7}{11}+4=\frac{72}{11}$ whereas $\ADVICE(G,A) = 8$. Their ratio is $\frac{9}{11}$, which is less than $\frac{8}{9}$.}.
Our algorithm returns the solution in \Cref{subfig:illus_c}, which follows the advice in that it completely prioritizes vertex 2 over vertex 1, but deviates by significantly filling offline vertex~4 (which has the highest weight) as a failsafe (although not as much as the fully robust solution).
This makes it both $\frac{5}{9}$-robust and $\frac{8}{9}$-consistent, and our algorithm can be adjusted accordingly to be both $R$-robust and $C$-consistent for any values $R,C$ satisfying $\sqrt{1-R}+\sqrt{1-C} \leq 1$.

Based on these examples, we highlight two desirable features of our algorithm.  First,
it naturally responds to the "extremity" of the advice, by deviating more from the more extreme advice (in the 2nd case) in order to maintain $R$-robustness.
Second, from the definition of $\fu$ and $\fl$ in Algorithm \ref{alg:main} together with \Cref{lem:decomp_vertex_weighted}, it can be seen that if $R\le \frac{1}{2}$ and the graph is unweighted, then our algorithm will \textit{always} follow the advice exactly. This is because $\fl(1) = R \leq 1 - R = \fu(0)$, which implies the algorithm will prioritize filling a vertex suggested by the advice over one that is not, even if the former is completely filled and the latter is empty.  Put another way, the algorithm automatically recognizes that any maximal matching will be at least $\frac{1}{2}$-robust in an unweighted graph.

\section{Analysis for Vertex-Weighted Matching {with Integral Advice}}
\label{sec:analysis_vertex_weighted}

We now analyze the robustness and consistency of Algorithm \ref{alg:vertex_weighted}. The theorem below gives a characterization of the penalty functions $f_j$ that are sufficient to guarantee $R$-robustness and $C$-consistency, respectively. 

\begin{theorem}
    \label{thm:vertex_weighted}
    Let $R \in [0, \frac{3}{4}]$ and let $C \in [0,1]$. Let $f_j$ denote the penalty functions used in Algorithm \ref{alg:vertex_weighted}, where $f_j: [0,1] \to [0,1]$ is continuous and increasing. Then, the following hold.
\begin{enumerate}
    \item (Robustness) Suppose for all $j \in S$, we have $1 - \frac{1-R}{x} \leq f_j(x) \leq \frac{1-R}{1-x}$ for all $x \in (0, 1)$. Then the algorithm is $R$-robust.
    \item (Consistency) Suppose for all $j \in S$, we have 
    \begin{itemize}
        \item $f_j(x) \leq \frac{1-C}{1-x}$ if  $j$ is matched by the advice, and
        \item $f_j(x) \geq 1 - \frac{1-C}{x}$ if $j$ is not matched by the advice.
    \end{itemize} Then the algorithm is $C$-consistent.
\end{enumerate}
\end{theorem}
Intuitively, the condition for robustness says that the penalty functions should not be too extreme. On the other hand, the condition for consistency says that one should set a lower penalty for $j$ if it is matched by the advice, and a higher penalty otherwise. A direct corollary of \Cref{thm:vertex_weighted} is that Algorithm \ref{alg:vertex_weighted} achieves the robustness-consistency tradeoff given by the equation $\sqrt{1-R} + \sqrt{1-C} = 1$.
\begin{corollary}
    \label{cor:vertex_weighted}
    Let $R \in [0, \frac34]$ and $C \in [0,1]$, and suppose $\sqrt{1-R} + \sqrt{1-C} \geq 1$. Let $\fl(x) = \max\{0, 1 - \frac{1-R}{x}\}$ and $\fu(x) = \min\{1, \frac{1-R}{1-x}\}$ be the penalty functions used in Algorithm \ref{alg:vertex_weighted}.  Setting 
\begin{itemize}
    \item $f_j(x) = \fl(x)$ for all $j \in S$ matched by the advice, and 
    \item $f_j(x) = \fu(x)$ for all $j \in S$ not matched by the advice 
\end{itemize} 
satisfies the conditions for $R$-robustness and $C$-consistency in \Cref{thm:vertex_weighted}.
\end{corollary}

\begin{proof}{Proof.} 
    Clearly $\fl$ and $\fu$ satisfy the condition for robustness in \Cref{thm:vertex_weighted}. To show they satisfy the condition for consistency, we just need to show that 
    $$\fl(x) \leq \frac{1-C}{1-x} ~~\text{and}~~ \fu(x) \geq 1 - \frac{1-C}{x}.$$ 
    The first inequality is equivalent to $1 - \frac{1-R}{x} \leq \frac{1-C}{1-x},$ which rearranges to $ \frac{1-R}{x} + \frac{1-C}{1-x} \geq 1$. Simple calculus shows that
    $$\min_{x \in [0,1]} \left\{\frac{1-R}{x} + \frac{1-C}{1-x} \right\}= \left(\sqrt{1-R} + \sqrt{1-C}\right)^2 \geq 1.$$
    On the other hand, the inequality $\fu(x) \geq 1 - \frac{1-C}{x}$ is equivalent to $\frac{1-R}{1-x} \geq 1- \frac{1-C}{x}$, which rearranges to $\frac{1-R}{1-x}  + \frac{1-C}{x}\geq 1$. Similar to above, this inequality holds because
    $$\min_{x \in [0,1]} \left\{\frac{1-R}{1-x} + \frac{1-C}{x} \right\}= \left(\sqrt{1-R} + \sqrt{1-C}\right)^2 \geq 1.$$
    \qed
\end{proof}

In the remainder of this section, we prove \Cref{thm:vertex_weighted}. The proof will employ the online primal-dual technique, where based on the algorithm's decisions, we construct dual variables that are approximately feasible and equal the algorithm's objective value. We break the proof into several parts. In \Cref{sec:analysis_duals_vertex_weighted}, we show how to set the dual variables, show that their value in the dual objective equals the value of the algorithm (\Cref{clm:pd_equal_vertex_weighted}), and and prove a key lower bound (\Cref{clm:sumdualvar_vertex_weighted}) regarding the sum of the dual variables on any edge. These two claims hold for any functions $f_j: [0,1] \to [0,1]$ that are continuous and increasing, and their proofs rely on the structural property. We then use these two claims in \Cref{sec:analysis_rc_vertex_weighted} to show that if the penalty functions additionally satisfy conditions 1 and 2 in \Cref{thm:vertex_weighted}, then the algorithm is guaranteed to be $R$-robust and $C$-consistent, respectively.

\subsection{Defining the Dual Variables}
\label{sec:analysis_duals_vertex_weighted}
To begin, recall the LP formulation for vertex-weighted bipartite matching and its dual:
\begin{align*}
\mbox{max }  \sum_{(i,j) \in E} w_j z_{ij}  & & \mbox{min }  \sum_{i \in D} \alpha_i + \sum_{j \in S} \beta_j \\
\mbox{s.t. } \sum_{j: (i,j) \in E} z_{ij} \leq 1 &  \qquad\forall i \in D & \mbox{s.t. } \alpha_i + \beta_j \geq w_j & \qquad \forall (i,j) \in E \\
\sum_{i:(i,j) \in E} z_{ij} \leq 1 &\qquad \forall j \in S & \alpha_i, \beta_j \geq 0 & \qquad \forall i \in D, j \in S. \\
z_{ij} \geq 0 &  \qquad \forall (i,j) \in E.
\end{align*}
Let $(c_i: i \in D_1)$ be the values from \Cref{lem:decomp_vertex_weighted}. We define dual variables as follows:

\textbf{First stage.} 
For all $i \in D_1$, set $\alpha_i \gets c_i$. For all $j \in S$, set
$$
\beta_j \gets
    w_jx_j - \sum_i c_i{x}_{ij}.
$$

\textbf{Second stage.} By the discussion in \Cref{subsec:worst_case_vertex_weighted}, we may assume without loss of generality that the second stage graph is a matching. For all $i \in D_2$, set $\alpha_i \gets w_j(1-x_j)$, where $j$ is the unique neighbor of $i$ in the second stage graph. (This is exactly how much the edge contributes to the algorithm.) Leave the other dual variables unchanged.

\begin{claim}  \label{clm:pd_equal_vertex_weighted}
    The value of the algorithm is equal to the objective value of $(\bar\alphav, \bar\betav)$ in the dual.
\end{claim}

\begin{proof}{Proof.}
    Clearly the change in primal equals the change in the dual in the second stage. The claim in the first stage follows from the below equation, which holds for each $j \in S$:
    \begin{equation}
        \label{eq:eachj_vertex_weighted}
        \beta_j + \sum_{i \in D_1} \alpha_i x_{ij} = w_j x_j.
    \end{equation}
     This will suffice to prove the claim, since assuming \eqref{eq:eachj_vertex_weighted} holds, we have
    \begin{align*}
        \sum_{j \in S} \beta_j + \sum_{i \in D_1} \alpha_i 
        &= \sum_{j \in S} \beta_j + \sum_{i \in D_1} \alpha_i x_i &\text{(since if $x_i < 1$ then $\alpha_i = 0$ by Part 3 of \Cref{lem:decomp_vertex_weighted})}\\
    &= \sum_{j \in S} \left( \beta_j + \sum_{i \in D_1}\alpha_i x_{ij} \right) \\
    &= \sum_{j \in S} w_j x_j &\text{(by \eqref{eq:eachj_vertex_weighted})}
    \end{align*}
    Let us now prove \eqref{eq:eachj_vertex_weighted}. For $j \in S$, we have
    $$\beta_j + \sum_{i \in D_1} \alpha_i x_{ij} = w_jx_j - \sum_{i \in D_1} c_i x_{ij} + \sum_{i \in D_1} c_i x_{ij} = w_jx_j,$$
    as claimed.    \qed
\end{proof}



Our next claim is a lower bound on the sum of the dual variables across any given edge, and will be crucial in the analysis of robustness and consistency.

\begin{claim}
\label{clm:sumdualvar_vertex_weighted}
    For all $(i,j) \in E$, we have
    $$\beta_j \geq w_jx_jf_j(x_j),$$
    and
    $$\alpha_i + \beta_j \geq
    \begin{cases}
        w_j(1-f_j(x_j)+x_jf_j(x_j)), &\text{if $(i,j) \in E_1$,} \\
        w_j(1-x_j + x_jf_j(x_j)), &\text{if $(i,j) \in E_2$.}
    \end{cases}
    $$
\end{claim}

 \begin{proof}{\emph{Proof of \Cref{clm:sumdualvar_vertex_weighted}}.} 
 We prove the two parts separately. 

\underline{First part of claim.} We have
$$\beta_j = w_jx_j - \sum_i c_ix_{ij} \stackrel{(a)}{\geq} 
w_jx_j - \sum_i w_{j}(1-f_j(x_j))x_{ij}
= w_jx_jf_j(x_j),
$$
where $(a)$ is by Part 1 of \Cref{lem:decomp_vertex_weighted}.

\underline{Second part of claim.} First, consider $(i,j) \in E_1$. If $x_j = 1$, then 
\begin{align*}
\alpha_i + \beta_j &= c_i + \left(w_j - \sum_{k} c_kx_{kj}\right) \\
&\geq c_i + w_j - \sum_kc_ix_{kj} &\text{(by Part 2 of \Cref{lem:decomp_vertex_weighted})}\\
&= w_j
\end{align*}
as desired.  On the other hand, if $x_j < 1$, then 
\begin{align*}
\alpha_i + \beta_j 
\geq c_i + w_jx_jf_j(x_j) 
= c_i + w_jx_jf_j(x_j) 
\geq w_j(1-f_j(x_j))+ w_j x_jf_j(x_j),
\end{align*}
where the first inequality applies the lower bound on $\beta_j$ from the first part of the claim, and the final inequality is by Part 4 of \Cref{lem:decomp_vertex_weighted} (where $x_j < 1$).

Next, consider $(i,j) \in E_2$. Then $\alpha_i = w_j(1-x_j)$. Using  $\beta_j \geq w_jx_jf_j(x_j)$, we have
\begin{align*}
\alpha_i + \beta_j &\geq w_j(1-x_j) + w_jx_jf_j(x_j) 
\end{align*}
and so the claim holds. 
\qed
\end{proof}

\subsection{Bounding Robustness and Consistency}
\label{sec:analysis_rc_vertex_weighted}
Next, we use \Cref{clm:pd_equal_vertex_weighted} and \Cref{clm:sumdualvar_vertex_weighted} to show that if the penalty functions satisfy the requirements in \Cref{thm:vertex_weighted}, then the algorithm is $R$-robust and $C$-consistent.
\begin{claim}
\label{clm:robust_vertex_weighted}
    If we run Algorithm \ref{alg:vertex_weighted} with continuous, increasing functions $f_j: [0,1] \to [0,1]$ with $$1 - \frac{1-R}{x} \leq f_j(x) \leq \frac{1-R}{1-x}$$ for all $x \in (0, 1)$, then the algorithm is $R$-robust.
\end{claim}
\begin{proof}{\emph{Proof.}} 
    By \Cref{clm:pd_equal_vertex_weighted}, for any graph $G$ and advice $A$, we have $\ALG(G,A) = \sum_{i \in D} \alpha_i + \sum_{j \in S} \beta_j$. Therefore, to show $R$-robustness, we just need to show that $(\alpha, \beta)$ is $R$-approximately feasible to the dual, i.e. $\alpha_i + \beta_j \geq R\,w_j$ for all $(i,j) \in E$. 
    
    If $(i,j) \in E_1$, then 
    $$
        \alpha_i + \beta_j
        \stackrel{(a)}{\geq}   w_j(1-f_j(x_j)+x_jf_j(x_j))
        \stackrel{(b)}{\geq} R\,w_j
    $$
    where $(a)$ is by \Cref{clm:sumdualvar_vertex_weighted} and $(b)$ is because $f_j(x) \leq \frac{1-R}{1-x}$.

    On the other hand, if $(i,j) \in E_2$, then 
    $$\alpha_i + \beta_j
        \stackrel{(a)}{\geq} w_j(1-x_j + x_jf_j(x_j)) 
        \stackrel{(b)}{\geq} R\,w_{j}$$
    where $(a)$ is by \Cref{clm:sumdualvar_vertex_weighted} and $(b)$ is because $f_j(x) \geq 1 - \frac{1-R}{x}$. \qed
\end{proof}

\begin{claim}
    \label{clm:consistency_vertex_weighted}
    If we run Algorithm \ref{alg:vertex_weighted} with continuous, increasing functions $f_j: [0,1] \to [0,1]$ with
    \begin{itemize}
        \item $f_j(x) \leq \frac{1-C}{1-x}$ if  $j$ is matched by the advice, and
        \item $f_j(x) \geq 1 - \frac{1-C}{x}$ if $j$ is not matched by the advice,
    \end{itemize} then the algorithm is $C$-consistent . 
\end{claim}
\begin{proof}{\emph{Proof.}} 
By our discussion in \Cref{subsec:worst_case_vertex_weighted}, the worst case for consistency is when the second-stage graph is a matching (call it $M_2$), consisting of exactly the edges selected by $\ADVICE(G, A)$ in the second stage. Let $M_1$ denote the first-stage suggested matching. Then
$$\ADVICE = \sum_{(i,j) \in M_1} w_j + \sum_{(i,j) \in M_2} w_j.$$
On the other hand, using \Cref{clm:pd_equal_vertex_weighted}, we have
\begin{align*}
\ALG &\geq \sum_{(i,j) \in M_1} (\alpha_i + \beta_j) + \sum_{(i,j) \in M_2} (\alpha_i + \beta_j).
\end{align*}
Thus, to show $\ALG \geq C\cdot \ADVICE$, it suffices to show 
\begin{equation}
\label{eq:sumdualvar_consistency_vertex_weighted}
    \alpha_i + \beta_j \geq C\cdot w_j ~\text{for all $(i,j) \in M_1 \cup M_2$.}
\end{equation}
We first show \eqref{eq:sumdualvar_consistency_vertex_weighted} for $(i,j) \in M_1$. For $(i,j) \in M_1$, we have
\begin{align*}
    \alpha_i + \beta_j \geq w_j(1-f_j(x_j) + x_jf_j(x_j)) \geq C \cdot w_j, 
\end{align*}
where the first inequality is by \Cref{clm:sumdualvar_vertex_weighted} and the second inequality is because $f_j(x) \leq \frac{1-C}{1-x}$ for vertices $j$ suggested by the advice.

On the other hand, for $(i,j) \in M_2$ we have
$$\alpha_i + \beta_j \geq w_j(1-x_j + x_jf_j(x_j)) \geq C \cdot w_j,$$
where the first inequality is by \Cref{clm:sumdualvar_vertex_weighted} and the second inequality is because $f_j(x) \geq 1 -\frac{1-C}{x}$ for vertices $j$ not suggested by the advice. 
\qed
\end{proof}

\section{Adwords and Fractional Advice}
\label{sec:alg_adwords}

\begin{algorithm}[t]
\caption{Two-Stage Fractional Adwords with Advice}
\label{alg:main}
\DontPrintSemicolon
\KwIn{Suggested allocation $\mathbf{a}$ in the first-stage graph and desired robustness level $R$.}
\begin{algorithmic}[1]
\STATE (Define penalty functions) For each $j \in S$, define $f_j: [0,1] \to [0,1]$ as follows:\footnotemark 
$$\text{If $a_j \leq 0.5$,} \quad f_j(s) = 
\begin{cases}
    \min\left\{1, \frac{1-R}{1-s}, \frac{a_j(1-C)}{a_j - s} \right\}&\text{if $s < a_j$,} \\
    \min\left\{1, \frac{1-R}{1-s} \right\} &\text{if $s \geq  a_j$.}
\end{cases}
$$
$$\text{If $a_j > 0.5$,} \quad f_j(s) = 
\begin{cases}
    \max\left\{0, 1 - \frac{1-R}{s}\right\}&\text{if $s \leq a_j$,} \\
    \max\left\{0, 1-\frac{1-R}{s}, 1 - \frac{1-C}{s-a_j} \right\} &\text{if $s >  a_j$.}
\end{cases}
$$
\STATE (First stage) 
Solve the following optimization problem for the first-stage allocation $\vx$:
\begin{align*}
    (P_1) \quad \max \quad &\sum_{j\in S}B_j\left(x_j-\int_0^{x_j} f_j(t) dt\right) \\
    \mathrm{s.t.} \quad  x_i~:=~&\sum_{j \in S} x_{ij}\leq 1 &  \forall \; i \in D_1 \\
    x_j~:=~&\frac{1}{B_j}\sum_{i\in D_1} b_{ij}x_{ij} \leq 1 &  \forall \; j \in S \\
    &x_{ij} \geq 0 & \forall \; i\in D_1,\,(i,j) \in E
\end{align*}



\STATE (Second stage) When the second-stage vertices $D_2$ arrive, solve for the optimal allocation $\yv$ subject to the capacities already taken by $\vx$:
\begin{align*}
    (P_2) \quad \max \quad &\sum_{j \in S} B_jy_j \\
    \mathrm{s.t.\ } y_i~:=~&\sum_{j \in S} y_{ij} \leq 1 &  \forall \; i \in D_2 \\
    y_j~:=~&\frac{1}{B_j}\sum_{i \in D_2} b_{ij}y_{ij} \leq 1 - {x}_j&  \forall \; j \in S \\
    &y_{ij} \geq 0 & \forall \; i \in D_2, \, (i,j) \in E
\end{align*}

\STATE {Return} $\vx + \yv$. 
\end{algorithmic}
\end{algorithm}

\footnotetext{As it will turn out, there are many choices for these penalty functions that will all give the optimal robustness-consistency tradeoff; see \Cref{thm:main_adwords}. Here we make a particular choice for these penalty functions in order to have a well-defined algorithm.}

In the previous sections, we described our algorithm for two-stage vertex-weighted bipartite matching {with integral advice} and analyzed its robustness-consistency tradeoff. We now show that the same ideas can be extended to the setting of two-stage Adwords, and when the advice can be fractional. Our algorithm for two-stage Adwords is described in Algorithm~\ref{alg:main}. At this point, it may be helpful to recall the notation we use in the Adwords setting, which is described in \Cref{subsec:notation_adwords}.

Much of the intuition behind Algorithm \ref{alg:main} is the same as in the vertex-weighted setting. We associate an advice-dependent penalty function to each offline vertex. These penalty functions are lower for vertices that are recommended by the advice and higher otherwise, which incentivizes the algorithm to to prioritize following the advice. By varying how "extreme" these penalty functions are (where least extreme corresponds to the fully robust algorithm and most extreme corresponds to the fully consistent algorithm), we obtain the Pareto-optimal tradeoff curve between robustness and consistency. The main difference now is that the advice can be fractional, so it no longer suffices to only consider two possible penalty functions $\fl$ and $\fu$. Specifically, the penalty function $f_j$ will now depend on $a_j$, the fraction that $j$ is filled under the advice. If $a_j$ is higher, then $f_j$ will be lower, and vice-versa. For any desired robustness level $R$, our analysis characterizes the envelope of penalty functions which guarantee the algorithm to be $R$-robust. Similarly, for any desired consistency level $C$, we characterize the envelope of penalty functions which guarantee the algorithm to be $C$-consistent. Finally, we show the intersection of these two envelopes is non-empty if $\sqrt{1-R} + \sqrt{1-C} \geq 1$,\footnote{This tradeoff is tight; see \Cref{sec:tight} for a lower bound.} which will give our main result. 

\begin{restatable}{theorem}{main}
\label{thm:main_adwords}
Let $R \in [0, \frac{3}{4}]$ and let $C \in [0,1]$. For each $j \in S$, let $a_j \in [0,1]$ be the fraction that $j$ is filled by the advice. Let $f_j$ denote the penalty functions used in Algorithm \ref{alg:main}, where $f_j: [0,1] \to [0,1]$ is continuous and increasing. Then, the following hold.
\begin{enumerate}
    \item (Robustness) Suppose for all $j \in S$, we have $1 - \frac{1-R}{x} \leq f_j(x) \leq \frac{1-R}{1-x}$ for all $x \in (0, 1)$. Then the algorithm is $R$-robust.
    \item (Consistency) Suppose for all $j \in S$, we have $f_j(x) \leq \frac{a_j(1-C)}{a_j-x}$ for all $x \in [0, a_j)$, and $f_j(x) \geq 1 - \frac{1-C}{x-a_j}$ for all $x \in (a_j, 1]$. Then the algorithm is $C$-consistent.
\end{enumerate}
Moreover, if $\sqrt{1-R} + \sqrt{1-C} \geq 1$, there exist penalty functions $f_j$ that simultaneously satisfy both conditions above.
\end{restatable}

The proof of \Cref{thm:main_adwords} is in \Cref{sec:adwords_proofs}. The first condition in the theorem describes the envelope of  penalty functions which guarantee $R$-robustness, and the second condition describes the envelope of  penalty functions which guarantee $C$-consistency. Intuitively, the robustness condition says that to guarantee robustness, the penalty functions cannot be too extreme. In the consistency condition, the upper bound on $f_j(x)$ for $x < a_j$ says that one should not overly penalize allocating to $j$ if $x_j < a_j$. Similarly, the lower bound on $f_j(x)$ for $x > a_j$ says that one should be penalized at least a certain amount for exceeding the amount that $j$ is filled under the advice. 
Note that the envelope for consistency is node-dependent: it
depends on the fraction $a_j$ that $j$ is filled under the advice, whereas the envelope for robustness does not. This makes sense, because the definition of $R$-robustness is advice-agnostic.

\begin{figure}[t!]
    \centering
        \begin{subfigure}{0.3\textwidth}
    \includegraphics[width=\linewidth]{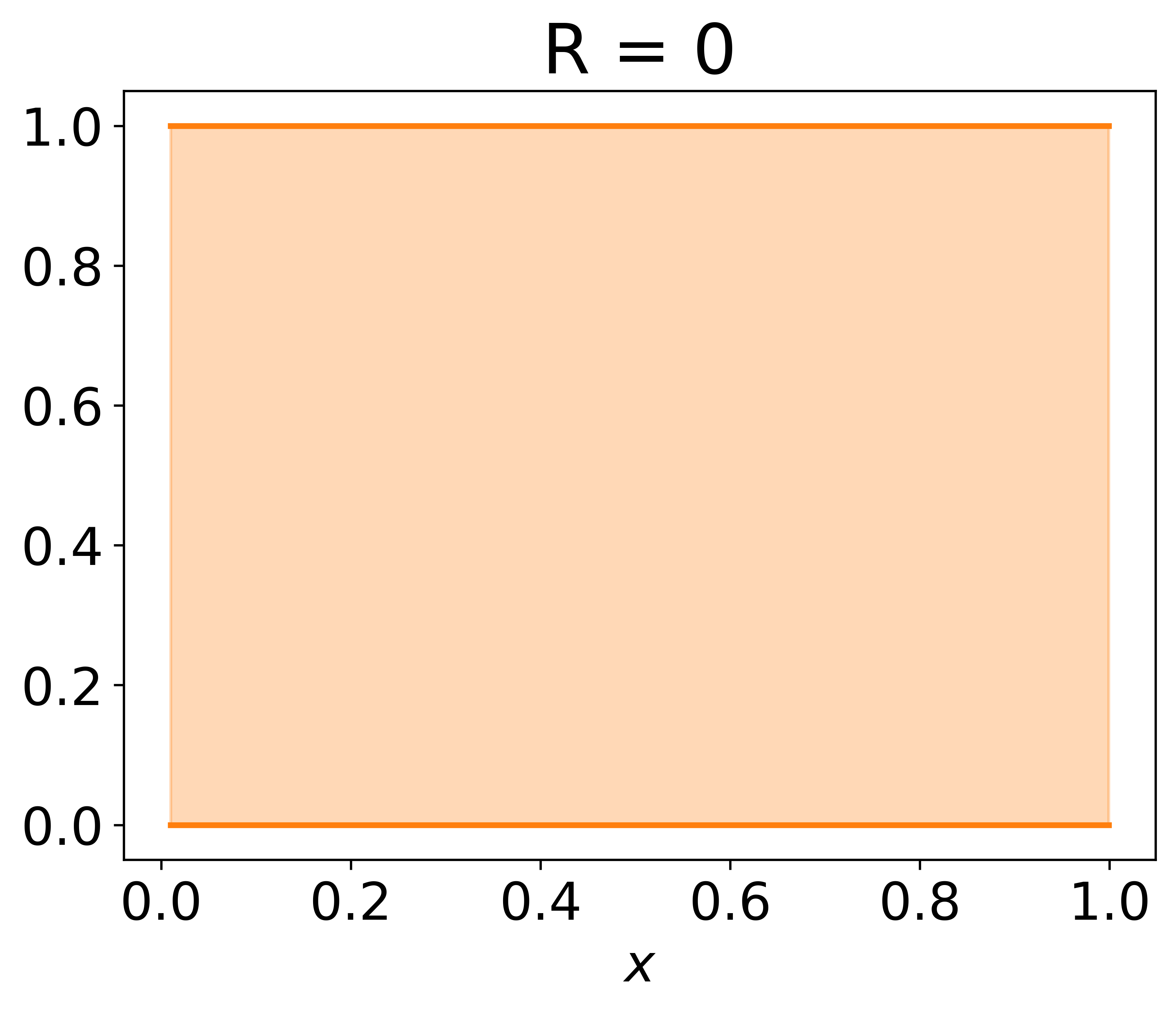}
    \end{subfigure}
    \hfill
    \begin{subfigure}{0.3\textwidth}
        \includegraphics[width=\linewidth]{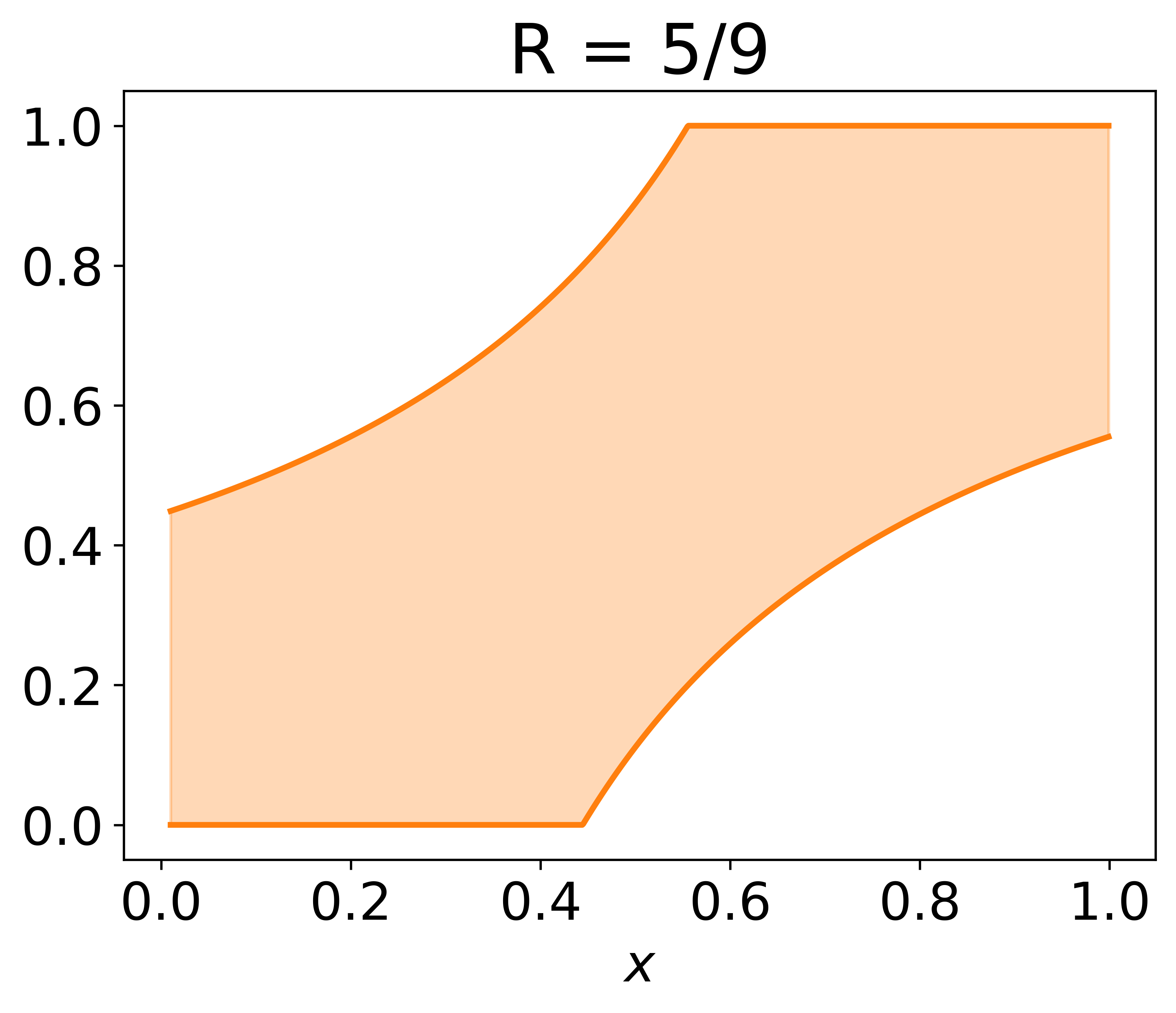}
    \end{subfigure}
    \hfill
    \begin{subfigure}{0.3\textwidth}
        \includegraphics[width=\linewidth]{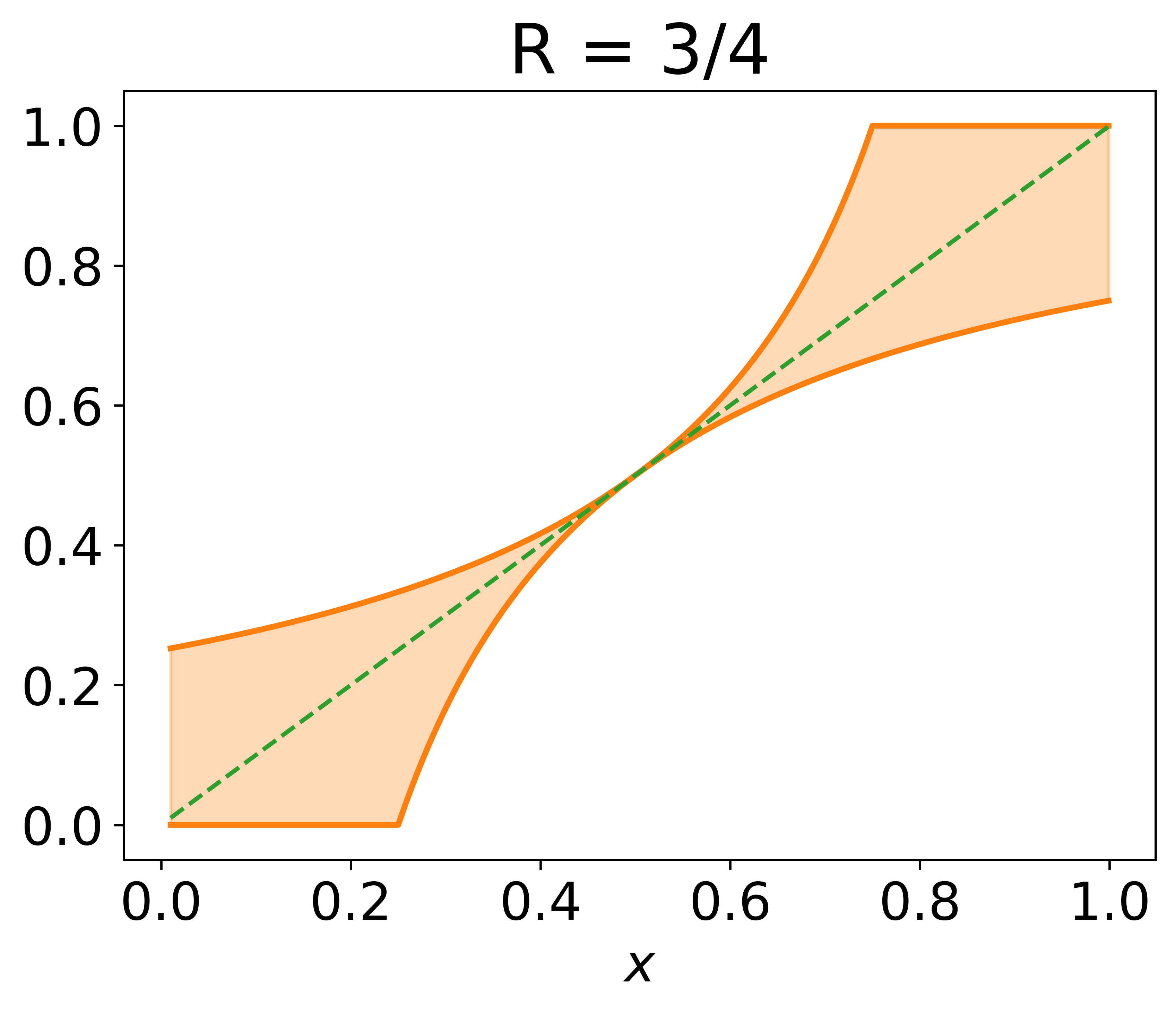}
    \end{subfigure}
    \caption{Plots of the envelope for $R$-robustness for three values of $R$. Running Algorithm \ref{alg:main} with any increasing, continuous penalty functions in the envelope guarantees $R$-robustness. In the right plot, the green dashed line indicates the penalty function used by \citet{feng2021two}.}
    \label{fig:envelope_r}
\end{figure}

\begin{figure}[t!]
    \centering
    \begin{subfigure}{0.3\textwidth}
        \includegraphics[width=\linewidth]{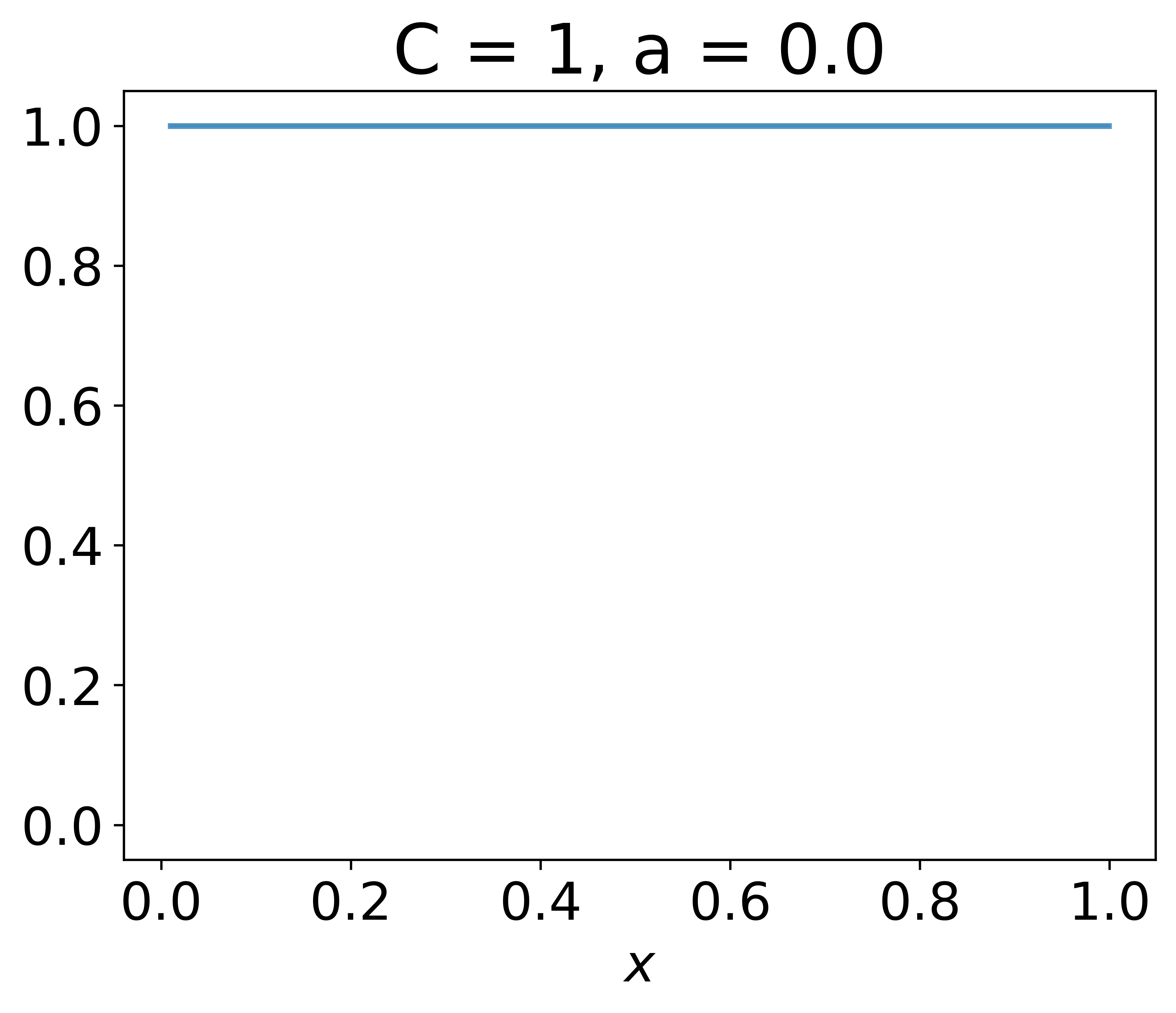} 
        \label{fig:sub1}
    \end{subfigure}
    \hfill
    \begin{subfigure}{0.3\textwidth}
        \includegraphics[width=\linewidth]{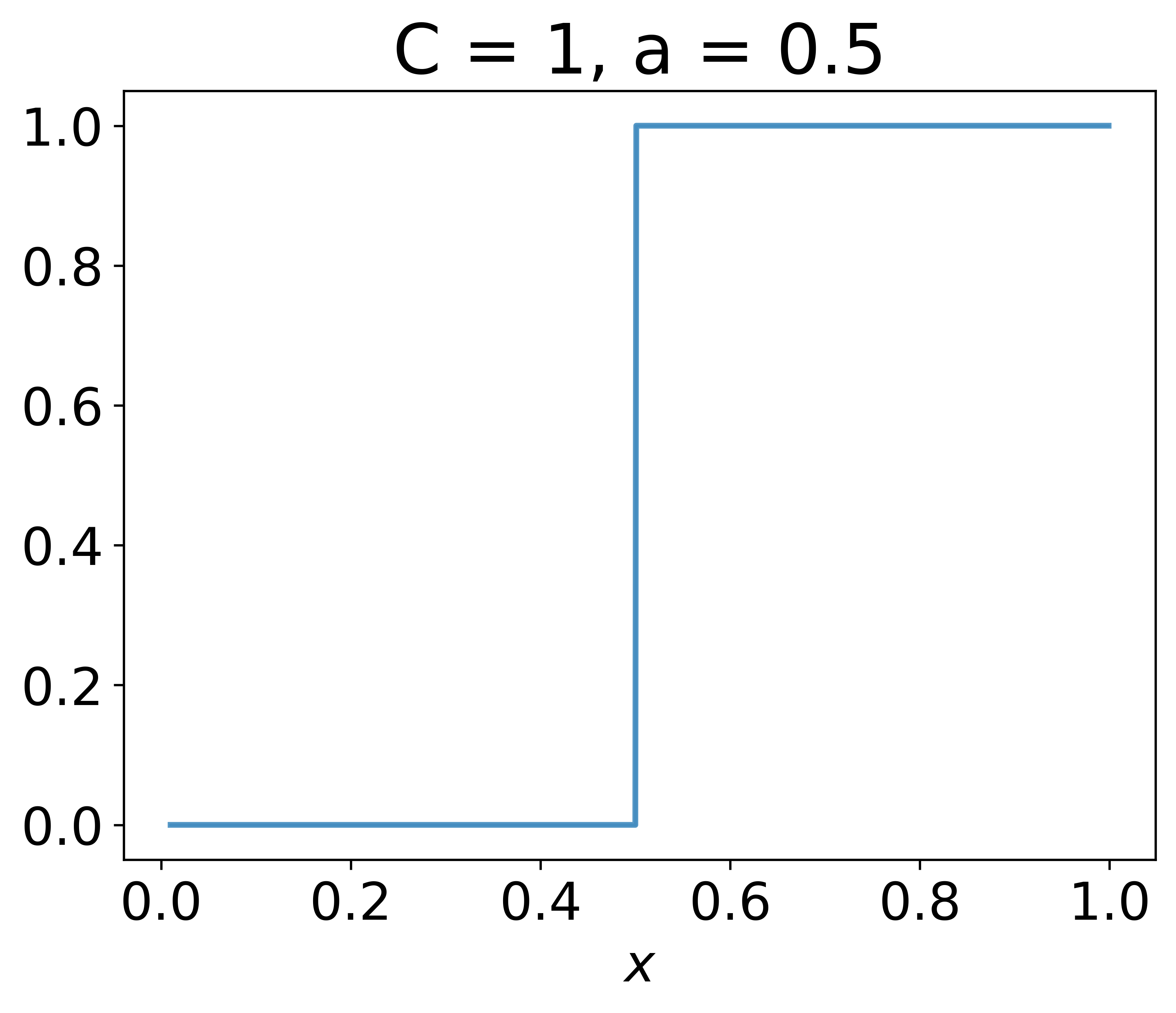}
        \label{fig:sub2}
    \end{subfigure}
    \hfill
    \begin{subfigure}{0.3\textwidth}
        \includegraphics[width=\linewidth]{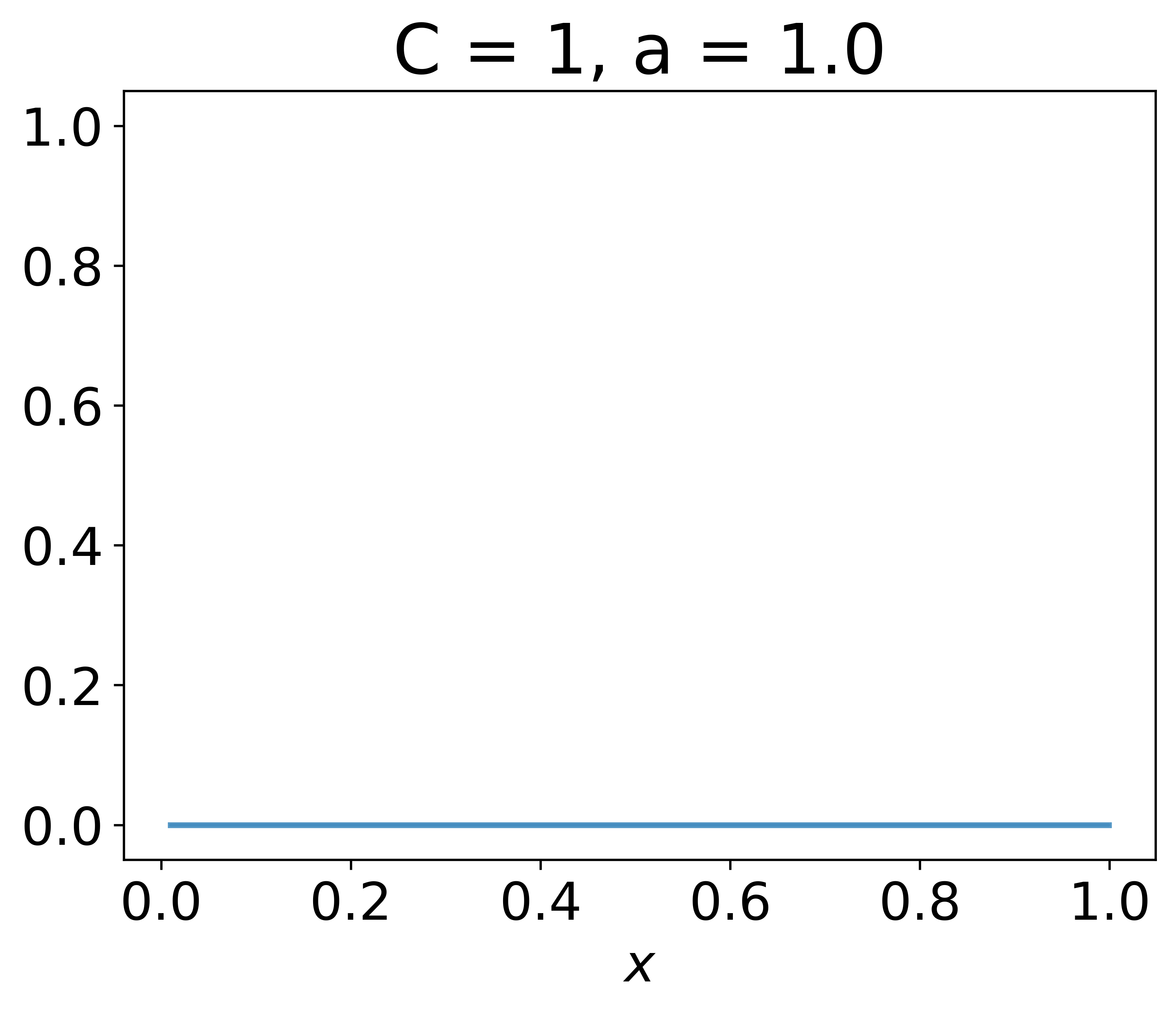}
        \label{fig:sub3}
    \end{subfigure}

    \medskip

    \begin{subfigure}{0.3\textwidth}
        \includegraphics[width=\linewidth]{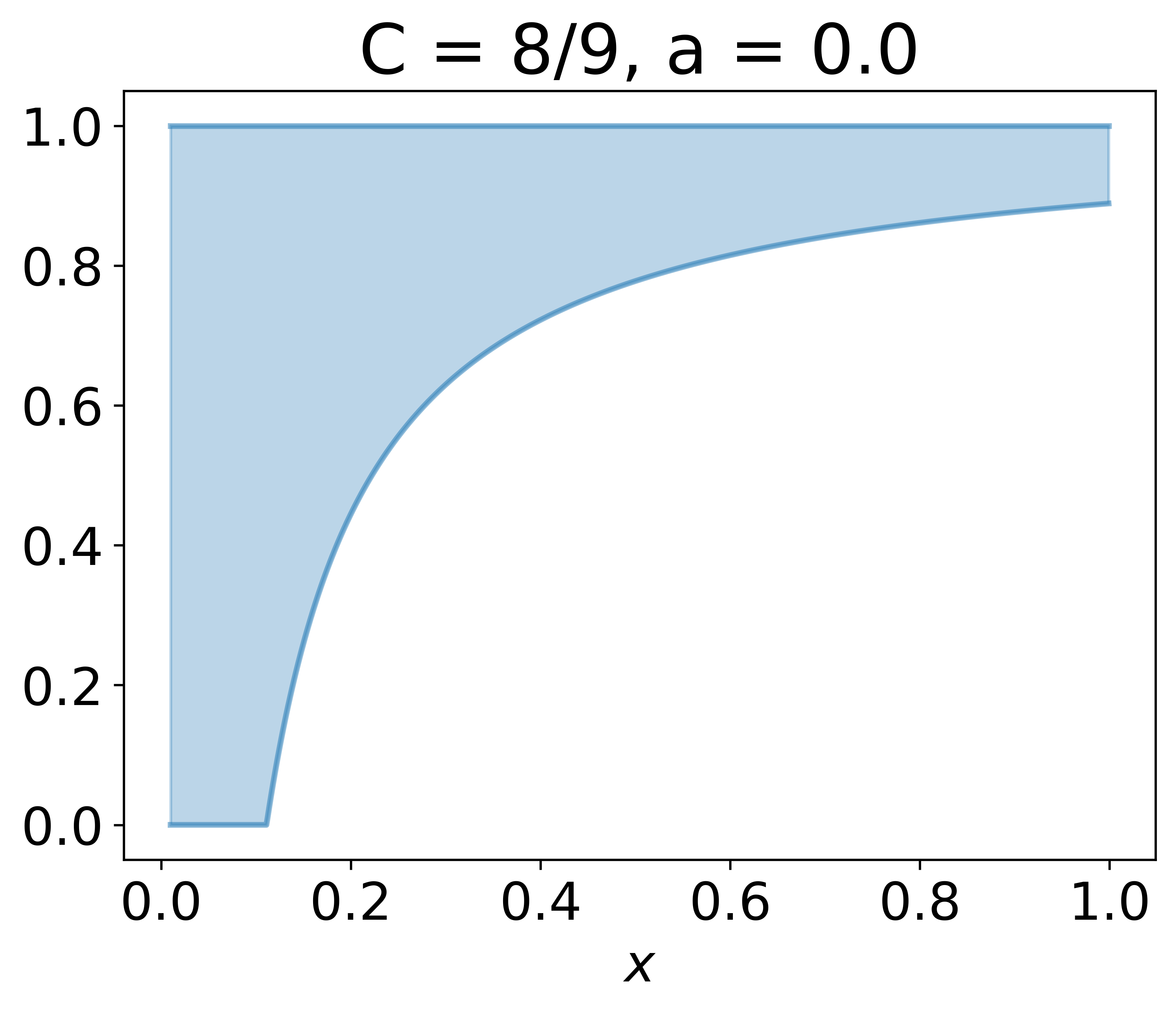}
        \label{fig:sub4}
    \end{subfigure}
    \hfill
    \begin{subfigure}{0.3\textwidth}
        \includegraphics[width=\linewidth]{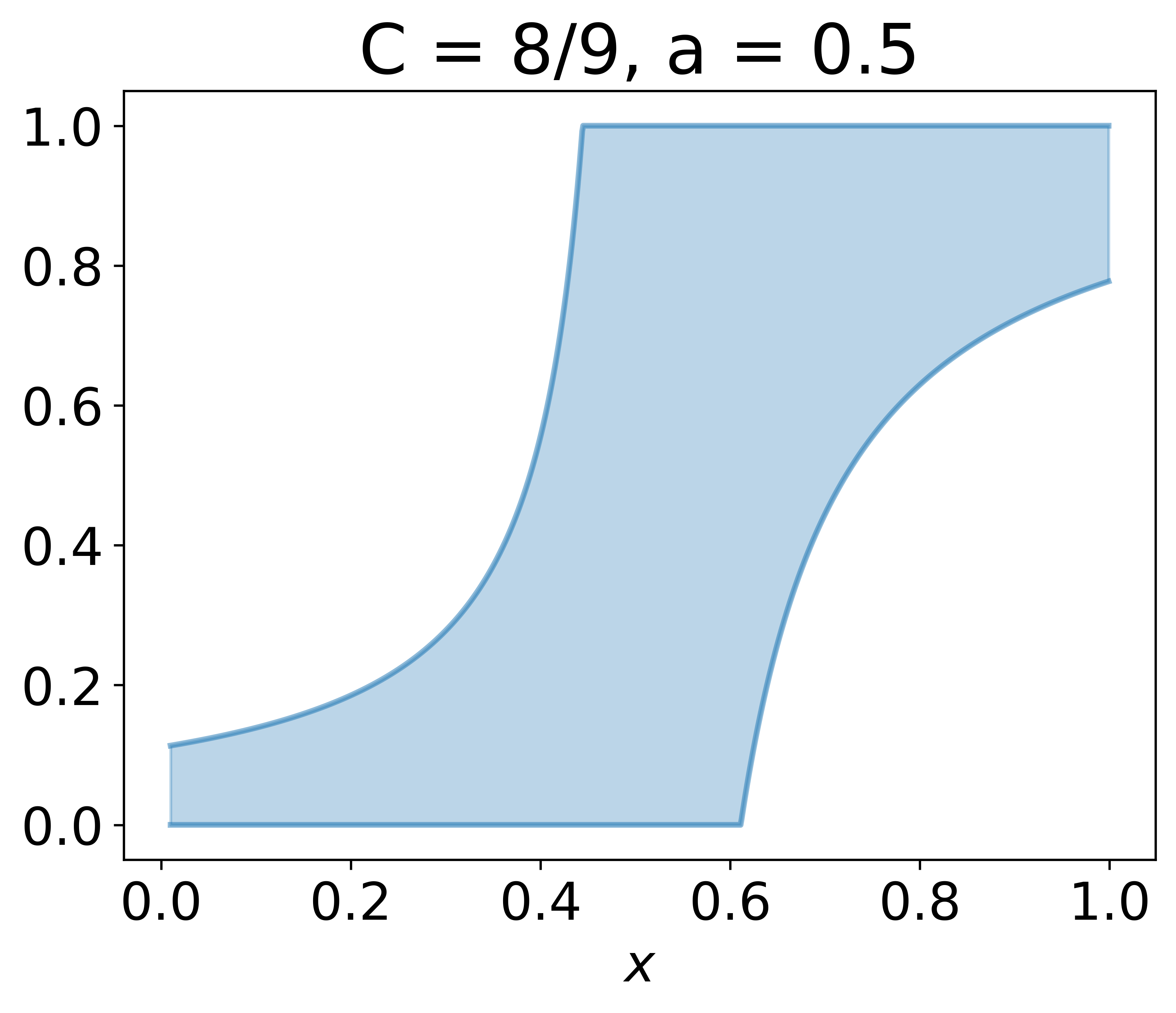}
        \label{fig:sub5}
    \end{subfigure}
    \hfill
    \begin{subfigure}{0.3\textwidth}
        \includegraphics[width=\linewidth]{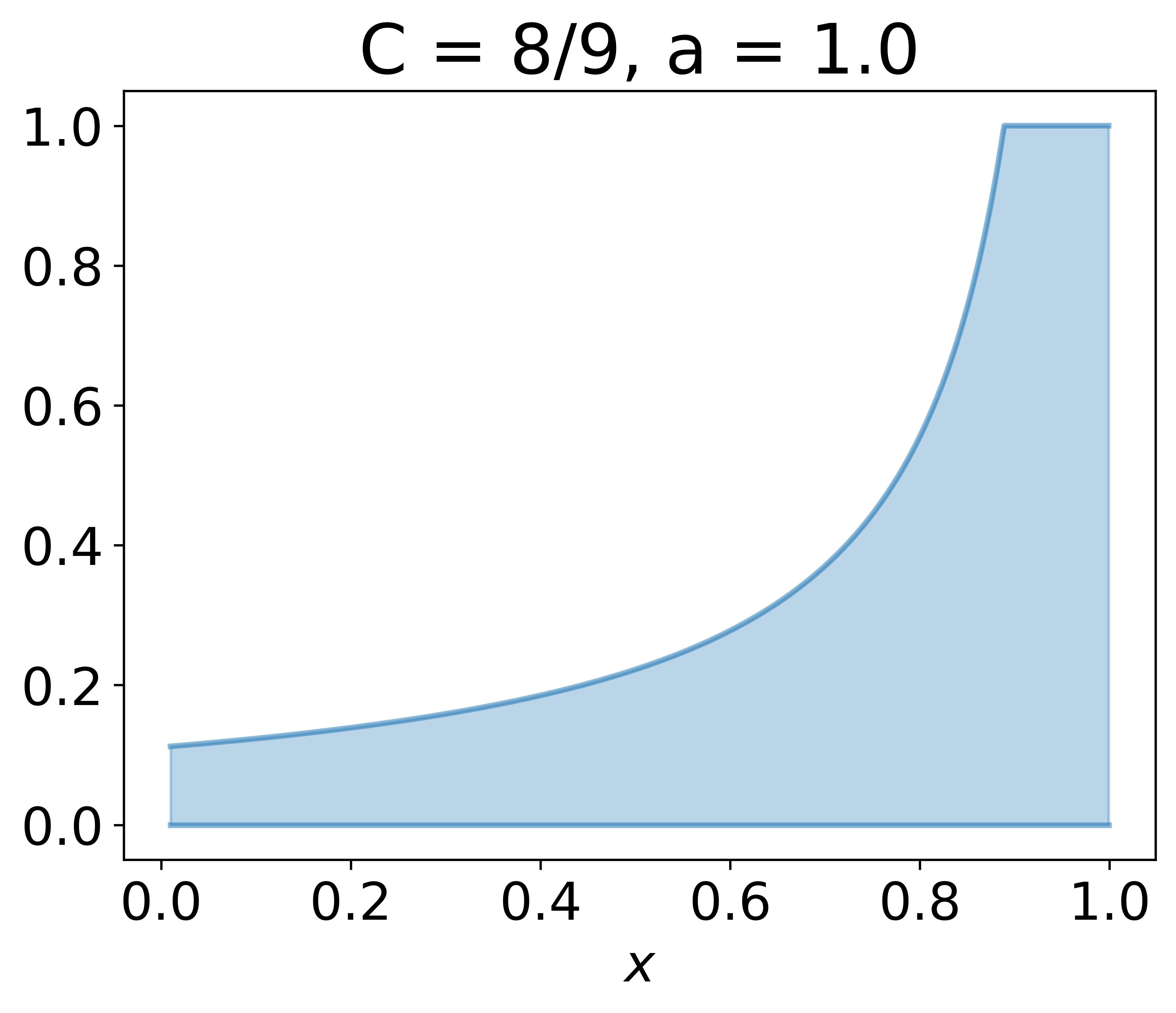}
        \label{fig:sub6}
    \end{subfigure}

    \medskip

    \begin{subfigure}{0.3\textwidth}
        \includegraphics[width=\linewidth]{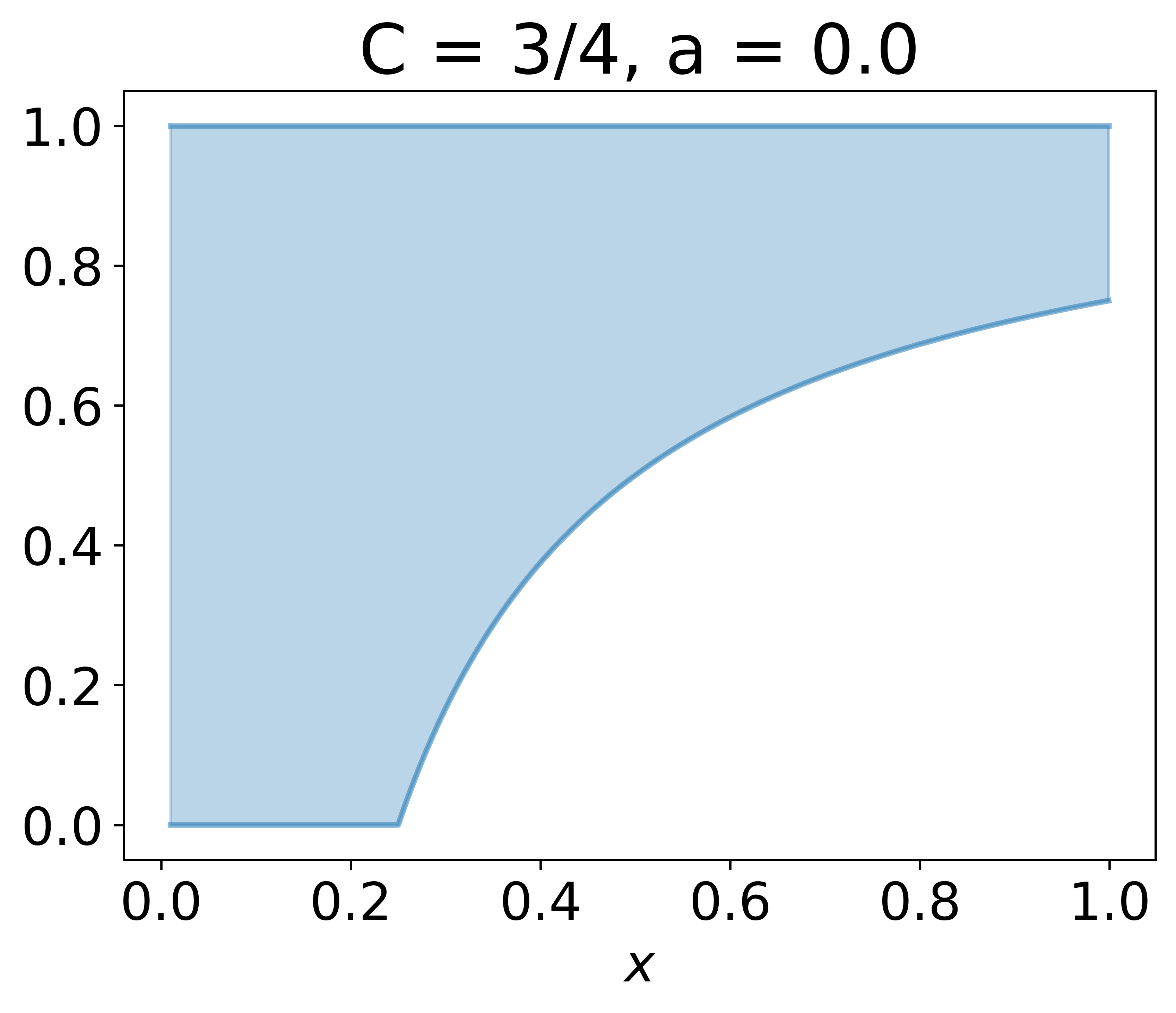}
        \label{fig:sub7}
    \end{subfigure}
    \hfill
    \begin{subfigure}{0.3\textwidth}
        \includegraphics[width=\linewidth]{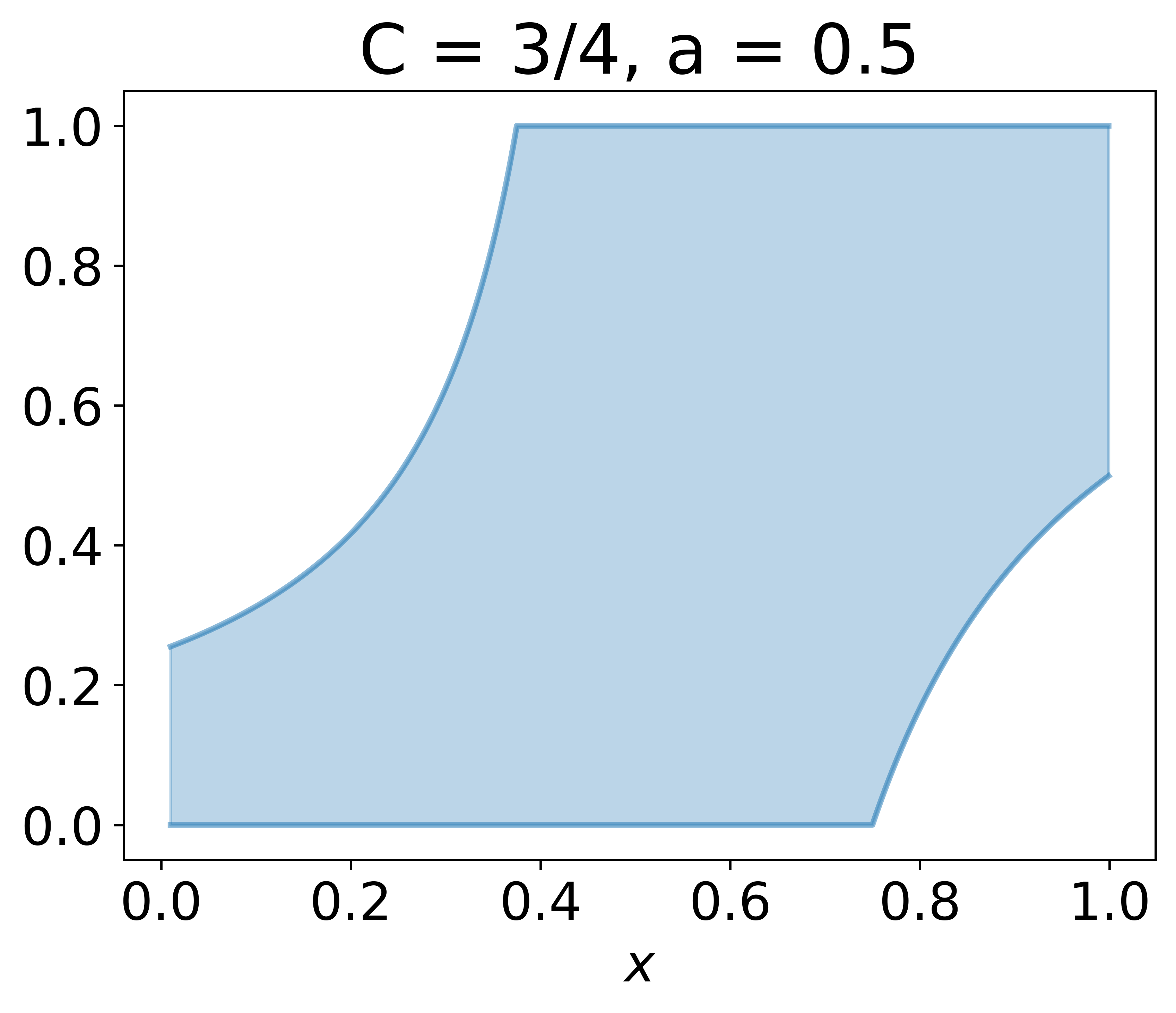}
        \label{fig:sub8}
    \end{subfigure}
    \hfill
    \begin{subfigure}{0.3\textwidth}
        \includegraphics[width=\linewidth]{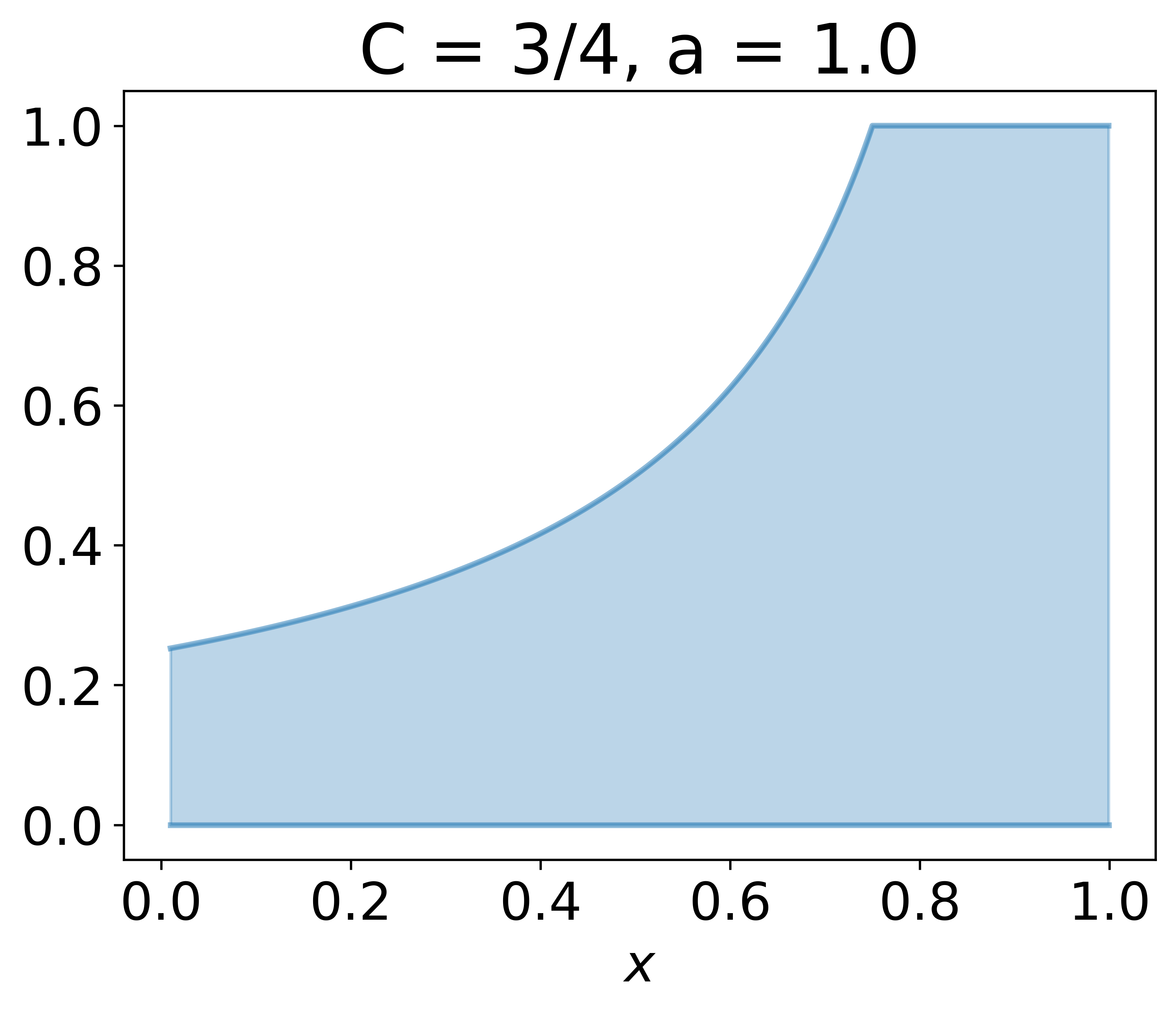}
        \label{fig:sub9}
    \end{subfigure}
    \caption{Plots of the envelope for $C$-consistency for three values of $C$ and three values of $a$, where $a$ is the fraction the offline node is filled under the advice.}
    \label{fig:envelope_c}
\end{figure}

\Cref{fig:envelope_r} plots the envelope for robustness for three values of $R$. The parameter $R\in[0,\frac{3}{4}]$ affects the separation between the lower and upper bounds of the envelope.
If $R=0$, then the lower and upper bounds are the constant functions $0,1$ respectively, so the envelope places no restrictions on the penalty functions. This makes sense, because any algorithm is 0-robust.
On the other extreme, setting $R=\frac{3}{4}$ guarantees the maximum possible robustness.
Surprisingly, even in this case there is a bit of separation between the upper and lower bounds of the envelope, showing that there are many penalty functions which all guarantee the maximum possible robustness.
We plot one intermediate value where $R=\frac{5}{9}$.

Next, \Cref{fig:envelope_c} plots the envelope for consistency for three values of $C$ and three values of $a$, where $a$ is the fraction the offline node is filled by the advice. The three values of $C \in \{\frac34, \frac89, 1\}$ were chosen to satisfy $\sqrt{1-R} + \sqrt{1-C} = 1$, for the three values of $R$ used in \Cref{fig:envelope_r}. The larger the value of $C$, the more restrictive the envelope for $C$-consistency becomes. For the largest possible value of $C = 1$, the envelope consists of a single function, which is the step function that switches from 0 to 1 at $x = a$. Using these step functions causes Algorithm \ref{alg:main} to follow the advice exactly in the first stage, which is indeed 1-consistent.  We overlay the envelopes for robustness and consistency in \Cref{fig:envelope_both}.

\begin{figure}[ht!]
    \centering
    \begin{subfigure}{0.3\textwidth}
        \includegraphics[width=\linewidth]{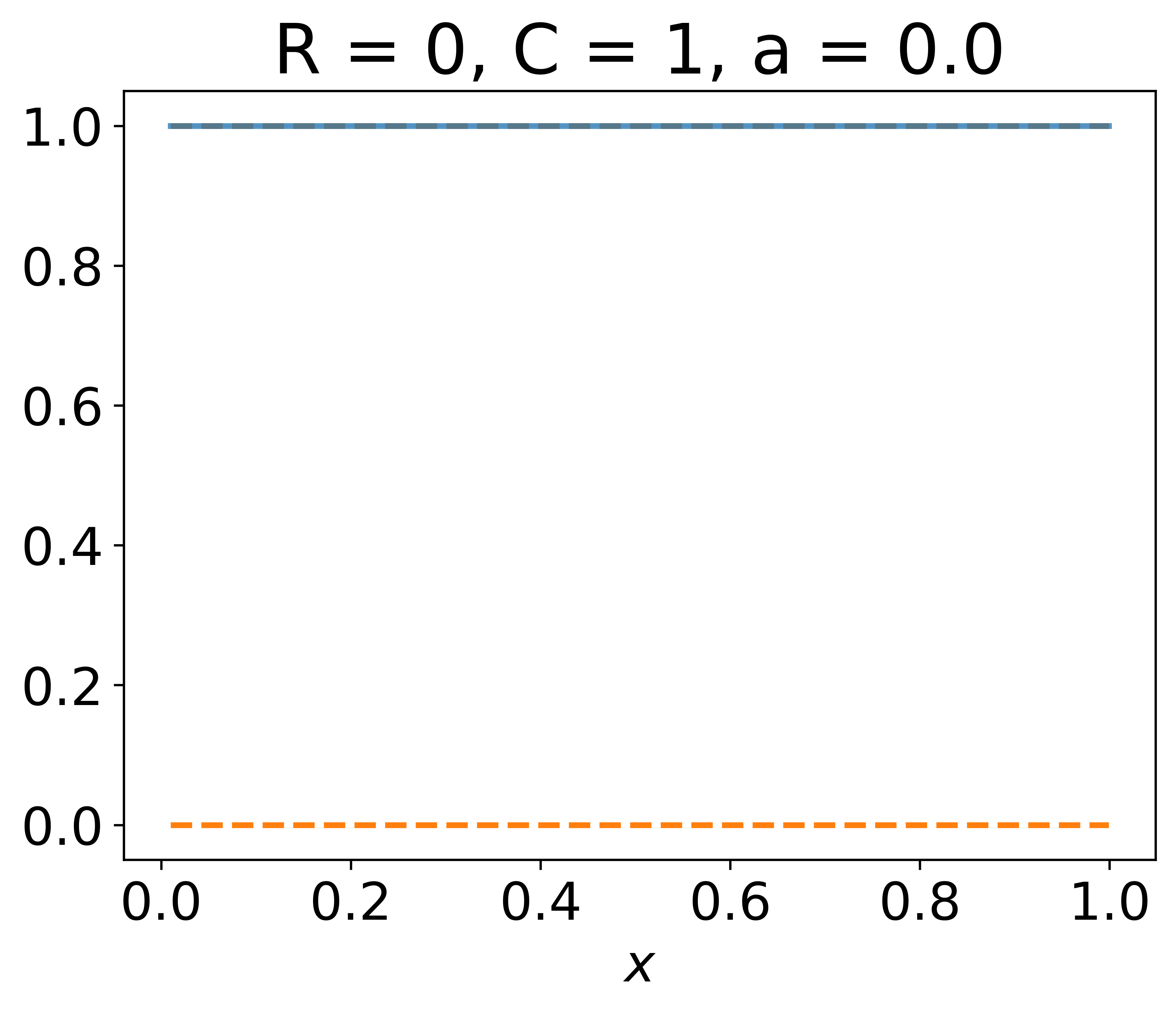} 
    \end{subfigure}
    \hfill
    \begin{subfigure}{0.3\textwidth}
        \includegraphics[width=\linewidth]{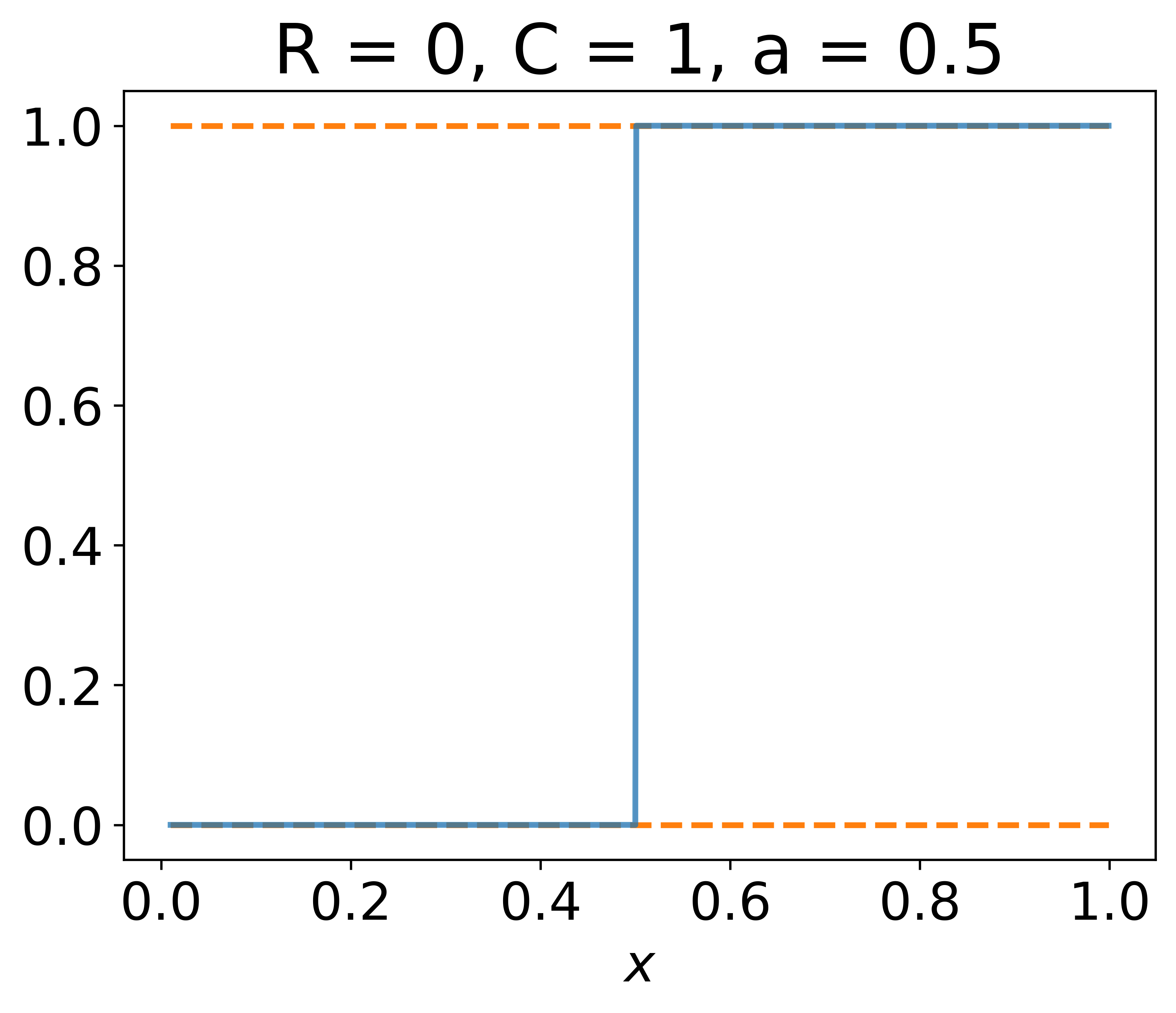}
    \end{subfigure}
    \hfill
    \begin{subfigure}{0.3\textwidth}
        \includegraphics[width=\linewidth]{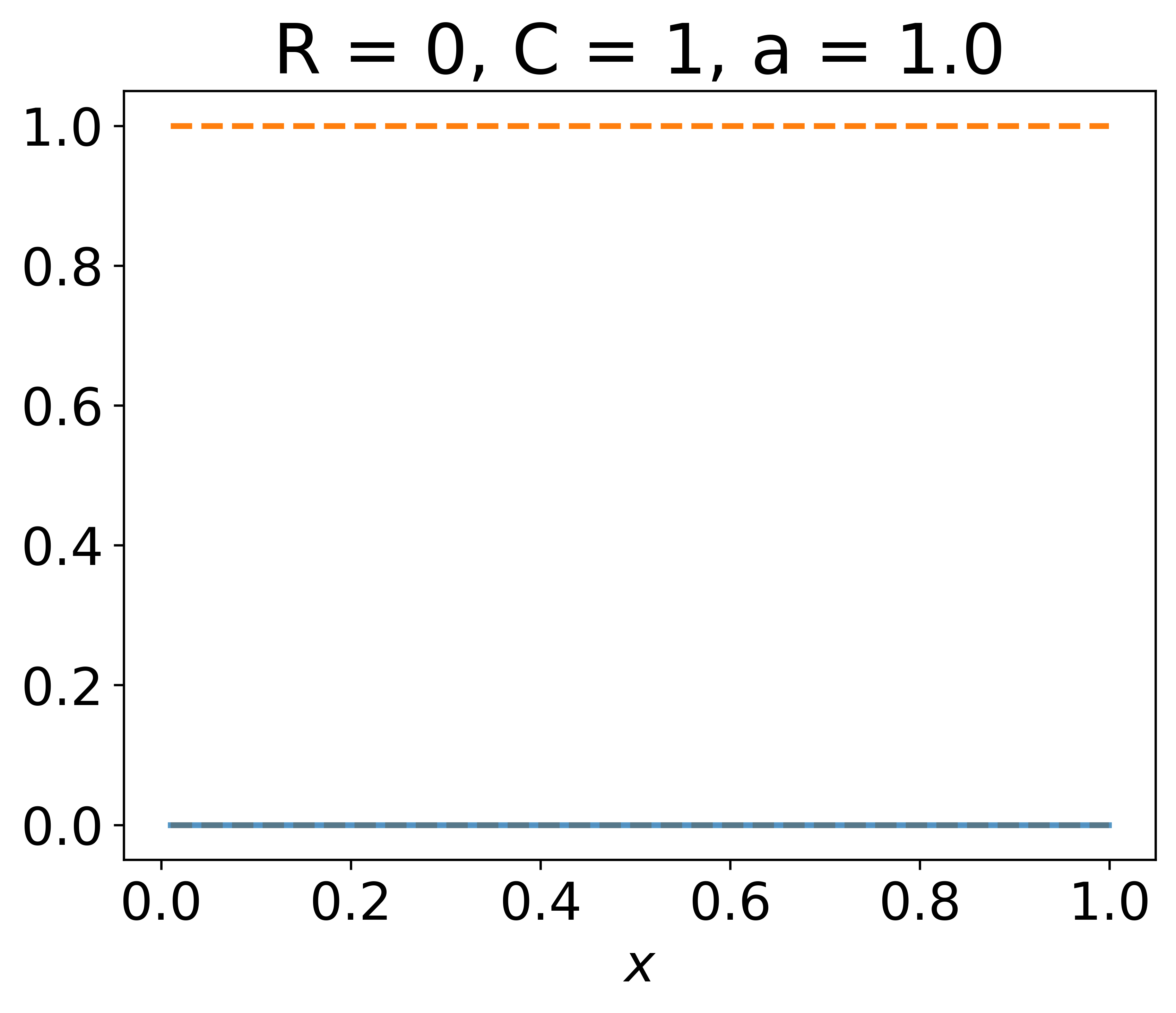}
    \end{subfigure}

    \medskip

    \begin{subfigure}{0.3\textwidth}
        \includegraphics[width=\linewidth]{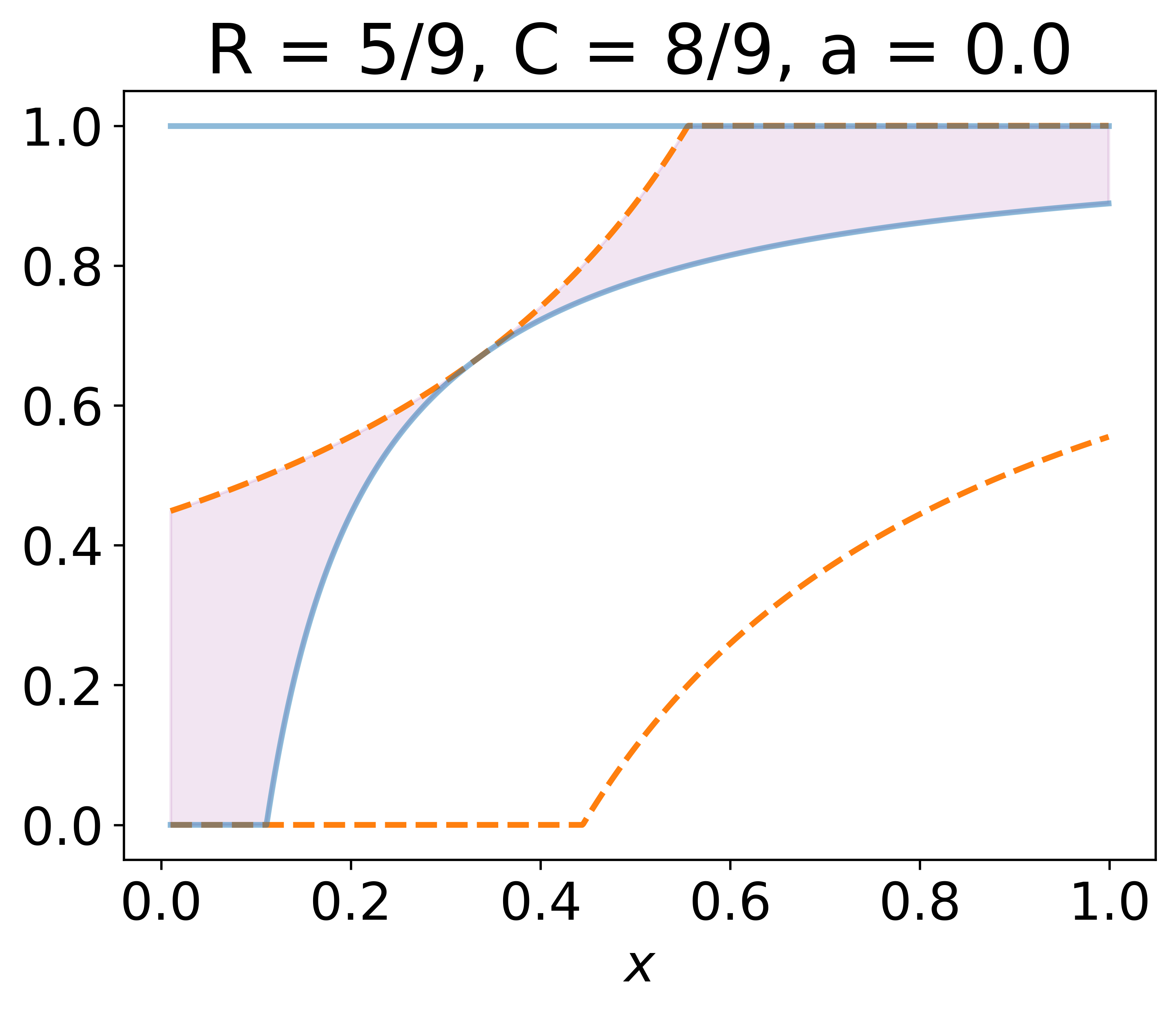}
    \end{subfigure}
    \hfill
    \begin{subfigure}{0.3\textwidth}
        \includegraphics[width=\linewidth]{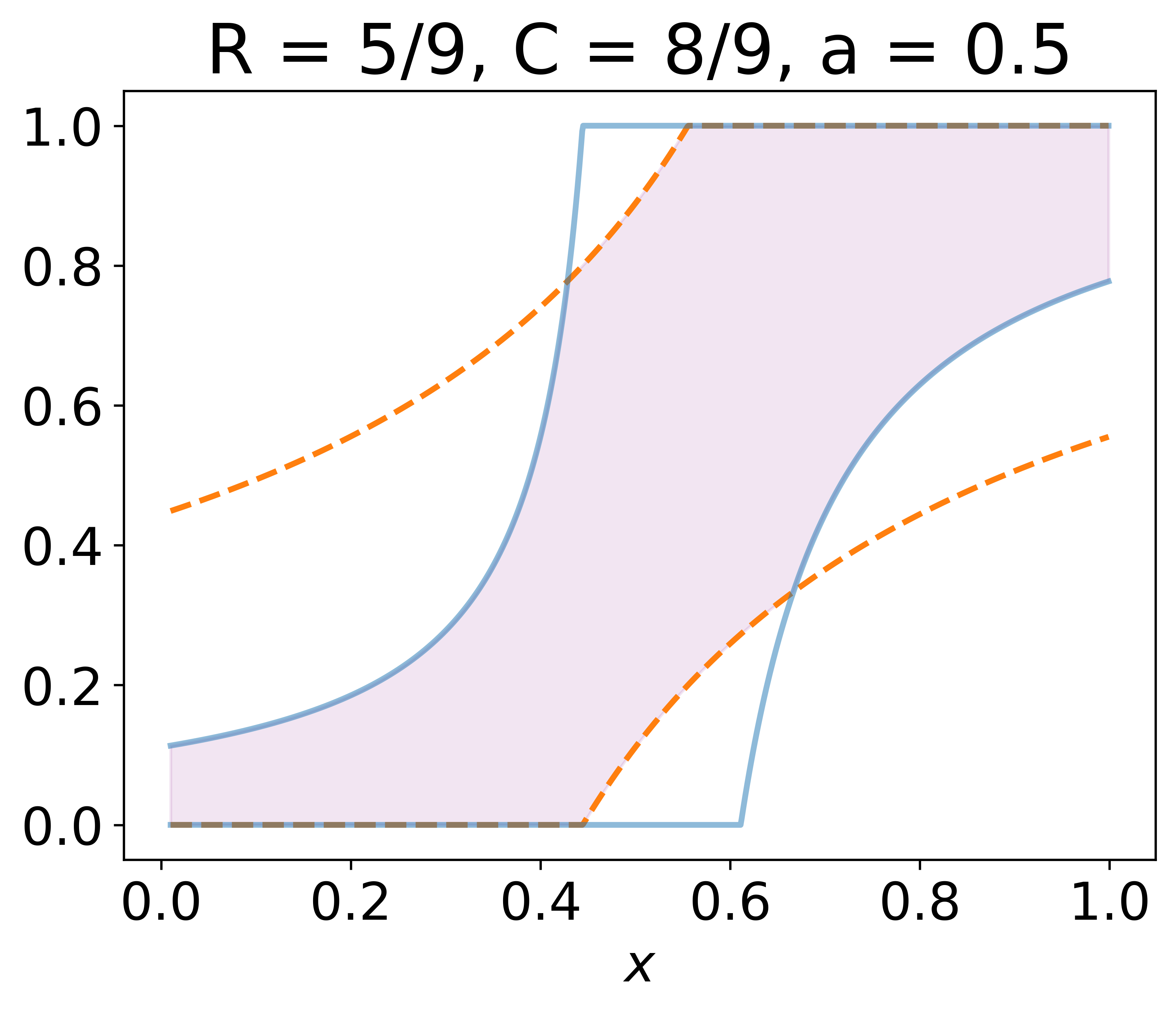}
    \end{subfigure}
    \hfill
    \begin{subfigure}{0.3\textwidth}
        \includegraphics[width=\linewidth]{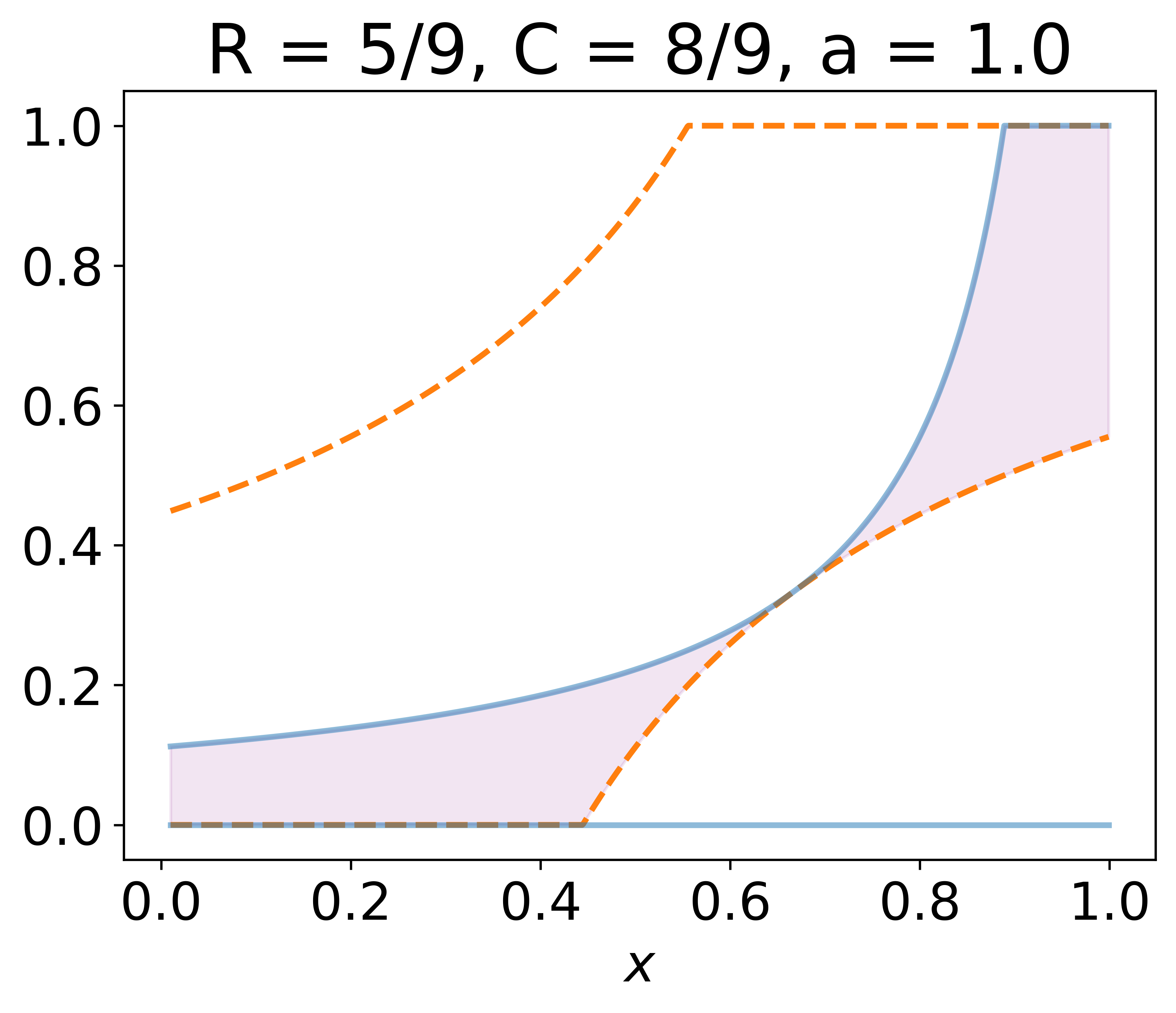}
    \end{subfigure}

    \medskip

    \begin{subfigure}{0.3\textwidth}
        \includegraphics[width=\linewidth]{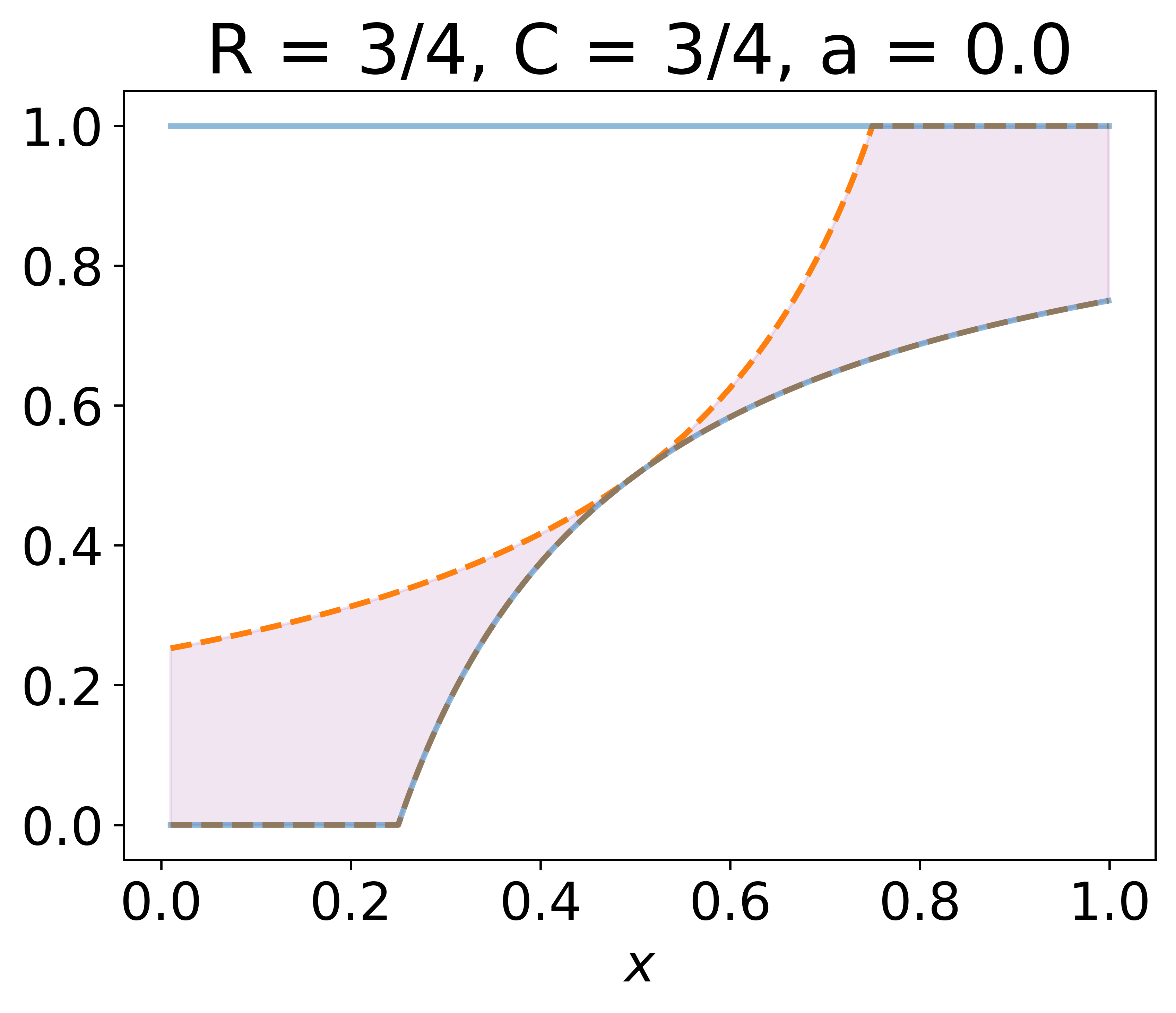}
    \end{subfigure}
    \hfill
    \begin{subfigure}{0.3\textwidth}
        \includegraphics[width=\linewidth]{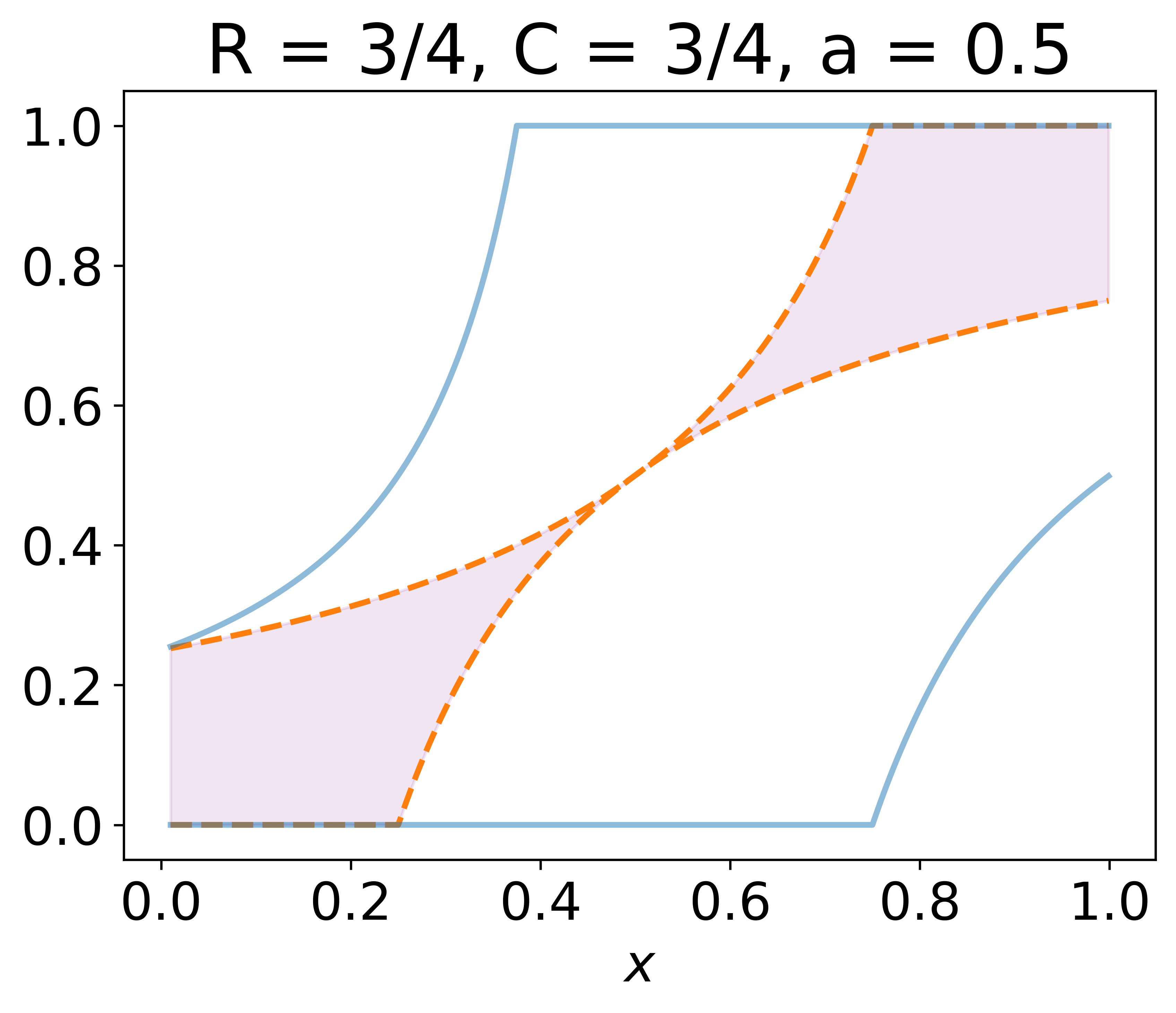}
    \end{subfigure}
    \hfill
    \begin{subfigure}{0.3\textwidth}
        \includegraphics[width=\linewidth]{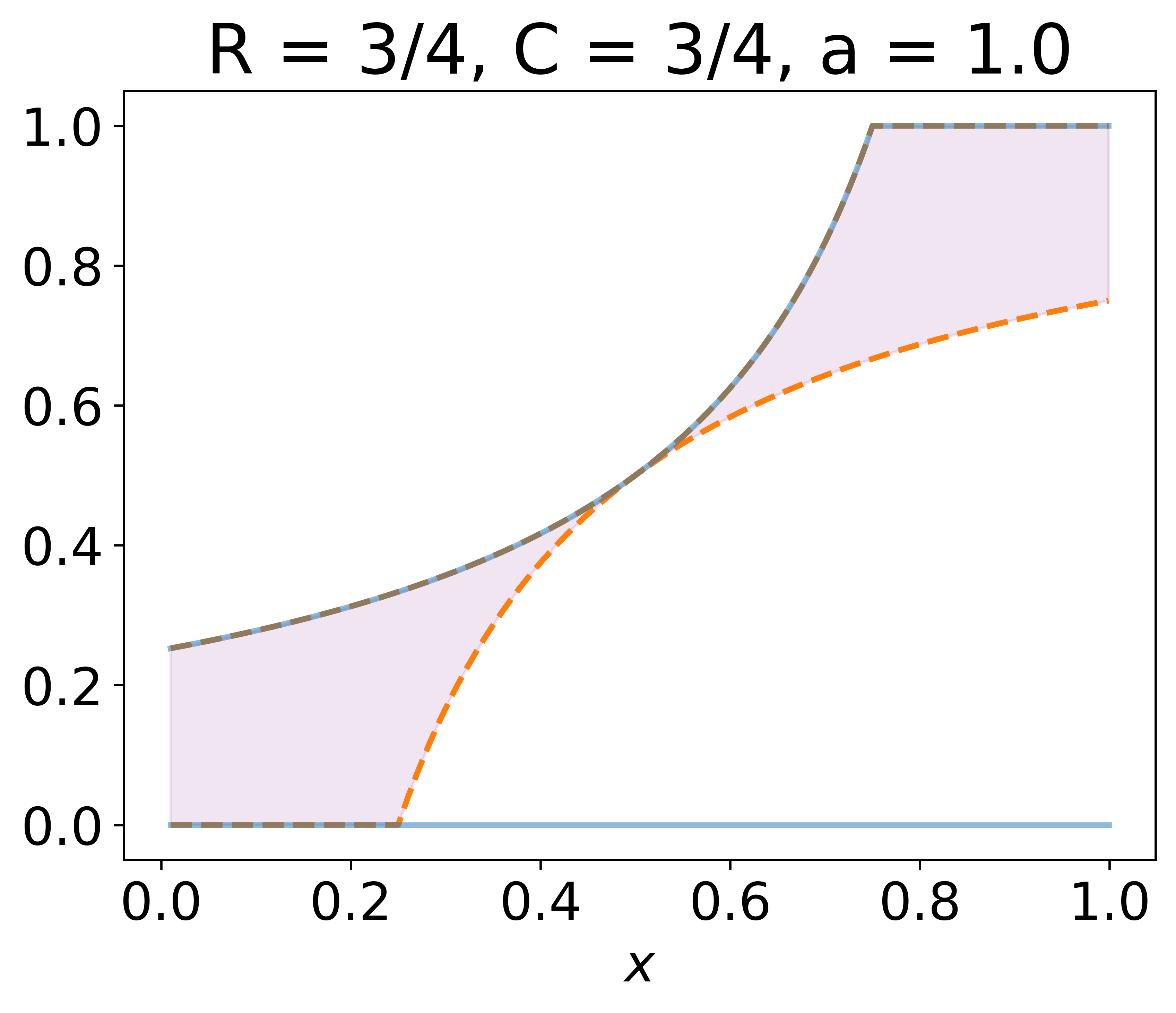}
    \end{subfigure}
    \caption{The envelopes for $R$-robustness and $C$-consistency shown together. The envelopes for robustness are the orange dotted lines, and the envelopes for consistency are the blue solid lines. The intersection of the envelopes is shaded in purple. As long as all penalty functions are continuous, increasing, and lie in the purple region, the algorithm is guaranteed to be $R$-robust and $C$-consistent.}
    \label{fig:envelope_both}
\end{figure}

If one is able to choose penalty functions $f_j$ that lie in both the envelope for $R$-robustness and the envelope for $C$-consistency, then Algorithm \ref{alg:main} using these functions $f_j$ is guaranteed to be $R$-robust and $C$-consistent. Our main result shows this to be possible whenever $R \leq \frac34$ and  $\sqrt{1-R} + \sqrt{1-C} \geq 1$. In general, there can be many feasible penalty functions even in the intersection of the two envelopes. In the {description} of Algorithm \ref{alg:main}, we chose to use one particular function to have a well-defined algorithm: the largest possible function that lies in both envelopes if $a_j < 0.5$ and the smallest possible function that lies in both envelopes if $a_j \geq 0.5$. Note that one can choose any other function in the intersection of the two envelopes without changing any of our theoretical results. We make this choice to be consistent with the functions $f_L$ and $f_U$ from earlier if $a_j=0$ or 1.

\section{Analysis for Adwords with Fractional Advice}
\label{sec:adwords_proofs}
To analyze the robustness-consistency tradeoff for two-stage Adwords with fractional advice, we will use \Cref{lem:decomp_adwords} below.  It is a generalization of the structural lemma for vertex-weighted bipartite matching (\Cref{lem:decomp_vertex_weighted}) to the Adwords setting. 

\begin{restatable}{lemma}{decompadwords}
    \label{lem:decomp_adwords}
    Let $x$ be the first-stage solution returned by Algorithm \ref{alg:main}. There exist non-negative real numbers $c_i$ for each $i \in D_1$ that satisfy the following properties:
    \begin{enumerate}
        \item For all $(i,j) \in E_1$ with $x_{ij} > 0$, we have $b_{ij} (1-f_j(x_j)) \geq c_i$. 
        \item For all $i, k \in D_1$ such that there exists some $j \in S$ with $x_{k j} > 0$, we have $\frac{c_i}{b_{ij}} \geq \frac{c_{k}}{b_{kj}}$. 
        \item For all $i \in D_1$ with  $x_i < 1$, we have $c_i = 0$. 
        \item For any $(i,j) \in E_1$ with $x_j < 1$, we have $c_i \geq b_{ij}(1-f_j(x_j))$.
    \end{enumerate}
\end{restatable}
\begin{proof}{\emph{Proof of \Cref{lem:decomp_adwords}.}}
    The proof will proceed as follows. We first write down the KKT  conditions of the convex program. Then, we show that setting $c_i$ equal to the optimal dual variable corresponding to the capacity constraint for node $i$ satisfies the desired properties.

    Introducing non-negative dual variables $\lambda_i$, $\theta_j$, and $\gamma_{ij}$ for the constraints of the convex optimization problem, the Lagrangian of the convex program is:
    $$L(\vx; \lambdav, \thetav, \gammav) = \sum_{j \in S} B_j(x_j-F_j(x_j)) + \sum_i \lambda_i(1-x_i) + \sum_j \theta_j(1-x_j) + \sum_{ij} \gamma_{ij} x_{ij}$$
    The KKT conditions state that at optimality, the primal/dual solutions satisfy
    \begin{enumerate}
        \item (Stationarity)
        $b_{ij}(1-f_j(x_j)) - \lambda_i - \frac{b_{ij}}{B_j}\theta_j + \gamma_{ij} = 0$.
        \item (Complementary Slackness)
        \begin{itemize}
            \item $\lambda_i(1-x_i) = 0$,
            \item $\theta_j(1-x_j) = 0$,
            \item $\gamma_{ij}x_{ij} = 0$.
        \end{itemize}
    \end{enumerate}
    We now show that setting $c_i = \lambda_i$ satisfies the properties in the Theorem. We check the properties one by one.
    \begin{enumerate}
        \item Let $(i,j) \in E_1$ be an edge with $x_{ij} > 0$. By complementary slackness, we know $\gamma_{ij} = 0$. Then, stationarity implies
        $$b_{ij}(1-f_j(x_j)) = \lambda_i + \frac{b_{ij}}{B_j}\theta_j \geq \lambda_i.$$
        \item Consider $i,k \in D_1$ and $j \in S$ with $x_{kj} > 0$. By complementary slackness, $\gamma_{kj} = 0$. Thus, we have
        $$b_{kj}(1-f_j(x_j)) = \lambda_k + \frac{b_{kj}}{B_j}\theta_j,$$
        and
        $$b_{ij}(1-f_j(x_j)) = \lambda_i + \frac{b_{ij}}{B_j}\theta_j - \gamma_{ij} \leq \lambda_i + \frac{b_{ij}}{B_j}\theta_j.$$
        The first equation gives $\frac{\lambda_k}{b_{kj}} = (1-f_j(x_j)) - \frac{\theta_j}{B_j}$, and the second gives $\frac{\lambda_i}{b_{ij}} \geq (1-f_j(x_j)) - \frac{\theta_j}{B_j}$. Thus $\frac{\lambda_k}{b_{kj}} \geq \frac{\lambda_i}{b_{ij}}$. 
        \item This follows directly from complementary slackness, since if $x_i < 1$ then $\lambda_i = 0$. 
        \item Consider $(i,j) \in E_1$ with $x_j < 1$. By complementary slackness, $\theta_j = 0$. Then, stationarity gives
        $$b_{ij}(1-f_j(x_j)) = \lambda_i - \gamma_{ij} \leq \lambda_i.$$
    \end{enumerate}
    \qed
\end{proof}

We now analyze the robustness and consistency of Algorithm \ref{alg:main}. The theorem below gives a characterization of the penalty functions $f_j$ that are sufficient to guarantee $R$-robustness and $C$-consistency, respectively. 

\main*

The proof of \Cref{thm:main_adwords} employs the online primal-dual technique. Based on the algorithm's decisions, we construct dual variables that are approximately feasible and equal the algorithm's objective value. 
The  structural property (\Cref{lem:decomp_adwords}) allows us to characterize the envelope of functions that guarantee $R$-robustness and $C$-consistency in the primal-dual analysis. In the requirement for robustness, the upper bound on $f_j$ guarantees approximate dual feasibility for the first stage edges, while the lower bound on $f_j$ guarantees it for the second stage edges. On the other hand, in the requirement for consistency, the upper bound on $f_j$ for $s < a_j$ intuitively says that we should not overly penalize allocating to $j$ if we have not yet matched $j$'s allocation in the advice; similarly, the lower bound on $f_j$ for $s > a_j$ penalizes the algorithm for  exceeding $j$'s allocation in the advice.

In the rest of this section, we break the proof of \Cref{thm:main_adwords} into several parts. In \Cref{sec:analysis_duals}, we show how to set the dual variables, show that their value in the dual objective equals the value of the algorithm (\Cref{clm:pd_equal_adwords}), and and prove a key lower bound (\Cref{clm:sumdualvar_adwords}) regarding the sum of the dual variables on any edge. These two claims hold for any functions $f_j: [0,1] \to [0,1]$ that are continuous and increasing, and their proofs rely on the structural property. We then use these two claims in \Cref{sec:analysis_rc} to show that if the penalty functions additionally satisfy conditions 1 and 2 in \Cref{thm:main_adwords}, then the algorithm is guaranteed to be $R$-robust and $C$-consistent, respectively. Finally, we complete the proof in \Cref{sec:funcs_exist} by showing that there exist penalty functions $f_j$ which simultaneously satisfy all the conditions in \Cref{thm:main_adwords} assuming $\sqrt{1-R} + \sqrt{1-C} \geq 1$.

\subsection{Defining the Dual Variables}
\label{sec:analysis_duals}
To begin the proof of \Cref{thm:main_adwords}, recall the LP formulation for Adwords and its dual:
\begin{align*}
\mbox{max }  \sum_{(i,j) \in E} b_{ij} z_{ij}  & & \mbox{min }  \sum_{i \in D} \alpha_i + \sum_{j \in S} \beta_j \\
\mbox{s.t. } \sum_{j: (i,j) \in E} z_{ij} \leq 1 &  \qquad\forall i \in D & \mbox{s.t. } \alpha_i + \frac{b_{ij}}{B_j}\,\beta_j \geq b_{ij} & \qquad \forall (i,j) \in E \\
\frac{1}{B_j}\sum_{i:(i,j) \in E} b_{ij}z_{ij} \leq 1 &\qquad \forall j \in S & \alpha_i, \beta_j \geq 0 & \qquad \forall i \in D, j \in S. \\
z_{ij} \geq 0 &  \qquad \forall (i,j) \in E.
\end{align*}
Let $(c_i: i \in D_1)$ be the values from \Cref{lem:decomp_adwords}. We define dual variables as follows:

\textbf{First stage.} 
For all $i \in D_1$, set $\alpha_i \gets c_i$. For all $j \in S$, set
$$
\beta_j \gets
    B_jx_j - \sum_i c_i{x}_{ij}.
$$

\textbf{Second stage.} By \Cref{prop:adwords_worst_case}, we may assume without loss of generality that the second stage graph is a matching, each of whose edges satisfy $b_{ij} \leq B_j$. For all $i \in D_2$, set $\alpha_i \gets \min\{b_{ij}, B_j(1-x_j)\}$, where $j$ is the unique neighbor of $i$ in the second stage graph. (This is exactly how much the edge contributes to the algorithm.) Leave the other dual variables unchanged.

\begin{claim}[$\ALG = \mathrm{Dual}$]
    \label{clm:pd_equal_adwords}
    The value of the algorithm is equal to the objective value of $(\bar\alphav, \bar\betav)$ in the dual.
\end{claim}

\begin{proof}{\emph{Proof.}}
    Clearly the change in primal equals the change in the dual in the second stage. The claim in the first stage follows from the below equation, which holds for each $j \in S$:
    \begin{equation}
        \label{eq:eachj}
        \beta_j + \sum_{i \in D_1} \alpha_i x_{ij} = B_j x_j.
    \end{equation}
     This will suffice to prove the claim, since assuming \eqref{eq:eachj} holds, we have
    \begin{align*}
        \sum_{j \in S} \beta_j + \sum_{i \in D_1} \alpha_i 
        &= \sum_{j \in S} \beta_j + \sum_{i \in D_1} \alpha_i x_i &\text{(since if $x_i < 1$ then $\alpha_i = 0$ by Part 3 of \Cref{lem:decomp_adwords})}\\
    &= \sum_{j \in S} \left( \beta_j + \sum_{i \in D_1}\alpha_i x_{ij} \right) \\
    &= \sum_{j \in S} B_j x_j &\text{(by \eqref{eq:eachj})}
    \end{align*}
    Let us now prove \eqref{eq:eachj}. For $j \in S$, we have
    $$\beta_j + \sum_{i \in D_1} \alpha_i x_{ij} = B_jx_j - \sum_{i \in D_1} c_i x_{ij} + \sum_{i \in D_1} c_i x_{ij} = B_jx_j,$$
    as claimed.    \qed
\end{proof}



Our next claim is a lower bound on the sum of the dual variables across any given edge, and will be crucial in the analysis of robustness and consistency.

\begin{claim}
\label{clm:sumdualvar_adwords}
    We have 
    $$\beta_j \geq B_jx_jf_j(x_j),$$
    and
    $$\alpha_i + \frac{b_{ij}}{B_j}\,\beta_j \geq
    \begin{cases}
        b_{ij}(1-f_j(x_j)+x_jf_j(x_j)), &\text{if $(i,j) \in E_1$,} \\
        b_{ij}(1-x_j + x_jf_j(x_j)), &\text{if $(i,j) \in E_2$.}
    \end{cases}
    $$
\end{claim}

\begin{proof}{\emph{Proof of \Cref{clm:sumdualvar_adwords}}.} We prove the two parts separately. 

\underline{First part of claim.} We have
$$\beta_j = B_jx_j - \sum_i c_ix_{ij} \stackrel{(a)}{\geq} 
B_jx_j - \sum_i b_{ij}(1-f_j(x_j))x_{ij}
= B_jx_jf_j(x_j),
$$
where $(a)$ is by Part 1 of \Cref{lem:decomp_adwords}.

\underline{Second part of claim.} First, consider $(i,j) \in E_1$. If $x_j = 1$, then 
\begin{align*}
\alpha_i + \frac{b_{ij}}{B_j}\,\beta_j &= c_i + \frac{b_{ij}}{B_j}\left(B_j - \sum_{k} c_kx_{kj} \right) \\
&= c_i + b_{ij} - \frac{b_{ij}}{B_j}\sum_k\left(\frac{c_k}{b_{kj}}\right) b_{kj}x_{kj} \\
&\geq c_i + b_{ij} - \frac{b_{ij}}{B_j}\sum_k\left(\frac{c_i}{b_{ij}}\right) b_{kj}x_{kj} &\text{(by Part 2 of \Cref{lem:decomp_adwords})}\\
&= b_{ij}
\end{align*}
as desired.  On the other hand, if $x_j < 1$, then 
\begin{align*}
\alpha_i + \frac{b_{ij}}{B_j}\,\beta_j 
\geq c_i + \frac{b_{ij}}{B_j}\left(B_jx_jf_j(x_j)\right) 
= c_i + b_{ij} x_jf_j(x_j) 
\geq b_{ij}(1-f_j(x_j))+ b_{ij} x_jf_j(x_j),
\end{align*}
where the first inequality applies the lower bound on $\beta_j$ from the first part of the claim and the final inequality is by Part 4 of \Cref{lem:decomp_adwords} (where $x_j < 1$).

Next, consider $(i,j) \in E_2$. Then $\alpha_i = \min\{b_{ij}, B_j(1-x_j)\}$. If $x_j < 1$ then using  $\beta_j \geq B_jx_jf_j(x_j)$,
\begin{align*}
\alpha_i + \frac{b_{ij}}{B_j}\,\beta_j &= \min\{b_{ij}, B_j(1-x_j)\} + b_{ij}x_jf_j(x_j) \\
&\geq\min\{b_{ij}, B_j\}(1-x_j) + b_{ij}x_jf_j(x_j) \\
&= b_{ij}(1-x_j) + b_{ij}x_jf_j(x_j), &\text{(we assumed $b_{ij}\leq B_j$ for second-stage edges)}
\end{align*}
and the claim holds. On the other hand, suppose $x_j = 1$. Then $\alpha_i = 0$. So,
\begin{align*}
\alpha_i + \frac{b_{ij}}{B_j}\,\beta_j \geq  \frac{b_{ij}}{B_j}\cdot B_jx_jf_j(x_j) = b_{ij}f_j(x_j)
\end{align*}
as desired.  This proves the claim. \qed
\end{proof}

\subsection{Bounding Robustness and Consistency}
\label{sec:analysis_rc}
Next, we use \Cref{clm:pd_equal_adwords} and \Cref{clm:sumdualvar_adwords} to show that if the penalty functions satisfy the requirements in \Cref{thm:main_adwords}, then the algorithm is $R$-robust and $C$-consistent.
\begin{claim}
    If we run Algorithm \ref{alg:main} with continuous, increasing functions $f_j: [0,1] \to [0,1]$ with $$1 - \frac{1-R}{s} \leq f_j(s) \leq \frac{1-R}{1-s}$$ for all $s \in (0, 1)$, then the algorithm is $R$-robust.
\end{claim}
\begin{proof}{\emph{Proof.}}
    By \Cref{clm:pd_equal_adwords}, for any graph $G$ and advice $A$, we have $\ALG(G,A) = \sum_{i \in D} \alpha_i + \sum_{j \in S} \beta_j$. Therefore, to show $R$-robustness, we just need to show that $(\alpha, \beta)$ is $R$-approximately feasible to the dual, i.e. $\alpha_i + \frac{b_{ij}}{B_j}\beta_j \geq R\,b_{ij}$ for all $(i,j) \in E$. 
    
    If $(i,j) \in E_1$, then 
    $$
        \alpha_i + \frac{b_{ij}}{B_j}\beta_j
        \stackrel{(a)}{\geq}   b_{ij}(1-f_j(x_j)+x_jf_j(x_j))
        \stackrel{(b)}{\geq} R\,b_{ij}
    $$
    where $(a)$ is by \Cref{clm:sumdualvar_adwords} and $(b)$ is because $f_j(s) \leq \frac{1-R}{1-s}$.

    On the other hand, if $(i,j) \in E_2$, then 
    $$\alpha_i + \frac{b_{ij}}{B_j}\beta_j
        \stackrel{(a)}{\geq} b_{ij}(1-x_j + x_jf_j(x_j)) 
        \stackrel{(b)}{\geq} R\,b_{ij}$$
    where $(a)$ is by \Cref{clm:sumdualvar_adwords} and $(b)$ is because $f_j(s) \geq 1 - \frac{1-R}{s}$. \qed
\end{proof}

\begin{claim}
    If we run Algorithm \ref{alg:main} with continuous, increasing functions $f_j: [0,1] \to [0,1]$ with
    $$\text{$f_j(s) \leq \frac{a_j(1-C)}{a_j-s}$ for all $s \in [0, a_j)$, and $f_j(s) \geq 1 - \frac{1-C}{s-a_j}$ for all $s \in (a_j, 1]$},$$ then the algorithm is $C$-consistent . 
\end{claim}
\begin{proof}{\emph{Proof}.}
By \Cref{prop:adwords_worst_case}, the worst case for consistency is when the second-stage graph is a matching (call it $M_2$), such that $b_{ij} = B_j(1-a_j)$ for all $(i,j) \in M_2$. Let $S_2$ be the nodes matched by $M_2$, and let $S_1 = S \setminus S_2$. Then,
$$\ADVICE = \sum_{j \in S_1} B_ja_j + \sum_{j \in S_2} B_j.$$
On the other hand, using \Cref{clm:pd_equal_adwords}, we can decompose the value of the algorithm as follows: 
\begin{align*}
\ALG &= \sum_{i \in D_1} \alpha_i + \sum_{i \in D_2} \alpha_i +\sum_{j \in S} \beta_j \\
&\geq \sum_{i \in D_1} \alpha_i \sum_{j \in S} a_{ij} + \sum_{j \in S} \beta_j + \sum_{i \in D_2} \alpha_i\\
&= \sum_{j \in S}\left( \sum_{i \in D_1}a_{ij}\alpha_i + \beta_j\right)+ \sum_{i \in D_2} \alpha_i \\
&= \sum_{j \in S_1}\left( \sum_{i \in D_1}a_{ij}\alpha_i + \beta_j\right)+ \sum_{j \in S_2}\left( \sum_{i \in D_1}a_{ij}\alpha_i + \beta_j\right) + \sum_{i \in D_2} \alpha_i
\end{align*}
Thus, to show $\ALG \geq C\cdot \ADVICE$, it suffices to show 
$$
(1)~ \left( \sum_{i \in D_1}a_{ij}\alpha_i + \beta_j\right) \geq C\cdot B_ja_j ~\text{for all $j \in S_1$}, ~~\text{and}~~ (2)~ \left( \sum_{i' \in D_1}a_{i'j}\alpha_{i'} + \beta_j\right) + \alpha_i \geq C \cdot B_j ~\text{for all $(i,j) \in M_2$.}
$$
We first show (1). For $j \in S_1$, we have
\begin{align*}
    \sum_{i \in D_1}a_{ij}\alpha_i + \beta_j
    &=  \sum_{i \in D_1}a_{ij}\alpha_i + a_j\beta_j +(1-a_j)\beta_j\\
    &= \sum_{i \in D_1}a_{ij}\left(\alpha_i + \frac{b_{ij}}{B_j}\beta_j\right) +(1-a_j)\beta_j \\
    &\geq \sum_{i \in D_1}a_{ij}b_{ij}(1-f_j(x_j)+x_jf_j(x_j)) +(1-a_j)B_jx_jf_j(x_j) &\text{(By \Cref{clm:sumdualvar_adwords})}\\
    &= B_ja_j(1-f_j(x_j)+x_jf_j(x_j)) +(1-a_j)B_jx_jf_j(x_j) \\
    &= B_j\left (a_j+ (x_j-a_j)f_j(x_j)\right) \\
    &\geq C\cdot B_ja_j,
\end{align*}
Where the last inequality holds trivially if $x_j \geq a_j$, and holds if $x_j < a_j$ since we have  $f_j(s) \leq \frac{a_j(1-C)}{a_j-s}$ for all $s \in [0, a_j)$.

Next we show (2). From the calculation above, we already know that
$$\sum_{i' \in D_1}a_{i'j}\alpha_{i'} + \beta_j \geq B_j\left (a_j+ (x_j-a_j)f_j(x_j)\right). $$
The additional contribution (due to $\alpha_i$), is 
$$\alpha_i = \min\{b_{ij}, B_j(1-x_j)\} = B_j\min\{(1-a_j), (1-x_j)\} = B_j - B_j\max\{a_j, x_j\}.$$
Thus, in total we have
\begin{align*}
\sum_{i' \in D_1}a_{i'j}\alpha_{i'} + \beta_j + \alpha_i &\geq B_j\left (a_j+ (x_j-a_j)f_j(x_j)\right) + B_j - B_j\max\{a_j, x_j\} \\
&= B_j\left(1 - (x_j-a_j)^+ + (x_j-a_j)f_j(x_j)\right) ~(\star)
\end{align*}
We consider two cases, depending on which of $x_j$ or $a_j$ is larger.

\underline{Case 1.} $a_j \geq x_j$. Then,
$$(\star) = B_j(1 +(x_j-a_j)f_j(x_j)) \stackrel{(a)}{\geq} B_j(1-a_j(1-C)) \geq B_j \cdot C,$$
where $(a)$ is because $f_j(s) \leq \frac{a_j(1-C)}{a_j-s}$ for all $s \in [0, a_j)$.

\underline{Case 2.} $a_j \leq x_j$. Then, 
$$(\star) = B_j(1 -(x_j-a_j)(1-f_j(x_j))) \stackrel{(b)}{\geq} B_j(1-(1-C)) = B_j \cdot C,$$
where $(b)$ is because $f_j(s) \geq 1 - \frac{1-C}{s-a_j}$ for all $s \in (a_j, 1]$.
\qed
\end{proof}

\subsection{Existence of Penalty Functions}
\label{sec:funcs_exist}

Finally, we show there exist functions $f_j$ satisfying the requirements in \Cref{thm:main_adwords}, which completes the analysis. 
\begin{claim}
    \label{clm:fn_exist}
    Let $R \in [0, \frac34]$ and let $C = 2\sqrt{1-R} - (1-R)$. For any $a \in [0,1]$, there exists a function $f:[0,1] \to [0,1]$ that satisfies:
    \begin{enumerate}
    \item $f$ is continuous and increasing,
    \item (Requirement for Robustness) $1 - \frac{1-R}{s} \leq f(s) \leq \frac{1-R}{1-s}$.
    \item (Requirement for Consistency) $f(s) \leq \frac{a(1-C)}{a-s}$ for all $s \in [0, a)$, and $f(s) \geq 1 - \frac{1-C}{s-a}$ for all $s \in (a, 1]$. 
\end{enumerate}
\end{claim}
\begin{proof}{\emph{Proof}.}
    Note that the required bounds on $f$ can be summarized as follows:
    \begin{itemize}
        \item \textbf{Lower bound:} 
        $$f(s) \geq \LB(s) := 
        \begin{cases}
        \max\left\{0,\, 1-\frac{1-R}{s}\right\}, &\text{$s\in [0,a]$} \\
        \max\left\{0,\, 1-\frac{1-R}{s},\, 1 - \frac{1-C}{s-a}\right\}, &\text{$s \in (a, 1]$}
        \end{cases}
        $$
        \item \textbf{Upper bound:}
        $$f(s) \leq \UB(s) :=
        \begin{cases}
        \min\left\{1,\, \frac{1-R}{1-s},\, \frac{a(1-C)}{a-s} \right\}, &\text{$s\in [0,a)$} \\
        \min\left\{1,\, \frac{1-R}{1-s}\right\}, &\text{$s \in [a, 1]$}
        \end{cases}
        $$
    \end{itemize}
    Note that $\LB(s)$ and $\UB(s)$ are both continuous and increasing in $s$. Thus, to show there exists a function $f$ that satisfies the three properties in the claim, it suffices to show that $\LB(s) \leq \UB(s)$ for all $s \in [0,1]$. Feasible choices for $f$ could then be, for example, $f(s) = \LB(s)$, $f(s) = \UB(s)$, or generally any continuous increasing function in the envelope between the lower and upper bounds.

    To show $\LB(s) \leq \UB(s)$, we will consider two cases depending on if $s < a$ or $s > a$. (The case $s = a$ will follow from continuity.)

    \underline{Case 1. $s < a$.} We need to check that
    $$\max\left\{0,\, 1-\frac{1-R}{s}\right\} \leq \min\left\{1,\, \frac{1-R}{1-s},\, \frac{a(1-C)}{a-s} \right\},$$
    or equivalently, that every term in the max on the left is less than every term in the min on the right. The only non-trivial relations to check are (1) $1-\frac{1-R}{s} \leq \frac{1-R}{1-s}$, and (2) $1-\frac{1-R}{s} \leq \frac{a(1-C)}{a-s}$. (1) is equivalent to 
    $$1 \leq (1-R)\left(\frac{1}{s} + \frac{1}{1-s}\right),$$
    which is true since $\frac{1}{s} + \frac{1}{1-s} \geq 4$ (with equality when $s = 2$), and $R \leq \frac34$.

    Next we show (2). (2) is equivalent to
    $$\frac{1-R}{s} + \frac{a(1-C)}{a-s} \geq 1.$$ 
    Let $h(s)$ denote the left-hand side expression. If $C=1$ then $R=0$ and the inequality is trivial. Otherwise, $\lim_{s \to 0^+}h(s) = \lim_{s \to a^-} h(s) = \infty$. Thus, the minimum of left-hand side over $s < a$ must occur at a point $s^* \in (0,a)$ with $h'(s^*) = 0$. Taking the derivative and setting to zero, we get
    $$-\frac{1-R}{(s^*)^2} + \frac{a(1-C)}{(a-s^*)^2} = 0 \implies s^* = \frac{a}{1 + \sqrt{\frac{a(1-C)}{1-R}}}.$$
    Substituting this back, we get
    $$h(s^*) = \left(\sqrt{1-C} + \sqrt{\frac{1-R}{a}}\right)^2 \geq \left(\sqrt{1-C} + \sqrt{1-R}\right)^2 = 1.$$
    \underline{Case 2. $s > a$.} We need to check that
    $$\max\left\{0,\, 1-\frac{1-R}{s},\, 1 - \frac{1-C}{s-a}\right\} \leq  \min\left\{1,\, \frac{1-R}{1-s}\right\}.$$
    The two non-trivial inequalities to check are (1) $1-\frac{1-R}{s} \leq \frac{1-R}{1-s}$ and (2) $1 - \frac{1-C}{s-a} \leq \frac{1-R}{1-s}$. We already verified (1) holds in Case 1 above, so we just need to show (2). (2) is equivalent to
    $$\frac{1-R}{1-s} + \frac{1-C}{s-a} \geq 1.$$
    Let $g(s)$ denote the left-hand side expression. If $C = 1$ then $R = 0$ and the inequality trivially holds. Otherwise, $\lim_{s \to a^+} g(s) = \lim_{s \to 1^-} g(s) = \infty$. Thus, the minimum of the left-hand side over $s > a$ must occur at a point $s^* \in (a, 1)$ with $g'(s^*)= 0$. Taking the derivative and setting to zero, we get
    $$\frac{1-R}{(1-s^*)^2} - \frac{1-C}{(s^*-a)^2} = 0 \implies s^* = \frac{a\sqrt{1-R} + \sqrt{1-C}}{\sqrt{1-R} + \sqrt{1-C}} = a\sqrt{1-R} + \sqrt{1-C}.$$
    Substituting this back, we get
    \begin{align*}
    h(s^*) &= \frac{1-R}{1-a\sqrt{1-R}-\sqrt{1-C}} + \frac{1-C}{a\sqrt{1-R}+\sqrt{1-C}-a} \\
    &= \frac{1-R}{(1-a)\sqrt{1-R}} + \frac{1-C}{(1-a)\sqrt{1-C}} \\
    &=\frac{1}{1-a} \geq 1.
    \end{align*}
\qed
\end{proof}

\section{Tightness of the Robustness-Consistency Tradeoff}
\label{sec:tight}

\begin{figure}

    \centering
    \begin{tikzpicture}[scale=0.73]
    \vertex (s1) at (0, 3) [label=above left:1] {};
    \vertex (s2) at (0, 0) [label=below left:2]{};
    \vertex (d1) at (3, 3) [label=right:1] {};
    \vertex (d2) at (3, 0) [label=right:2] {};
    \node (label1) at (-0.5, 2.6) {\text{\footnotesize{$w_1 = w$}}};
    \node (label2) at  (-0.5, 0.4) {\text{\footnotesize{$w_2 = 1$}}};
    \node (label3) at  (2.5, 1.8) {\text{\footnotesize{$1-x$}}};
    \draw (s1) --node[above] {\text{\footnotesize{$x$}}} (d1);
    \draw[green] (s1) -- (d1);
    \draw (s2) -- (d1);
    \draw[dotted] (d2) -- (s1);
    \draw[dotted] (d2) -- (s2);
    \end{tikzpicture}
    \caption{Illustration of the hardness instance. $S$ is on the left and $D$ is on the right. The first arrival neighbors both vertices of $S$. The second arrival neighbors exactly one vertex of $S$, but it could be either vertex. The advice is to match the green edge $(1,1)$.}
    \label{fig:hardness}
\end{figure}

\Cref{thm:main_adwords,thm:vertex_weighted} show that for two-stage vertex-weighted matching and Adwords, the tradeoff curve $\sqrt{1-R} + \sqrt{1-C} = 1$ for $R \in [0, \frac{3}{4}]$ is achievable. We now show this tradeoff to be tight.

\noindent\textbf{The hardness instance.} The hardness instance is illustrated in \Cref{fig:hardness}. It is an instance of vertex-weighted bipartite matching. (As Adwords is a generalization of vertex-weighted matching, this lower bound will also apply to Adwords.) The instance is a graph where $D$ has two vertices $i = 1,2$ that arrive one in each stage. $S$ has two vertices $j = 1,2$ with weights $w_1 = w$ (for some $w  \in [0,1]$ that we will set later) and $w_2 = 1$. The first-stage graph consists of both edges $(1,1)$ and $(1,2)$, and the advice suggests matching along the edge $(1,1)$, which is "extreme" advice in that offline vertex 1 has the lower weight one edge (which can be either $(2,1)$ or $(2,2)$). Since the algorithm does not know which edge will arrive in the second stage, it must hedge against both possibilities when making its first-stage decision. 

\begin{theorem}
\label{thm:tight}
Let $w = \frac{1}{\sqrt{1-R}} - 1$ in the instance described above.\footnote{This choice of $w$ was obtained by minimizing $(1-R)w + \frac{1}{1+w}$ over $w$.} Then any algorithm that is $R$-robust is at most $C$-consistent where $C = 2\sqrt{1-R} - (1-R)$. 
\end{theorem}

\begin{proof}{\emph{Proof of \Cref{thm:tight}}.} 
Let $x := x_{11}$ and $1-x := x_{12}$, so that the algorithm's first-stage decision is entirely characterized by the value of $x$.  There are two cases.
\begin{enumerate}
    \item Edge $(2, 1)$ arrives in the second stage. Then $\ALG = w + 1 - x$ and $\ADVICE = w$. 
    \item Edge $(2,2)$ arrives in the second stage. Then $\ALG = wx + 1$ and $\ADVICE = 1 + w$. 
\end{enumerate}
(Note that regardless of which edge arrives in the second stage, $\OPT$ is always equal to $1+w$.) For the algorithm to be $R$-robust in Case 1, we must have
$$
\frac{w+1-x}{1+w} \geq R \implies x \leq (1-R)(1+w).
$$
On the other hand, for the algorithm to be $C$-consistent in Case 2, we must have
$$
\frac{wx+1}{1+w} \geq C \implies x \geq \frac{C(1+w)-1}{w}.
$$
Since the algorithm does not know which of the two cases will happen in the second stage, it must choose an $x$ that satisfies both of the inequalities above. For a desired robustness $R$ and consistency $C$, this is only possible if
$$\frac{C(1+w)-1}{w} \leq (1-R)(1+w),$$
which when rearranged becomes 
$$
C \leq (1-R)w + \frac{1}{1+w}.
$$
Substituting $w = \frac{1}{\sqrt{1-R}} - 1$ above gives $C \leq 2\sqrt{1-R} - (1-R)$, as desired. \qed
\end{proof}

\section{Experiments} \label{sec:experiments}
The code for all experiments can be accessed at \href{https://github.com/MapleOx/matching_predictions}{https://github.com/MapleOx/matching\_predictions}.

\subsection{Synthetic Experiments}
\label{subsec:exp_synth}

We begin with experiments on fractional\footnote{
In a bipartite graph, any fractional first-stage solution $\mathbf{x}$ can be expressed as a distribution over first-stage matchings $M_1$.  On a worst-case $G_2$, the performance of the fractional solution will satisfy $w(\mathbf{x},G_2)=\bE_{M_1}[w(M_1,G_2)]$, as we saw in \Cref{prop:vertex_weighted_fractional}.  However, on the more benign instances that we simulate here, we will generally have $w(\mathbf{x},G_2)>\bE_{M_1}[w(M_1,G_2)]$, noting the concavity of the function $w(\cdot,G_2)$. It is also true that if $G_2$ was known, then in a bipartite graph one can always find an integral matching $M_1$ such that $w(M_1,G_2) \ge w(\mathbf{x},G_2)$, but this is not true for unknown $G_2$ (as evidenced by the optimal SAA solution for two-stage matching needing to be fractional).

In sum, there is a difference between fractional vs.\ integral bipartite matchings in these simulations, and we assume fractional matchings to avoid the complexities of randomized rounding.
} vertex-weighted bipartite matching with synthetic data. In all the experiments here, the base graph will have 12 supply nodes and 12 demand nodes, with 6 demand nodes arriving in the first stage and 6 arriving in the second stage. 

\noindent\textbf{First-stage demand $D_1$.} The first-stage graph is fully connected, meaning each demand node in this stage has an edge to every supply node. Because the setting is vertex-weighted, it is always best to fully match the 6 first-stage demand nodes, implying that the sum of fractions filled of supply nodes after the first stage will be 6.  The sum of fractions unfilled will also be $12-6=6$, corresponding to the warehouses in which inventory was (fractionally) placed in the e-commerce fulfillment problem discussed in \Cref{subsec:introNum}.

\noindent\textbf{Second-stage demand $D_2$.} The second-stage graph is drawn from a probability distribution as follows. Each supply node is associated with a probability $p_j$. Then, each second-stage demand node independently samples an edge to supply node $j$ with probability $p_j$. Thus, the expected degree of each second-stage demand node is exactly $\sum_{j \in S} p_j$.

\noindent\textbf{Advice.} The advice is generated using sample average approximation (SAA) as follows. We generate a number of samples of the second-stage graph, and take the advice to be the fractional matching in the first stage that maximizes the expected value over the uniform distribution of the samples. We conduct two sets of experiments: one without distribution shift and one with distribution shift. In the experiment without distribution shift, the SAA samples are drawn from the true second-stage distribution. In the experiment with distribution shift, the SAA samples are drawn using probabilities $\hat{p}_j$ that are perturbed from the true probabilities $p_j$. 

\noindent\textbf{Algorithms.} We tested the following algorithms:
\begin{enumerate}
    \item The \textbf{greedy} algorithm, which selects the maximum weight matching in the first stage;
    \item The fully \textbf{robust} algorithm from \cite{feng2021two};
    \item The algorithm which exactly follows the \textbf{advice} in the first stage;
    \item Our algorithm with $R \in \{0.15, 0.3, 0.45, 0.6, 0.75\}$.  Note that the robust decision often coincides with that of our algorithm with $R=0.75$, which is why it may be hidden in the plots.
\end{enumerate}

We compared the algorithms' performance with three levels of variation in the weights and two levels of variation in the probabilities.

\noindent\textbf{Weight variations.}
\begin{itemize}
    \item No variation: All supply nodes have a weight of 1.
    \item Small variation: The weight of each supply node is sampled from a $\mathrm{Uniform}(1, 1.2)$ distribution.
    \item Large variation: The weight of each supply node is sampled from a $\mathrm{Uniform}(1, 2)$ distribution.
\end{itemize}

\noindent\textbf{Probability variation.}
\begin{itemize}
    \item Small variation: Each $p_j$ is independently sampled from a $\mathrm{Uniform}(\frac{1}{4}, \frac{1}{2})$ distribution.
    \item Large variation: Each $p_j$ is independently sampled from a $\mathrm{Uniform}(\frac{1}{6}, \frac{1}{2})$ distribution.
\end{itemize}

\Cref{fig:synth} shows the results without distribution shift. Here, SAA samples are generated using the true probabilities $p_j$. We plot the performance of the algorithms as the number of SAA samples increases from 1 to 50. For each SAA sample count, we ran 100 independent trials where we redrew the samples. The $y$-axis represents performance, calculated as the average value earned by the algorithm over these 100 trials divided by the average value of the optimal hindsight matching, resulting in $y$-axis values between 0 and 1. Note that the curves for Greedy and Robust are flat because they do not use the advice. Generally, algorithms following the advice perform better when the probability variation is large, as there is more signal in the samples and more benefit to learning which supply nodes should be de-prioritized because they are easier to match in the second stage. Overall, the Robust algorithm of \cite{feng2021two} and our algorithm with $R = 0.75$ excel when the sample count is low, while Advice excels when the sample count is high.  Our algorithms with intermediate $R$ values interpolate between Advice and Robust, often outperforming both when the number of SAA samples is moderate.

\Cref{fig:synth_perturbed} shows the results with distribution shift, where SAA samples for the second-stage graph are not generated from the true distribution (governed by probabilities $p_j$) but from a perturbed distribution with probabilities $\hat{p}_j$ instead. In this experiment we set $\hat{p}_j$ to be $p_j$ plus a $\mathrm{Uniform}(-0.1, 0.1)$ perturbation, introducing bias in the samples. As expected, algorithms using the advice show decreased performance under distribution shift. However, algorithms with intermediate $R$ values are less impacted by the shift than the purely Advice-based algorithm. Thus, our algorithms hedge against both insufficient number of samples and also potential distribution shifts or corruptions in the data.

\noindent\textbf{Summary.}
Our algorithm outperforms Advice (which directly implements the SAA solution) when the number of samples is small, especially in the setting with distribution shift.
The intuition is that SAA might "overfit" to the supply nodes that it thinks can be filled in the second stage, leaving them completely unfilled after the first stage, whereas our algorithm produces a more balanced solution akin to Robust after the first stage.
Outperforming SAA over average, benign instances might be surprising because it is not what our ALPS desiderata~\eqref{eqn:introConsistent}--\eqref{eqn:introRobust} were designed to do.
This reinforces our message that the ALPS framework can be additionally viewed as a way to "regularize" decisions made from partially-reliable information, in addition to the typically touted benefits of robustness guarantees in practice \citep[cf.][]{wiermanPres}.
That being said, our improvement over SAA can be viewed as modest, which is to be expected when trying to beat SAA over average benign instances \citep[cf.][\S6]{besbes2023big}.

\begin{figure}[ht]
    \centering
    \begin{minipage}{.33\textwidth}
        \centering
        \includegraphics[width=\linewidth]{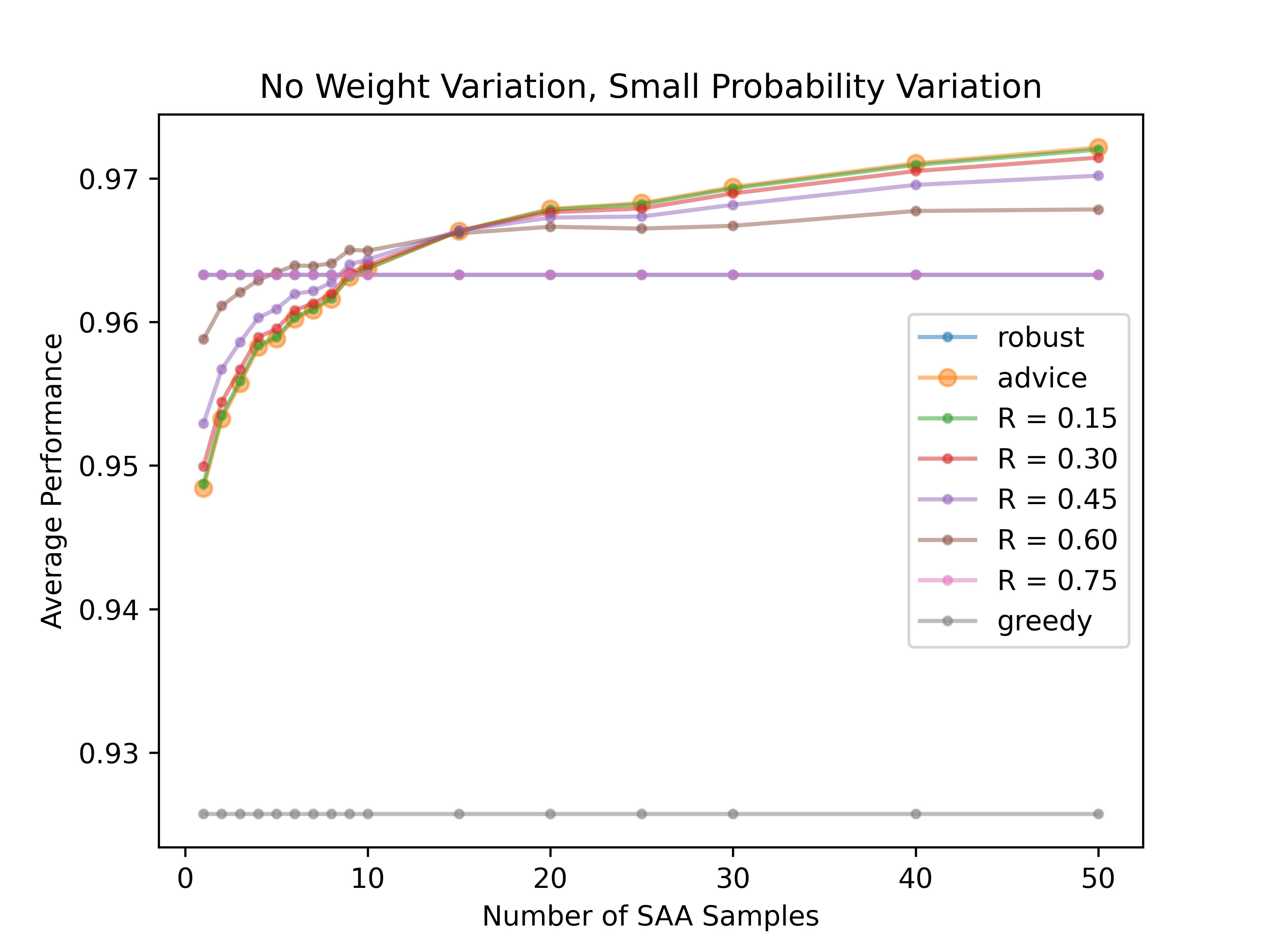}
    \end{minipage}%
    \begin{minipage}{.33\textwidth}
        \centering
        \includegraphics[width=\linewidth]{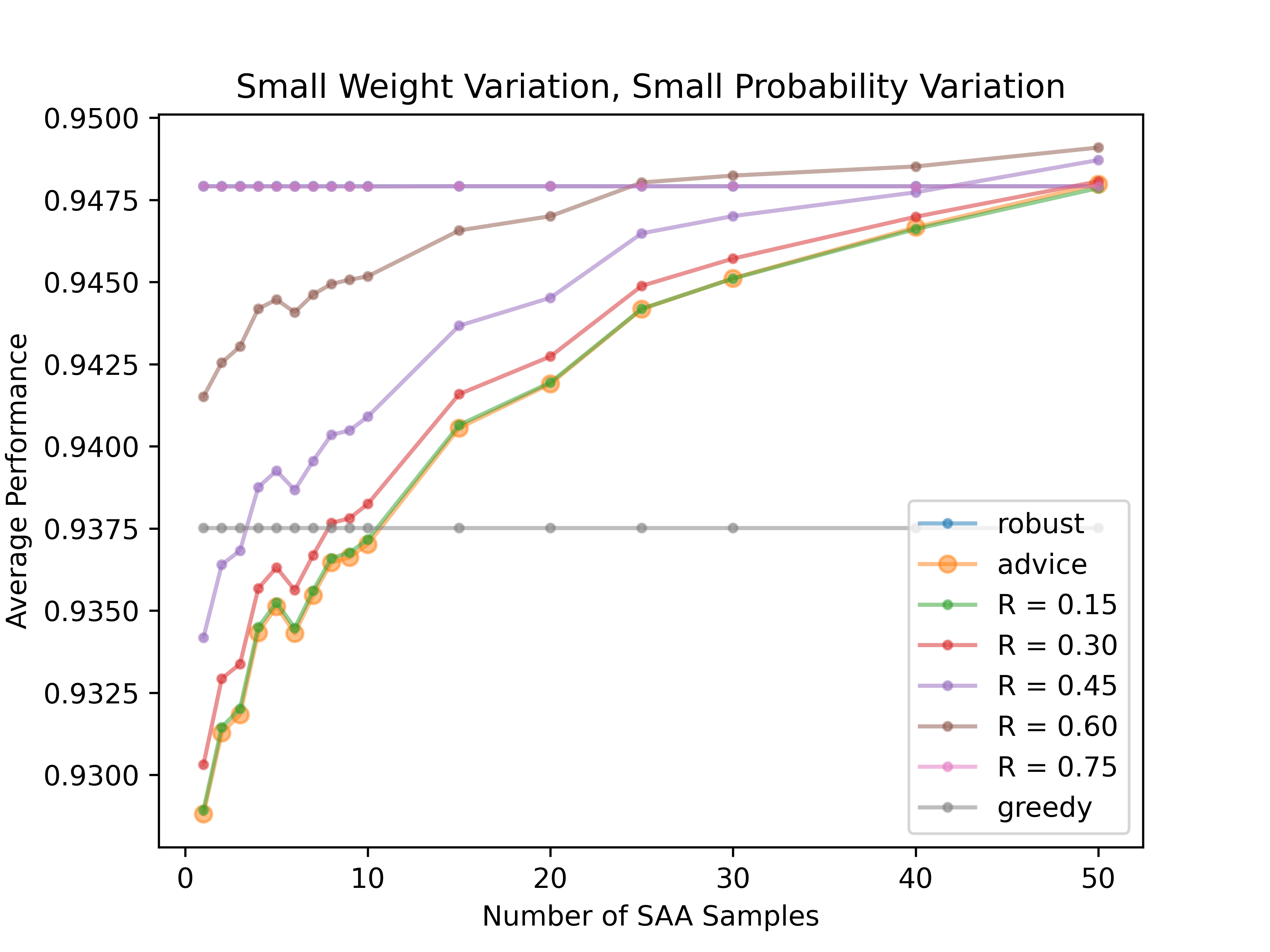}
    \end{minipage}%
    \begin{minipage}{.33\textwidth}
        \centering
        \includegraphics[width=\linewidth]{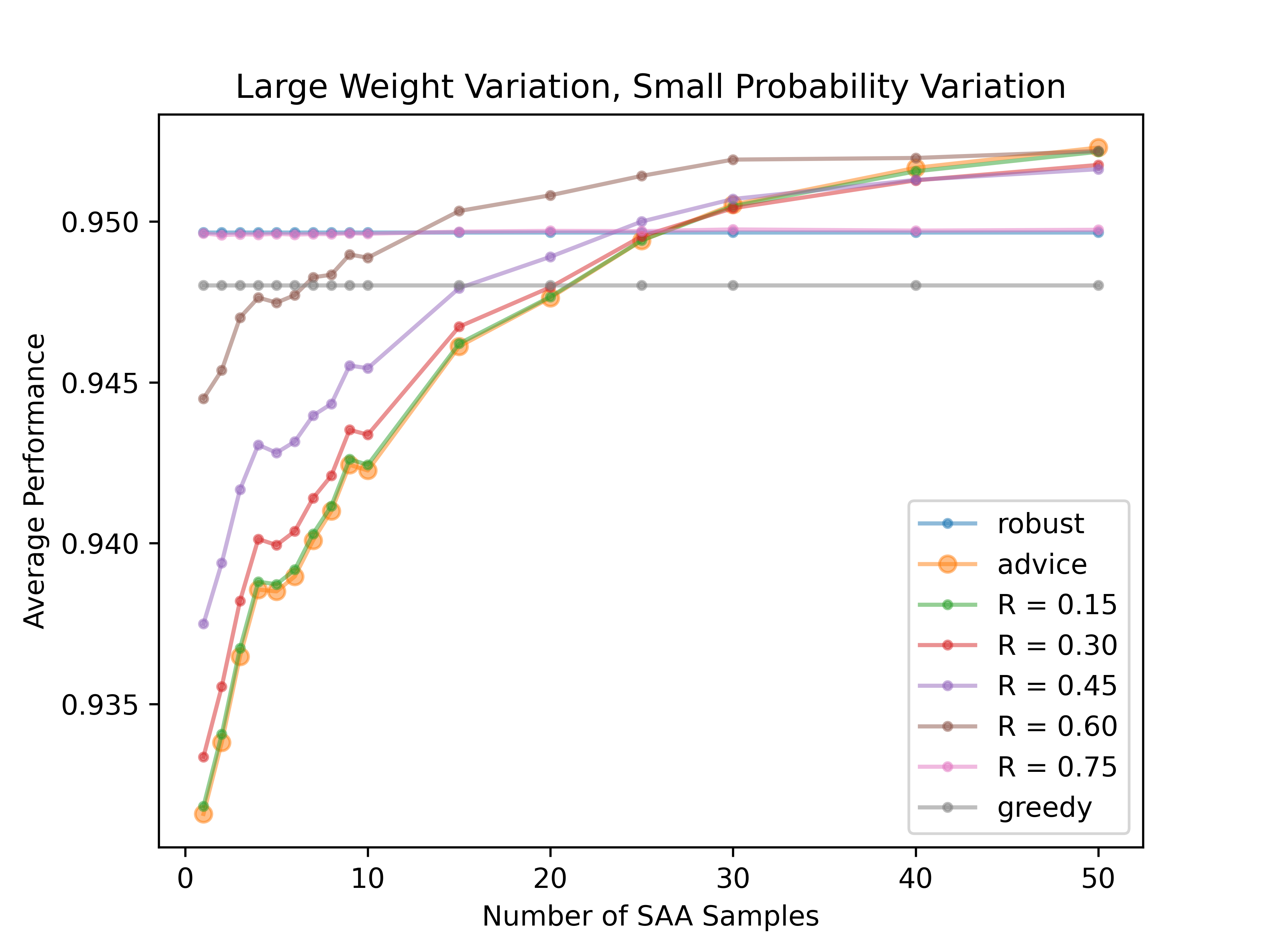}
    \end{minipage}

    \begin{minipage}{.33\textwidth}
        \centering
        \includegraphics[width=\linewidth]{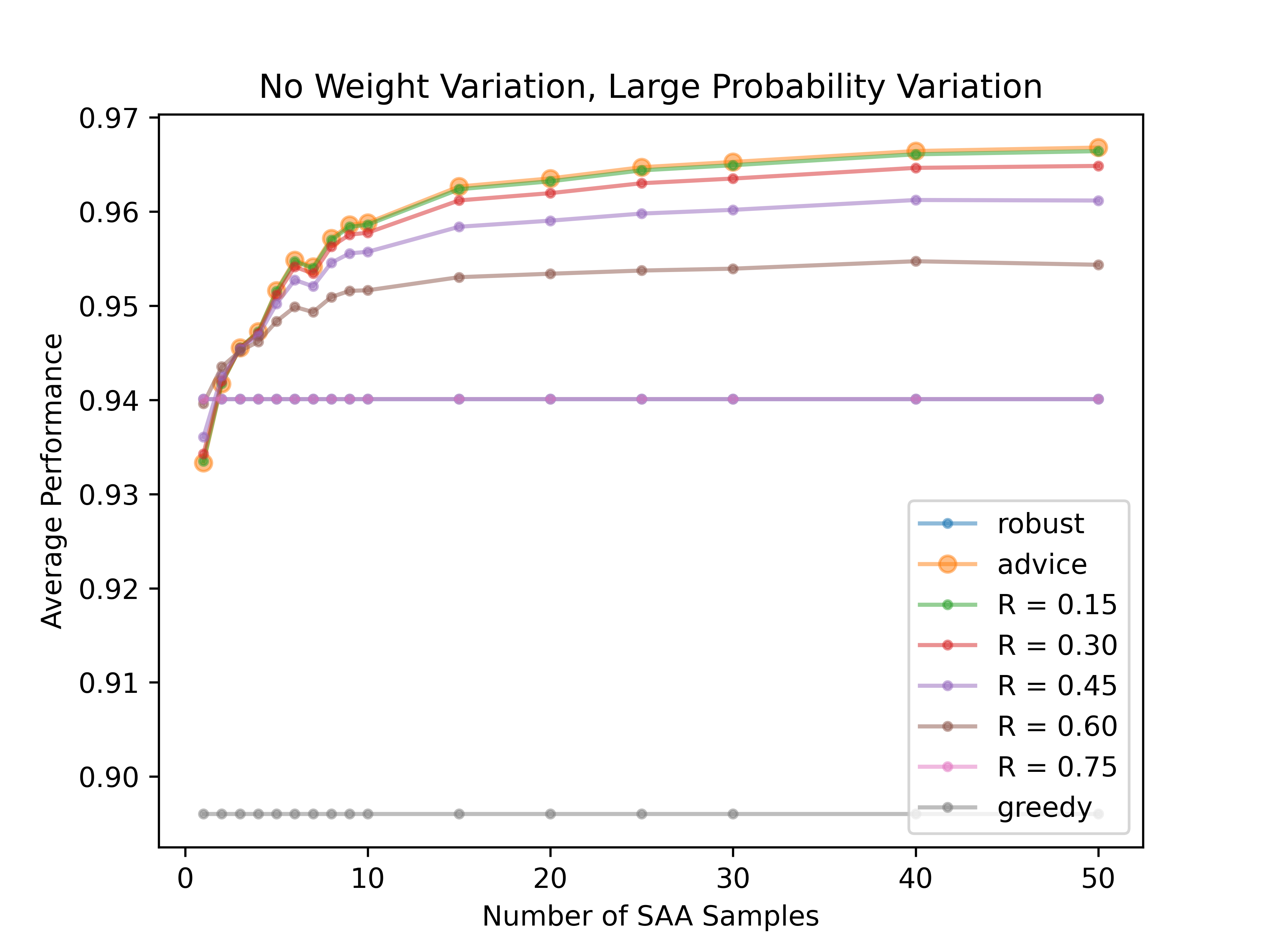}
    \end{minipage}%
    \begin{minipage}{.33\textwidth}
        \centering
        \includegraphics[width=\linewidth]{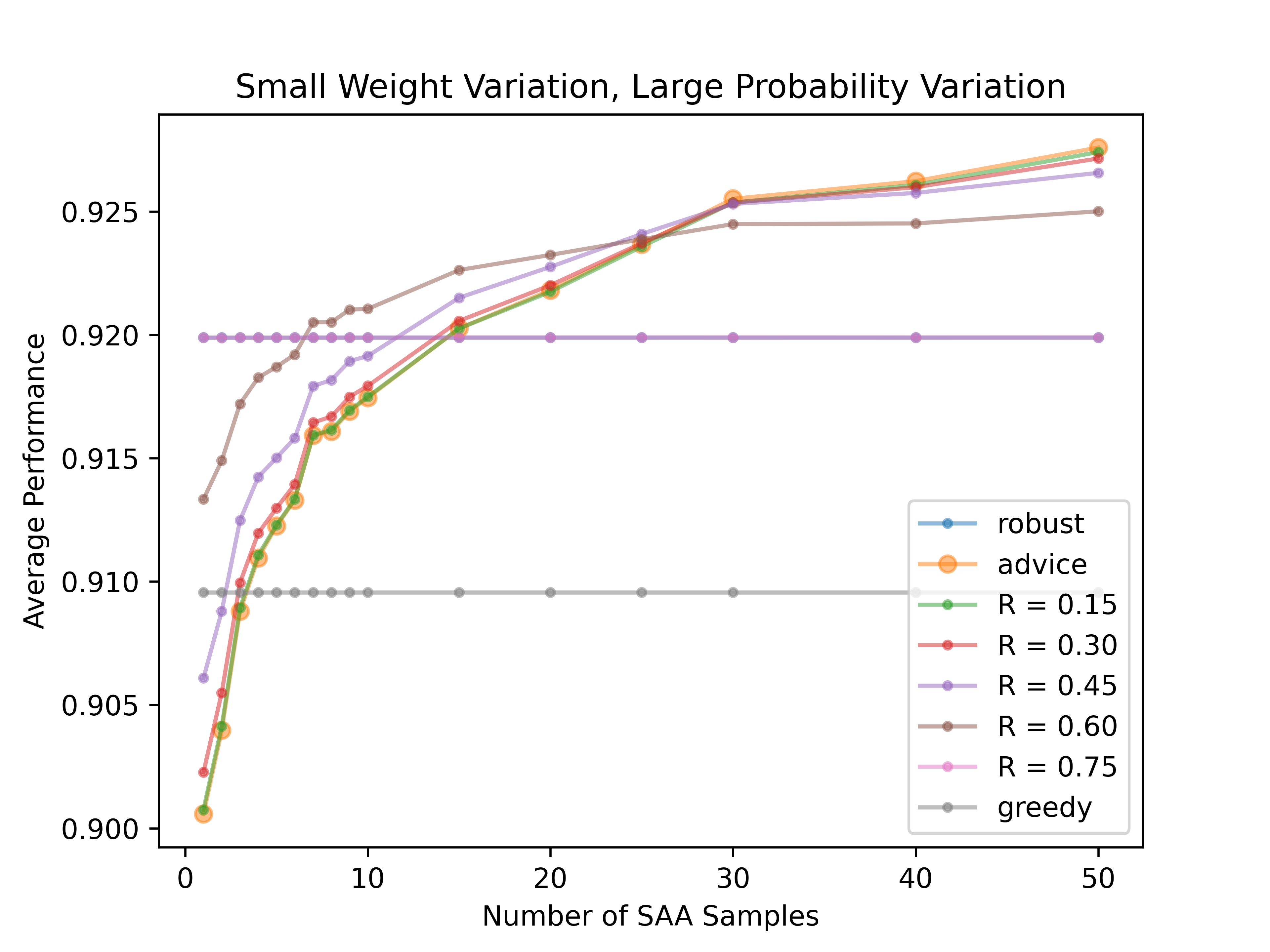}
    \end{minipage}%
    \begin{minipage}{.33\textwidth}
        \centering
        \includegraphics[width=\linewidth]{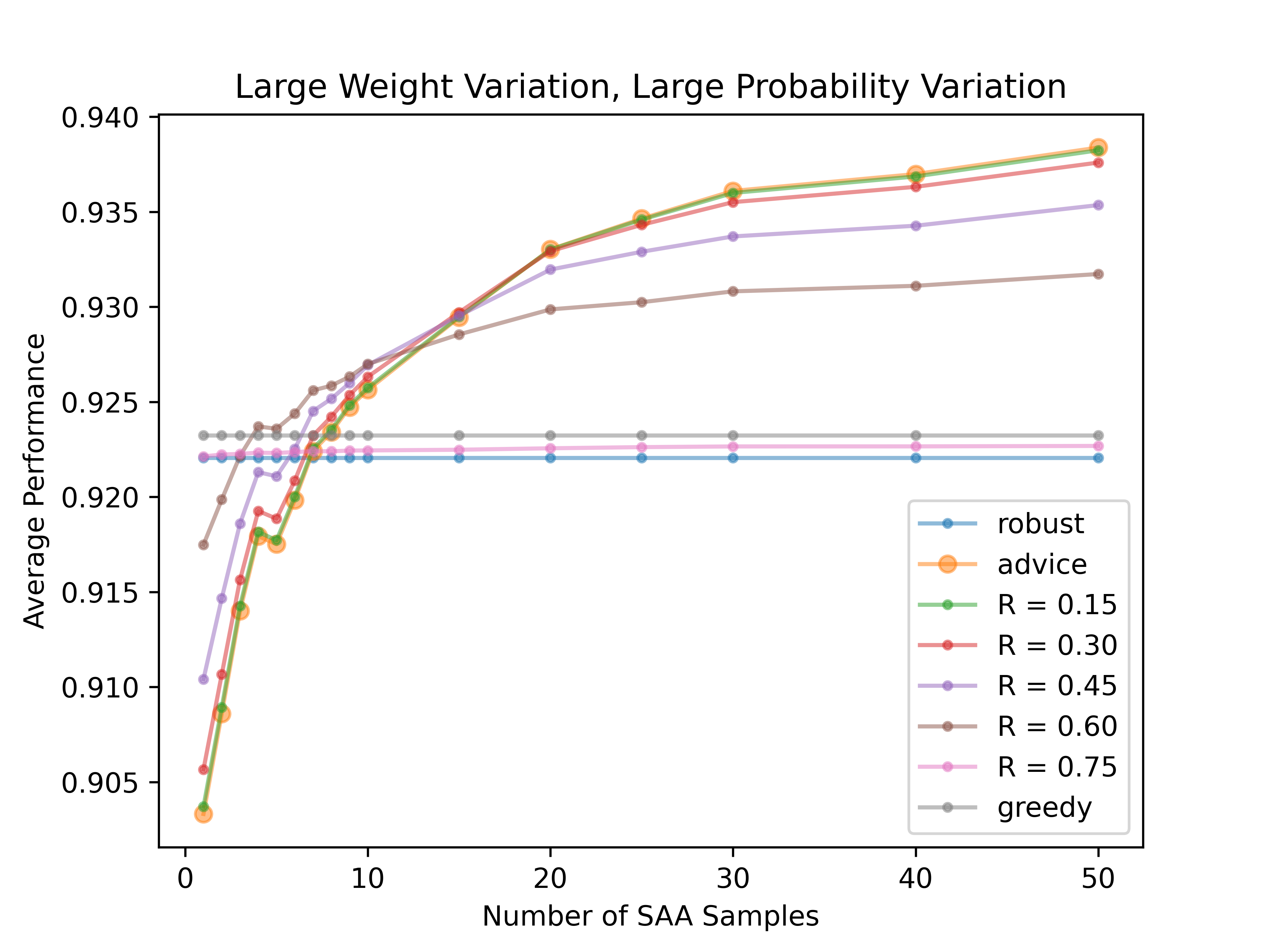}
    \end{minipage}    
    \caption{Synthetic experiments with no distribution shift. We plot the  performance of the algorithms across  three levels of variation in the weights (columns) and two levels of variation in the probabilities (rows).}
    \label{fig:synth}
\end{figure}

\begin{figure}[ht]
    \centering
    \begin{minipage}{.33\textwidth}
        \centering
        \includegraphics[width=\linewidth]{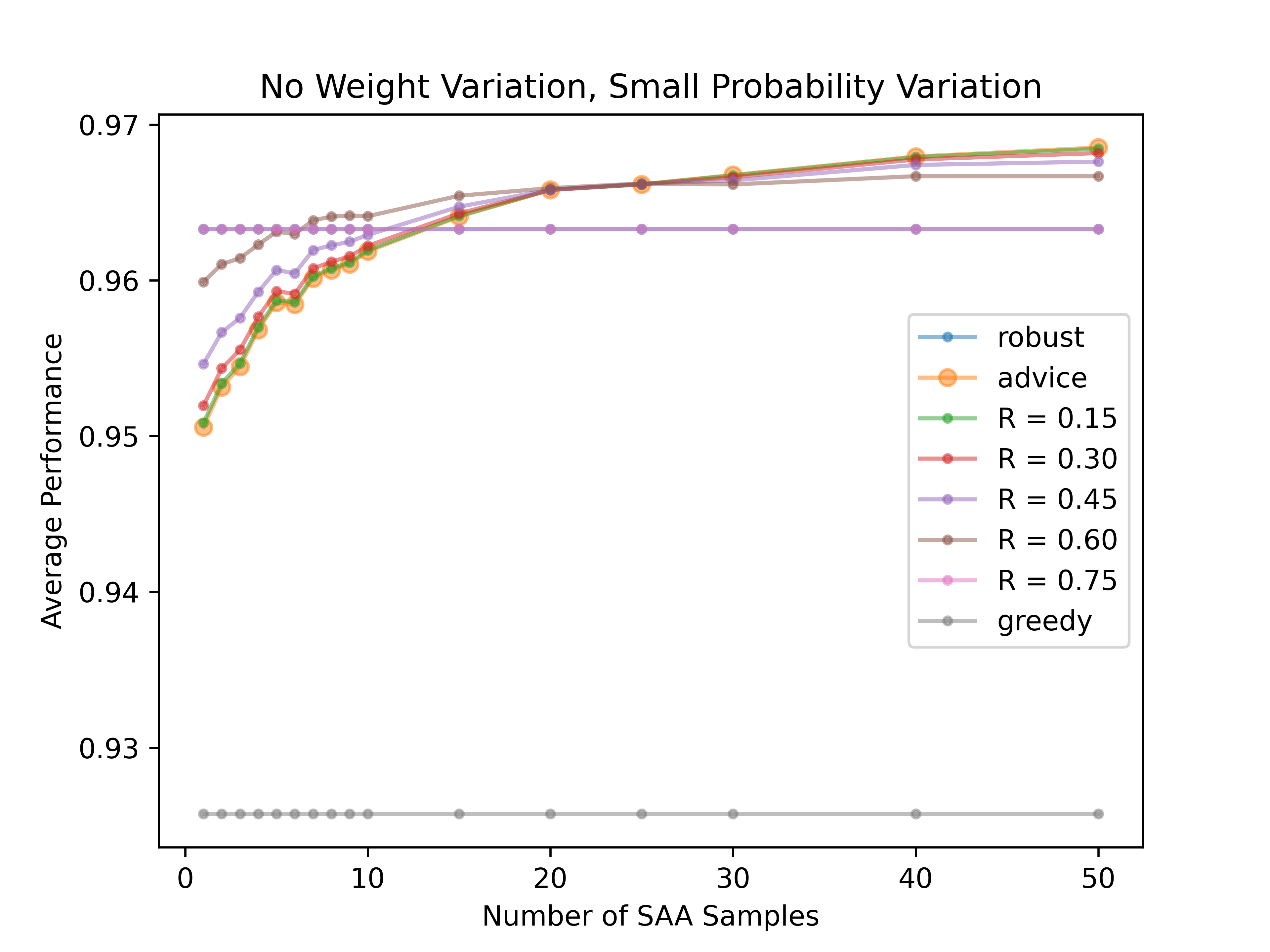}
    \end{minipage}%
    \begin{minipage}{.33\textwidth}
        \centering
        \includegraphics[width=\linewidth]{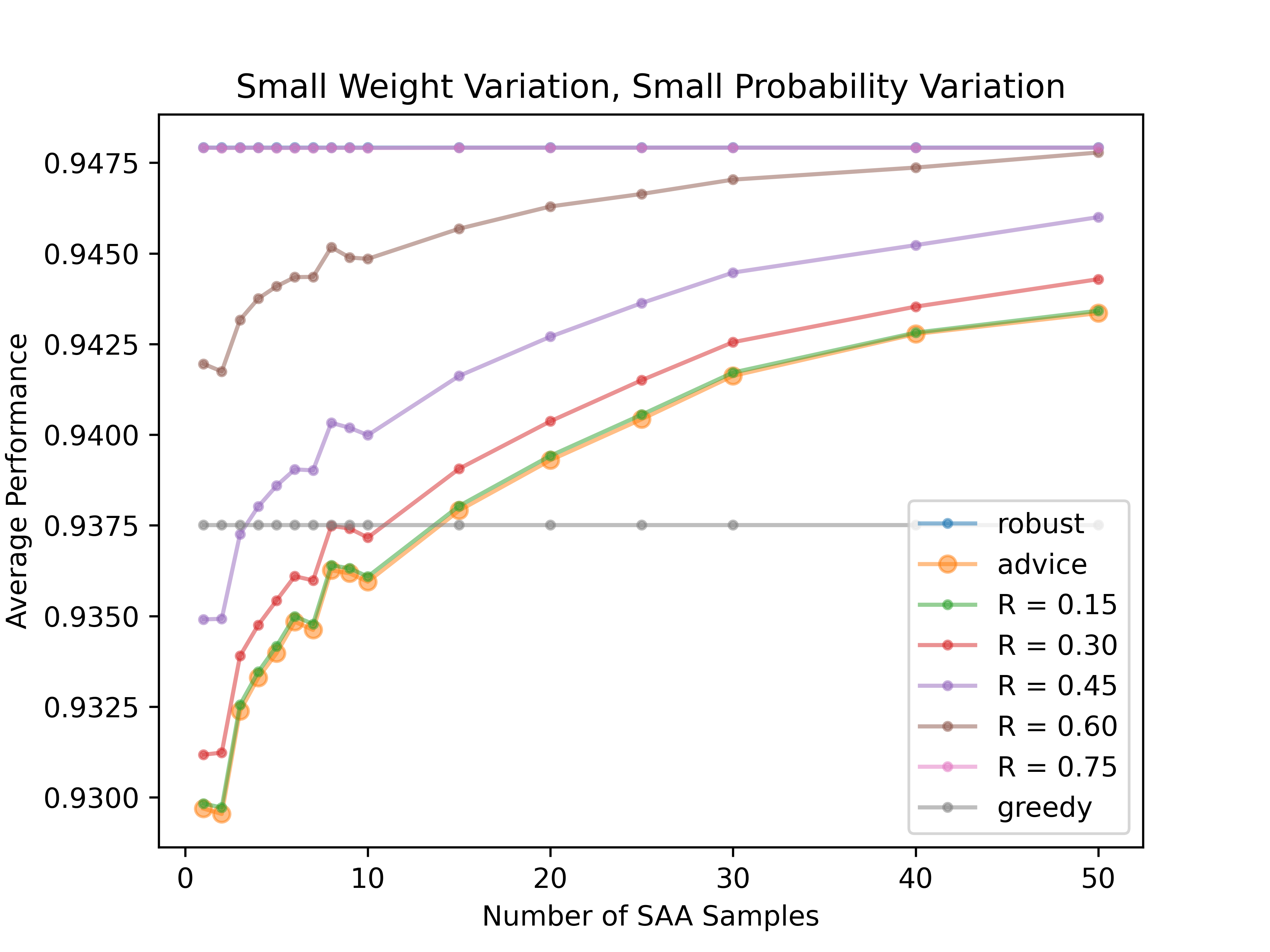}
    \end{minipage}%
    \begin{minipage}{.33\textwidth}
        \centering
        \includegraphics[width=\linewidth]{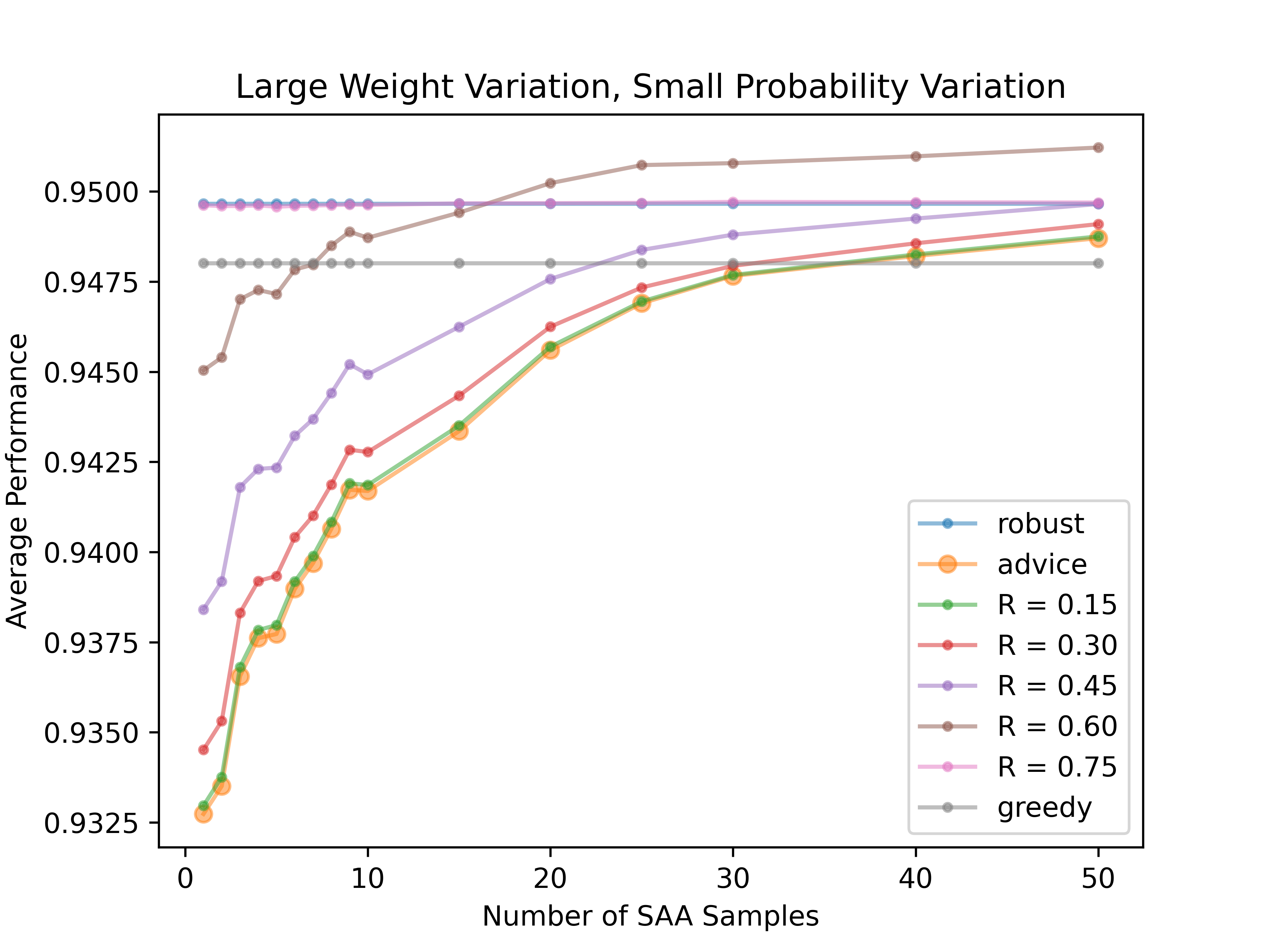}
    \end{minipage}

    \begin{minipage}{.33\textwidth}
        \centering
        \includegraphics[width=\linewidth]{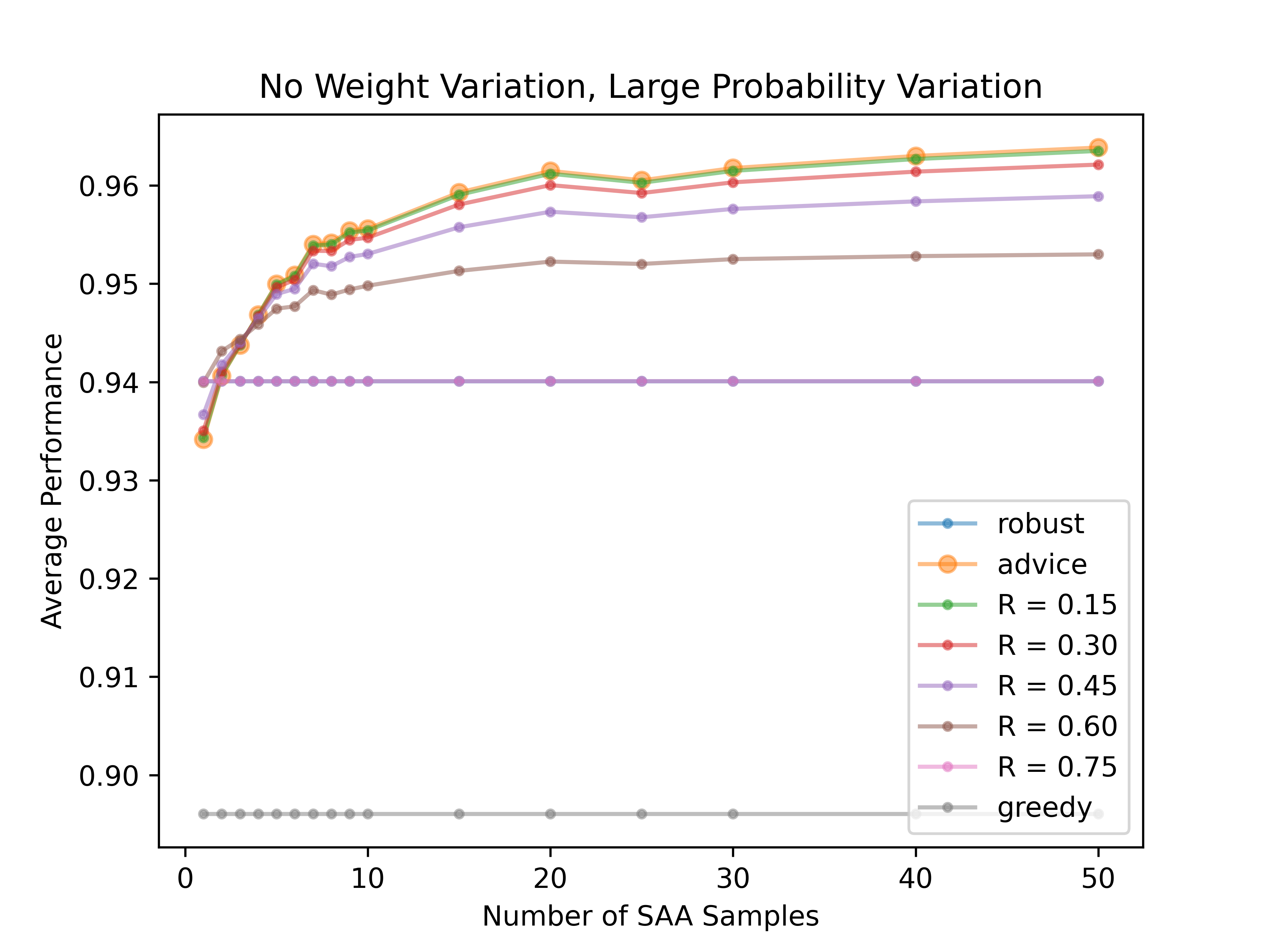}
    \end{minipage}%
    \begin{minipage}{.33\textwidth}
        \centering
        \includegraphics[width=\linewidth]{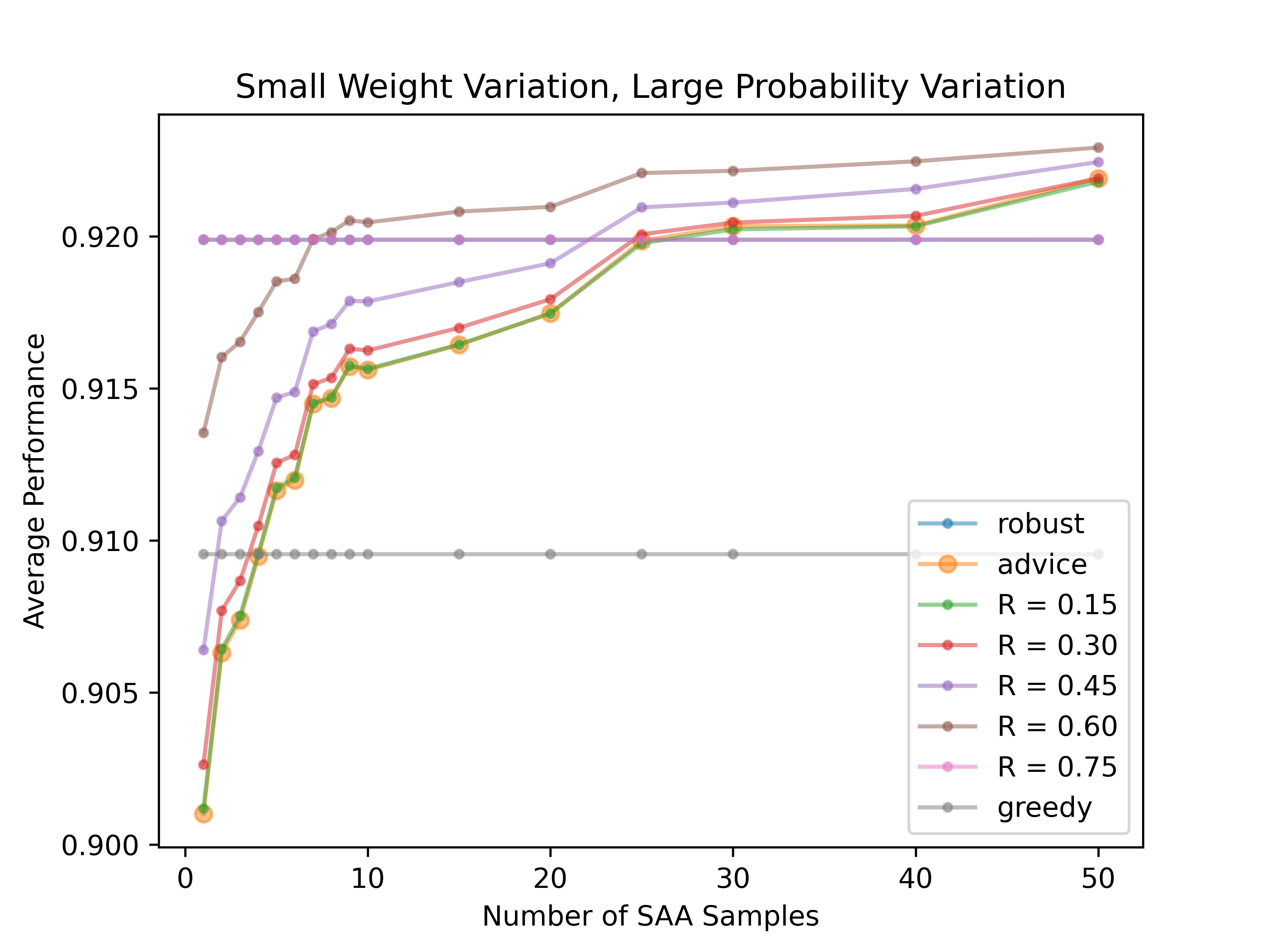}
    \end{minipage}%
    \begin{minipage}{.33\textwidth}
        \centering
        \includegraphics[width=\linewidth]{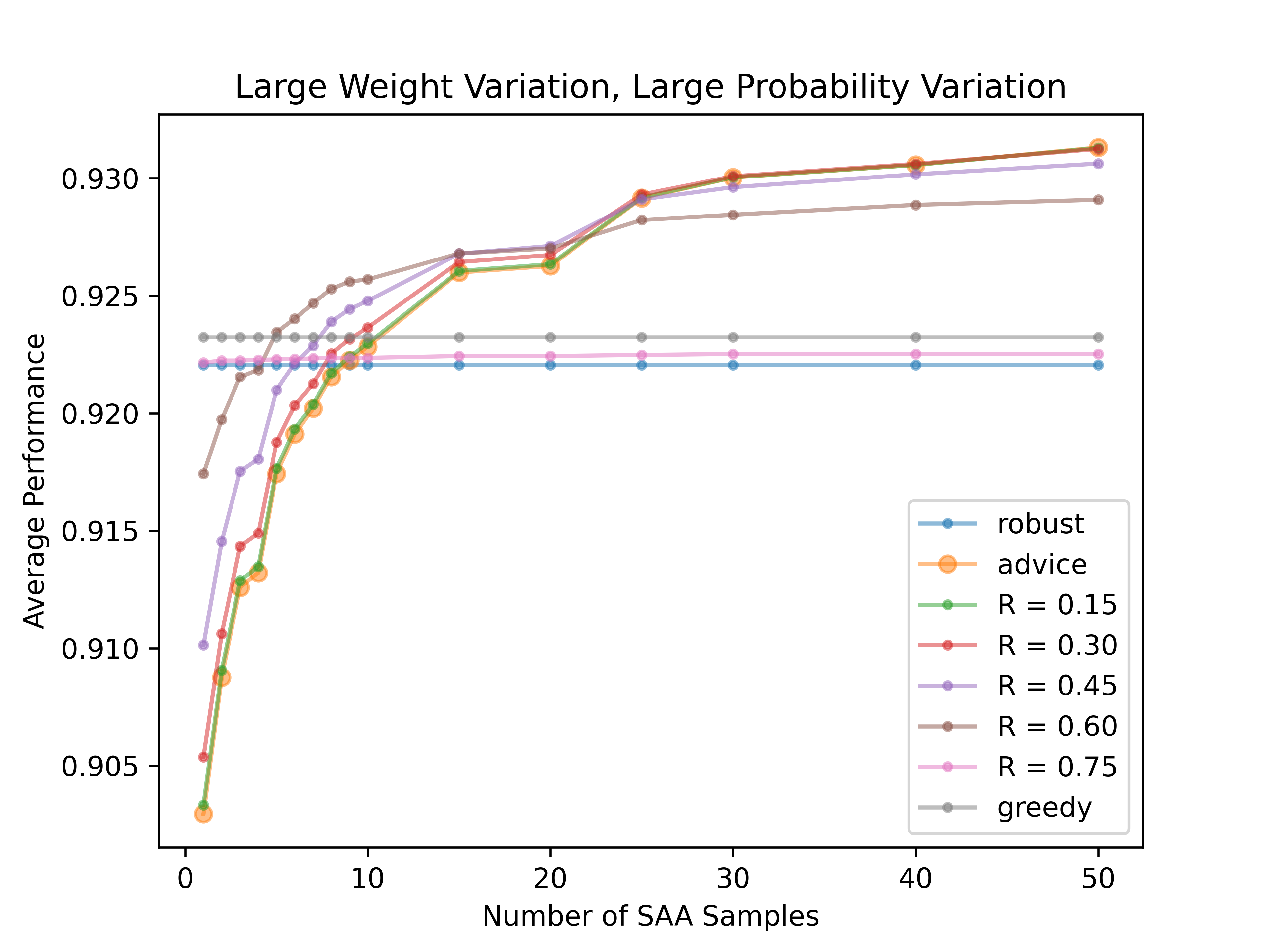}
    \end{minipage}    
    \caption{Synthetic experiments with  distribution shift. Here the SAA samples are generated using probabilities $\hat{p}_j = p_j + \mathrm{Unif}(-0.1, 0.1)$. We plot the  performance of the algorithms across  three levels of variation in the weights (columns) and two levels of variation in the probabilities (rows).}
    \label{fig:synth_perturbed}
\end{figure}


\subsection{Chicago Dataset}
\label{subsec:exp_chicago}
We provide a numerical evaluation of our algorithm on two-stage matching instances generated from a public ride-sharing dataset from \cite{chicago}.
%

\subsubsection{Dataset and instance generation.}
Since November 2018, the City of Chicago has required ride-sharing companies operating in the city to report anonymized trip data. We used the data from 2022, which can be accessed from the \href{https://data.cityofchicago.org/Transportation/Transportation-Network-Providers-Trips-2022/2tdj-ffvb}{City of Chicago Open Data Portal}. In our experiments, we focused on the trip data from 10:00 to 17:00 in the 30-day period from April 1, 2022 to April 30, 2022.\footnote{We checked that the results are robust to different choices of dates and times.}

Each trip has a start and end time that are rounded to the nearest 15 minutes for privacy reasons. The dataset also includes the pickup and dropoff locations for each trip, given as the centroid latitude and longitude coordinates of the corresponding census tract or community area.\footnote{Chicago is divided into approximately 800 census tracts, ranging in size from about 89,000 square feet to 8 square miles. For privacy reasons, the pickup and dropoff coordinates are given as the centroid of the corresponding census tract. If this would result in 2 or fewer unique trips in the same census tract and the same 15 minute window, the coordinates are given as the centroid of the (larger) community area instead. The average size of a community area is about 3 square miles. \href{http://dev.cityofchicago.org/open\%20data/data\%20portal/2019/04/12/tnp-taxi-privacy.html}{This page} gives more details on the approach to privacy in this dataset.} We restricted the dataset to trips in a rectangular region encompassing the downtown area of Chicago, with latitude between 41.8 and 42.0 and longitude between $-87.7$ and $-87.6$.
The map of this region is depicted in \Cref{fig:chicago}.



\begin{figure}[ht]
        \centering
        \begin{subfigure}[b]{0.475\textwidth}
            \centering
            \includegraphics[scale=0.45]{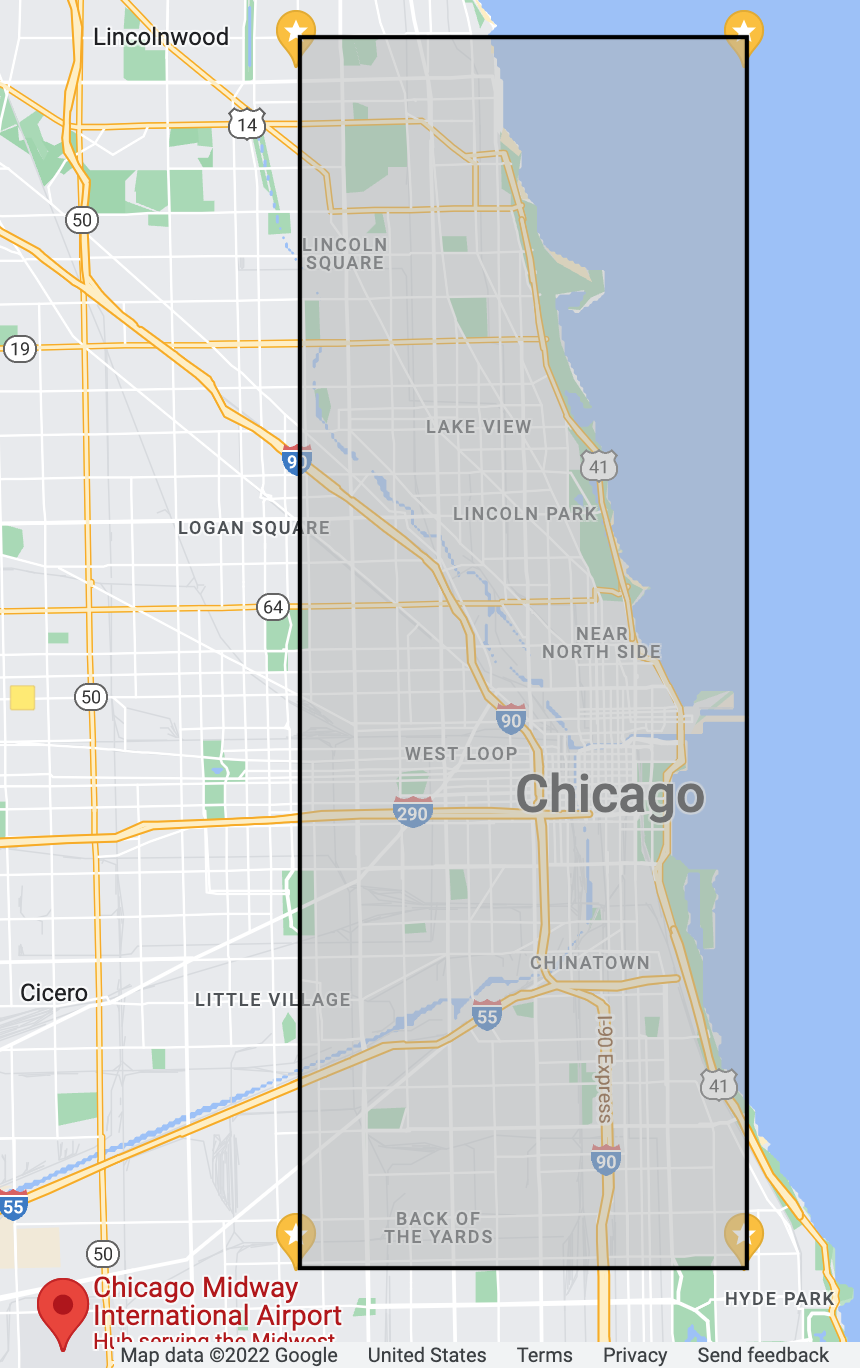}
    \caption{A map of Chicago. The trips used in our experiment were drawn from the shaded rectangle.}
    \label{fig:chicago}
        \end{subfigure}
        \hfill
        \begin{subfigure}[b]{0.475\textwidth}  
            \centering 
            \includegraphics[scale=0.05]{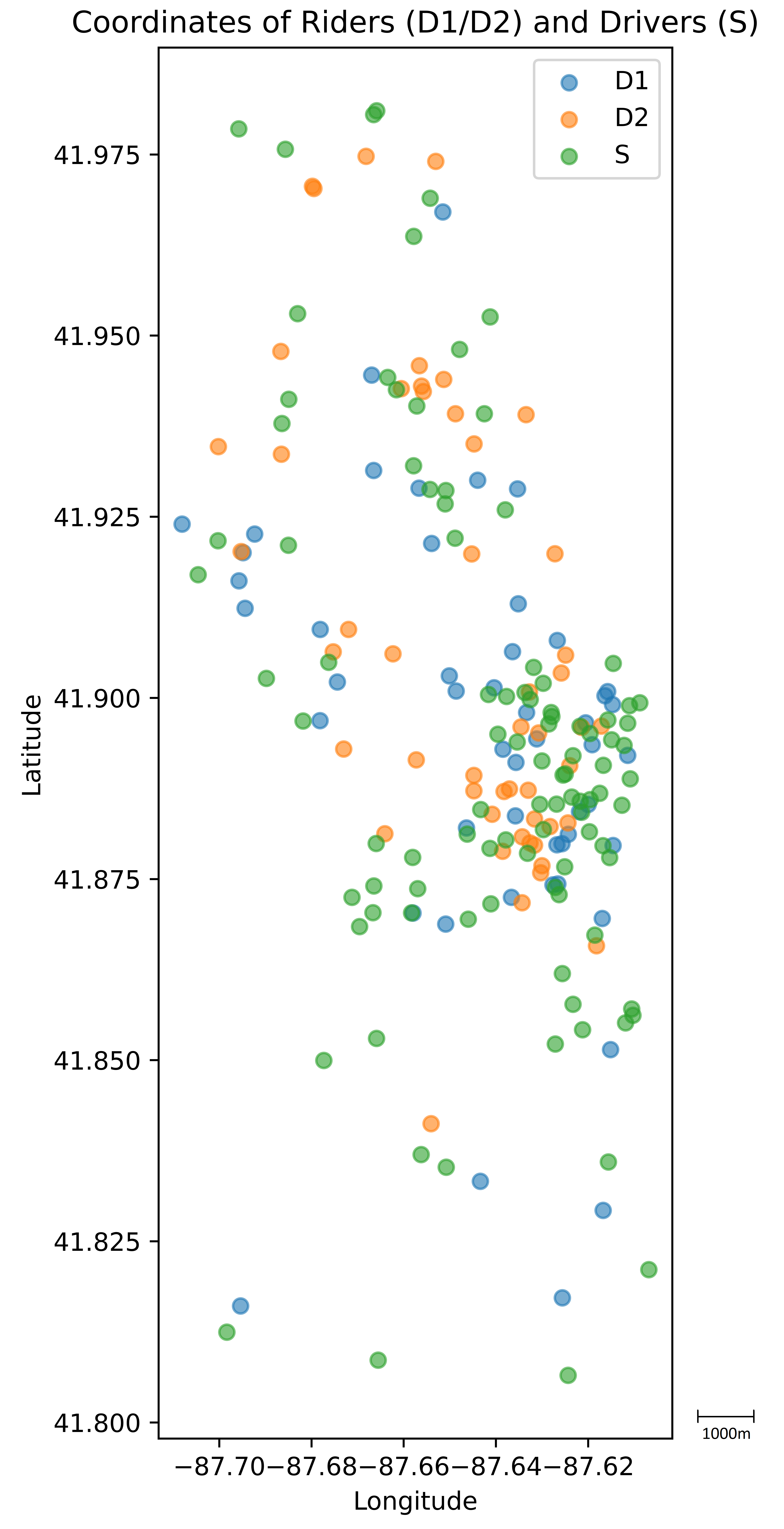}
    \caption{The locations of the supply and demand nodes in one Monte-Carlo replication. 
    }
    \label{fig:coords}
        \end{subfigure}
        \caption
        {\small An illustration of the geographic locations of the supply and demand nodes.}
        \label{fig:location}
\end{figure}
To generate two-stage matching instances from the dataset, we mimic the approach of \cite{feng2021two}. In each replication, we first sample a day $\texttt{DAY} \in [30]$ uniformly at random from the 30 possible days in April 2022. 

\noindent\textbf{First-stage demand $D_1$.} Generate a uniformly random time $T$ between 10:00 and 16:45, rounded to the nearest 15 minutes. Fetch all ride requests on day $\texttt{DAY}$ whose \texttt{Trip Start Timestamp} equals $T$. (Recall that all timestamps in the dataset are rounded to the nearest 15 minutes.) We then take $D_1$ to be a uniform random subset of these requests of size 50. 

\noindent\textbf{Second-stage demand $D_2$.}
Fetch all ride requests on day $\texttt{DAY}$ whose \texttt{Trip Start Timestamp} equals $T + 15$ (minutes). Then take $D_2$ to be a uniformly random subset of these requests of size 50. 

\noindent\textbf{Supply vertices $S$.} Take $S$ be the set of all drivers whose \texttt{Trip End Timestamp} equals $T - 15$, uniformly sampled to have $\abs{S} = 100$.
We assume that drivers whose trip ends during $T$ or $T+15$ are ineligible to serve either first-stage or second-stage demand.

\noindent\textbf{Edge set $E$.} Recall that the pickup and dropoff locations are reported as the centroid of the nearest census tract.
We first randomly perturb each location to a uniformly random point within a ball of radius 1000m, to capture the dispersion of customers with the same census tract.  Then,
we add an edge between a rider in $D$ and a driver in $S$ if the distance between the rider's pickup location and the driver's dropoff location is less than 1000m. This distance is calculated using the longitude and latitude coordinates. Figure \ref{fig:coords} illustrates the locations of the supply and demand nodes in one Monte-Carlo replication.


\noindent\textbf{Supply weights.} We experiment with both unweighted graphs, and weighted graphs, where the weights for each supply node are generated independently at random from a probability distribution. We test the following weight distributions: $\abs{N(0,1)}$\footnote{Without loss of generality, the   variance of the Normal distribution is taken to be 1; scaling all vertex weights by a constant factor simply scales the value of each algorithm by the same factor.} and $\text{Unif}(1, k)$ for $k \in \{2, 4\}$.

\noindent\textbf{Advice.} We consider two different ways to generate the advice. In \Cref{subsubsec:advice_corrupt}, we take the advice  to be a corrupted version of the optimal first-stage matching in hindsight, and we compare the performance of the algorithms as the corruption level varies. In \Cref{subsubsec:advice_d1}, we take the advice to be the optimal first-stage matching given a prediction of the second-stage demand.

\subsubsection{Advice from corrupted optimal matching.}
\label{subsubsec:advice_corrupt}

\begin{figure*}[ht]
        \centering
        \begin{subfigure}[b]{0.475\textwidth}
            \centering
            \includegraphics[width=\textwidth]{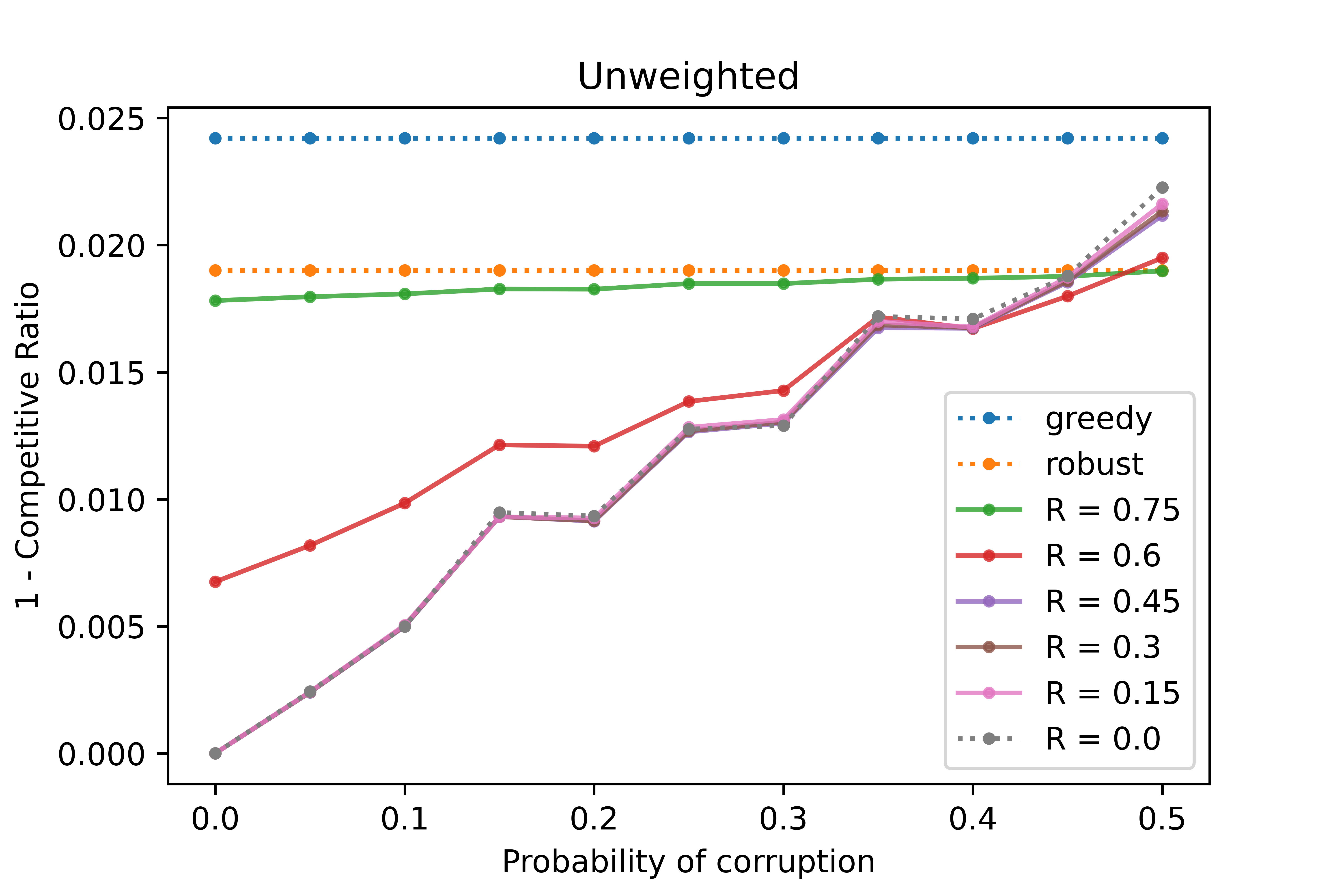}
            \caption%
            {{\small Unweighted ($w_j = 1$)}}    
            \label{fig:unweighted}
        \end{subfigure}
        \hfill
        \begin{subfigure}[b]{0.475\textwidth}  
            \centering 
            \includegraphics[width=\textwidth]{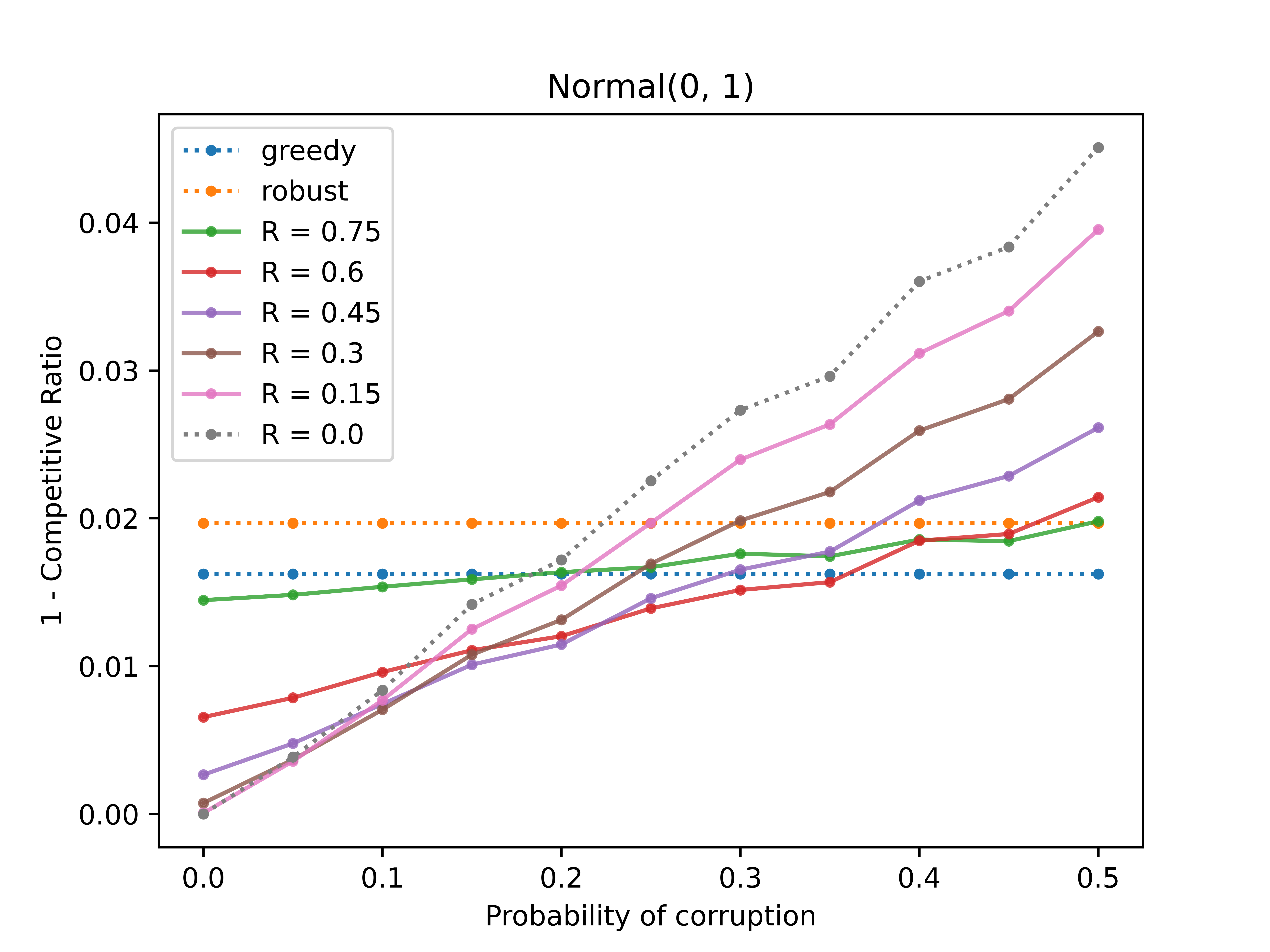}
            \caption[]%
            {\small $w_j \sim \abs{N(0,1)}$}    
            \label{fig:normal}
        \end{subfigure}
        \vskip\baselineskip
        \begin{subfigure}[b]{0.475\textwidth}   
            \centering 
            \includegraphics[width=\textwidth]{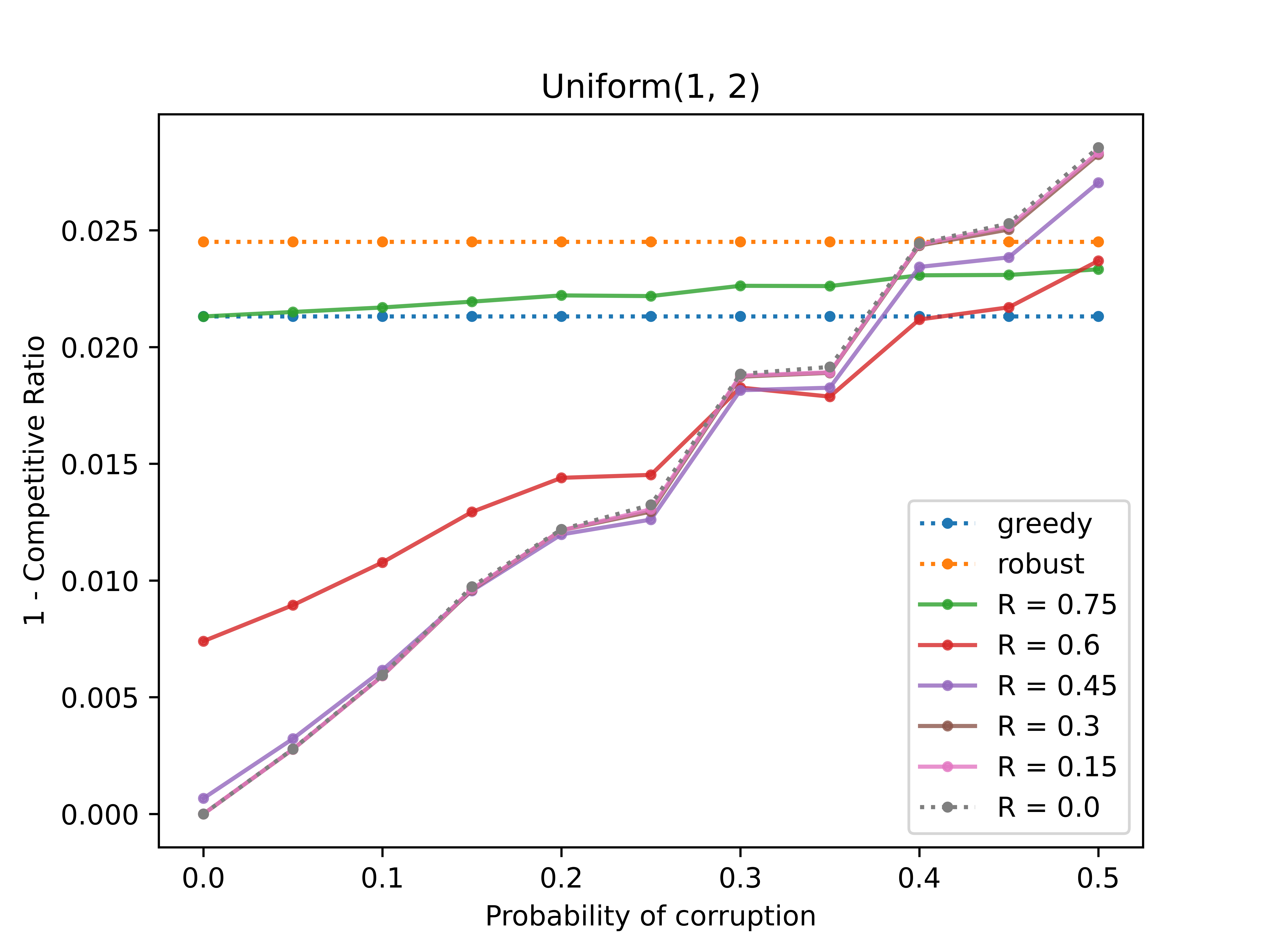}
            \caption[]%
            {{\small $w_j \sim \text{Uniform}(1,2)$}}    
            \label{fig:unif_2}
        \end{subfigure}
        \hfill
        \begin{subfigure}[b]{0.475\textwidth}   
            \centering 
            \includegraphics[width=\textwidth]{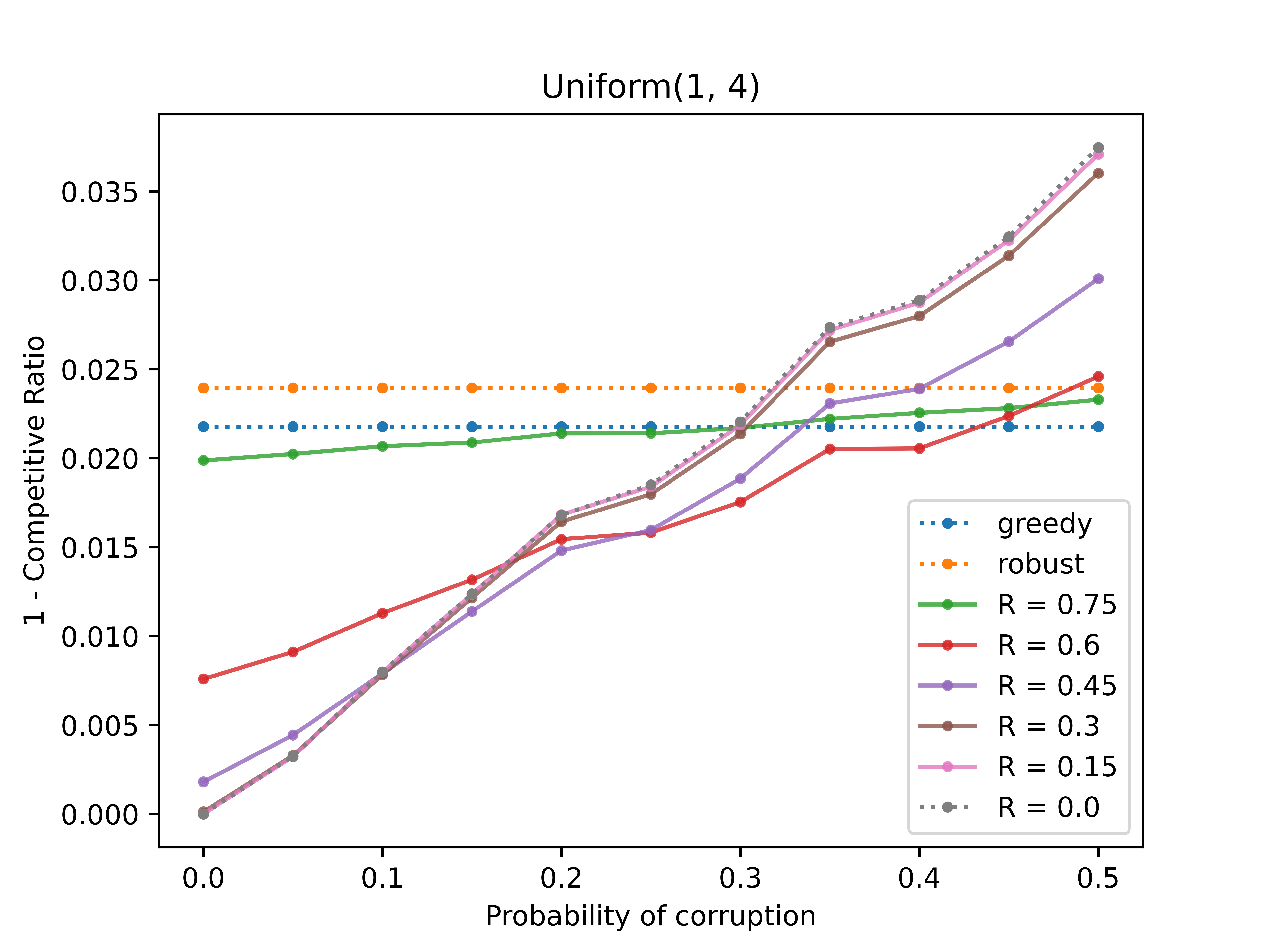}
            \caption[]%
            {{\small $w_j \sim \text{Uniform}(1,4)$}}    
            \label{fig:unif_4}
        \end{subfigure}
        \caption
        {\small Plots of the performance of the algorithms on various weight distributions.  The $y$-axis is 1 minus the average competitive ratio over the replications, so \textit{lower is better}. Robust refers to the $\frac34$-competitive algorithm of \cite{feng2021two}, and Advice refers to the algorithm that directly follows the advice ($R=0$).  The algorithms with $R=0.15,\ldots,0.75$ are all based on new developments from this paper.}
        \label{fig:crs}
    \end{figure*}

The best possible advice is clearly the first-stage edges $M_1^*$ that are chosen by the optimal matching in hindsight. To generate advice of varying quality, we use a "corrupted" suggested matching, obtained by taking each edge $(i,j)$ in the optimal suggested matching $M_1^*$, and with probability $p$, replacing it with an edge $(i, j')$ where $j'\neq j$ is a random supply node not already in $M^*_1$. (If no such $j'$ exists then we delete edge $(i,j)$ from the suggested matching.) We test 11 values of $p$ ranging from 0 to 0.5 inclusive, in increments of $0.05$. A higher value of $p$ represents a higher degree of corruption, and hence leads to a worse advice; adjusting $p$ allows us to test the sensitivity of our algorithms to the quality of the advice.

We tested our algorithm for $R \in \{0,0.15,0.3,0.45,0.6,0.75\}$, and compared against  the Greedy algorithm (which always chooses the max-weight matching in the first stage), and the ``Robust'' algorithm of \cite{feng2021two}. We ran 100 replications.  In each replication, we 1) generate the base graph $G$ of demand/supply nodes $D_1$, $D_2$, $S$, and edges $E$ (as described above), and 2) test each of the above algorithms on four sets of vertex weights on $G$ (one for each of the weight distributions described previously). 

The results are shown in \Cref{fig:crs}. Each plot shows a different weight distribution.
In each plot, the points represent the average competitive ratio obtained by a given algorithm over the 100 replications. 
We make the following observations:
\begin{itemize}
\item Robust is the better advice-agnostic algorithm in the unweighted setting, while Greedy is the better advice-agnostic algorithm with varied weights.
\item Our algorithms are advice-aware, and their performance improves with the quality of the advice.  As $R$ increases, our algorithms' decisions are less affected by the advice, so their performance becomes less sensitive to changes in the advice quality.
\item With $R = 0.75$, our algorithm dominates Robust, even when as many as half of the edges in the suggested matching are corrupted. This suggests there can be value in incorporating advice into the algorithm's decisions, even if the advice is quite bad.
\item When the probability of corruption is $p=0$, one cannot do better than always following the advice (setting $R=0$).  On the other extreme, $p=1/2$ with varied weights makes Greedy perform best.  However, in all intermediate regimes, algorithms based on new developments in this paper (with $R=0.15,\ldots,0.75$) perform best.
\item It is rarely optimal to set $R = 0.0$ or $R = 0.75$. In particular, setting $R = 0.15$ is only marginally worse than $R = 0.0$ when the probability of corruption is 0, and is noticeably better when the probability gets large. Similarly, setting $R = 0.6$ is is only marginally worse than $R = 0.75$ when the probability of corruption is $0.5$, and is noticeably better when this probability is small. 
\item The previous observation agrees with the shape of the concave tradeoff curve from our theory: When $R \approx 0$, one can gain a large improvement in robustness by sacrificing a bit of consistency because the curve is flat at $R = 0$. Similarly, when $R \approx 0.75$, one can gain a large improvement in consistency by sacrificing a bit of robustness, because the curve is steep at $R =0.75$.
\end{itemize}

\subsubsection{Advice from prediction of second-stage demand.}
\label{subsubsec:advice_d1}

\begin{figure}[h]
    \centering
    \begin{minipage}{0.45\textwidth}
        \centering
        \includegraphics[width=\textwidth]{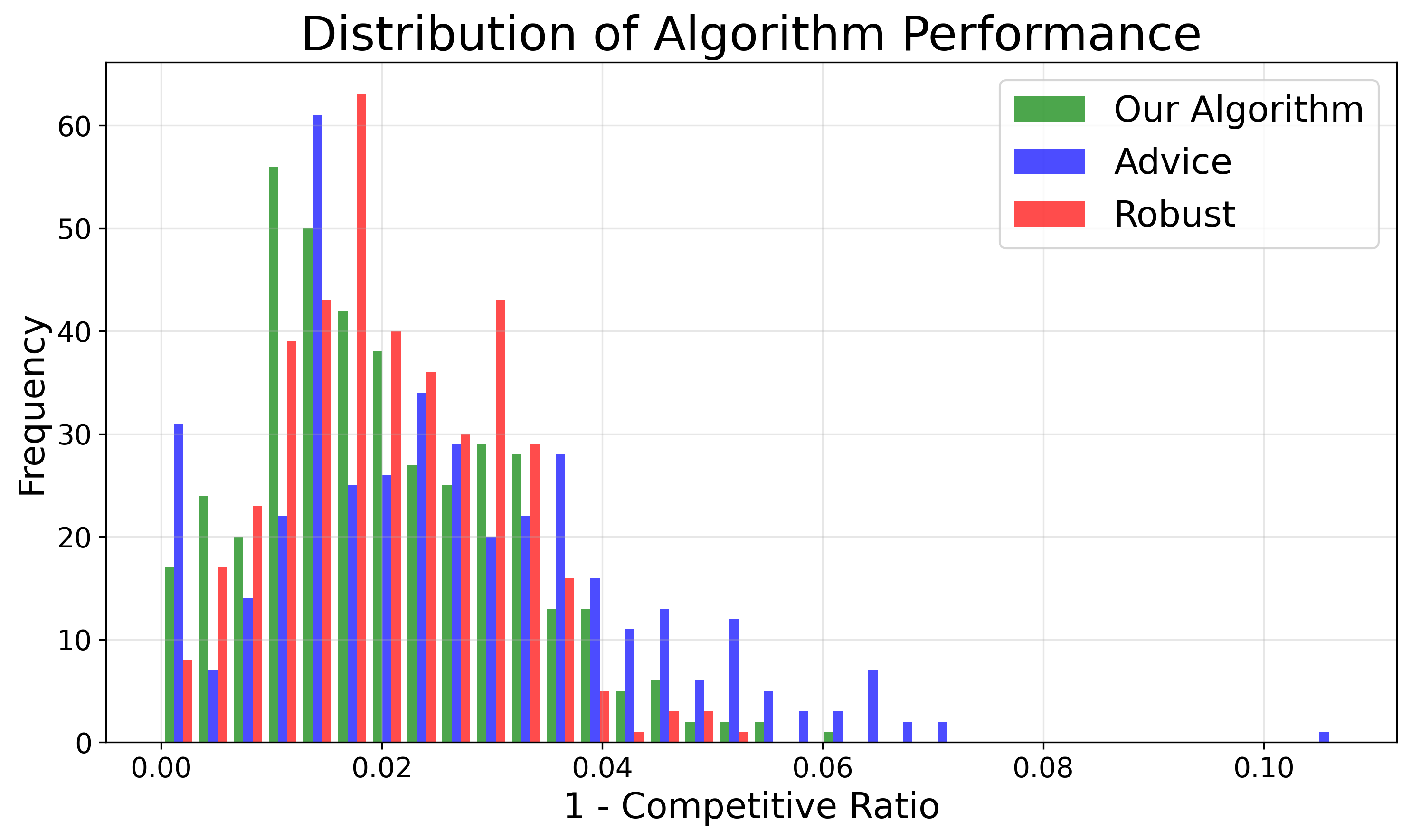}
    \end{minipage}\hfill
    \begin{minipage}{0.45\textwidth}
        \centering
        \includegraphics[width=\textwidth]{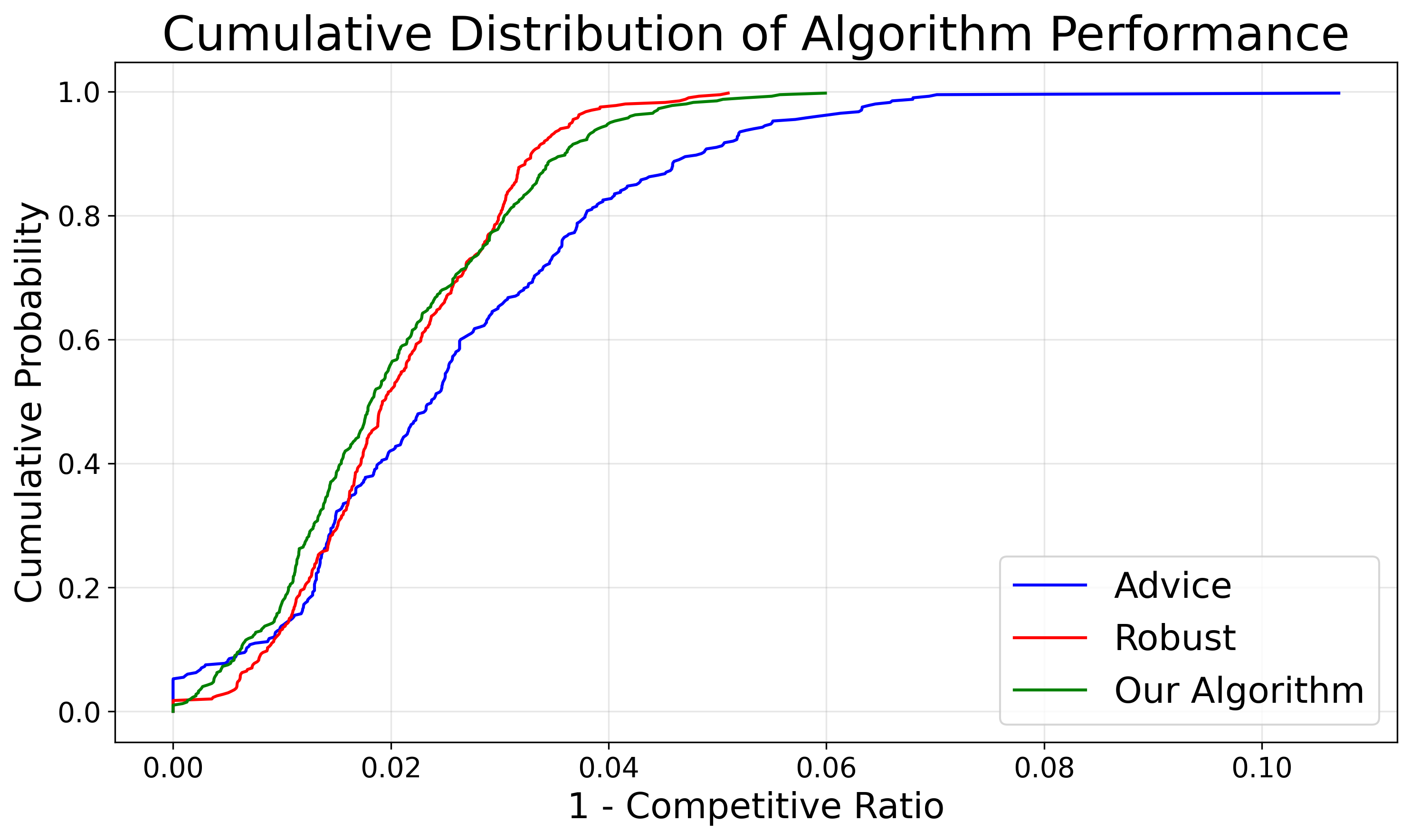}

    \end{minipage}
    \caption{Histograms of the performance of Robust, Advice, and our algorithm with $R = 0.6$ aggregated across the four weight classes. On the left is a raw histogram of the data. On the right is a cumulative histogram, where the $y$-value on a curve represents the percentage of all instances where the given algorithm achieved at least that level of performance.}
 \label{fig:new_chicago}
\end{figure}

We use the same setup as before, except now the Advice is now generated as follows:
\begin{enumerate}
    \item Use the first-stage graph as a prediction for the second-stage graph. In other words, we are using the ride requests at time $T$ as a prediction for the ride requests at time $T + 15$.  
    \item Take Advice to be the optimal first-stage matching under this prediction. 
\end{enumerate}
We randomly sampled 100 two-stage matching instances from the data as before. For each instance, we tested the Robust algorithm of \cite{feng2021two}, the algorithm which always follows the Advice, and our algorithm with $R = 0.6$ (which appears to generally be the best parameter value). We also vary the weights according to the four previously defined weight classes. 

The  results are shown in \Cref{fig:new_chicago}.Unlike \Cref{subsec:exp_synth}, here we show the histogram of performances over instances to demonstrate the robustness that ALPS brings. Each point in the histogram represents the performance of an algorithm on a specific instance with a particular weight class, resulting in 400 points per algorithm: one for each instance and weight class combination. The average error (1 - Competitive Ratio) and standard deviation of the three algorithms are as follows:
\begin{itemize}
    \item \textbf{Advice}: Average Error = 2.55\%, Standard Deviation = 1.62\%.
    \item \textbf{Robust}: Average Error = 2.06\%, Standard Deviation = 0.97\%.
    \item \textbf{Our Algorithm}: Average Error = 2.02\%, Standard Deviation = 1.15\%.
\end{itemize}
We make the following observations:
\begin{itemize}
    \item While Advice has the most upside in the best case, it also has the most downside in the worst case. In other words, blindly following the advice works well if the prediction happens to be good, but it runs the risk of performing very poorly if the prediction is inaccurate.
    \item The Robust algorithm of \cite{feng2021two}, on the other hand, has the smallest tail, meaning its worst-case performance surpasses that of both Advice and our algorithm.
    \item The advantage of our algorithm lies in its balanced performance: it is only marginally worse than Robust in the worst case, yet it still captures much of the upside of Advice in favorable cases. Our algorithm also has the best average-case performance, although only by a bit, consistent with the finding from \Cref{subsec:exp_synth} on the synthetic experiments.
\end{itemize}

\noindent\textbf{Summary.}
All in all, our experiments indicate that incorporating advice can be beneficial---in particular, there is utility in being able to control the degree to which the advice is followed. Our algorithm gives a principled way to do this for two-stage matching using the parameter $R$. We do not address the question of how one should choose $R$ given data from practice, and we leave this as a subject for future research.




\bibliographystyle{informs2014} 
\bibliography{bibliography} 




\APPENDICES
\crefalias{section}{appendix}
\crefalias{subsection}{appendix}

\end{document}